\documentclass[12pt, a4paper]{report}

\usepackage{titlesec}

\titleformat{\chapter}[block]
{\normalfont\fontsize{18}{22}\bfseries\MakeUppercase\centering}
{\thechapter.}
{1em}
{}
\titlespacing*{\chapter}{0pt}{-20pt}{40pt}

\titleformat{\section}
{\normalfont\fontsize{16}{20}\bfseries\MakeUppercase}
{\thesection}
{1em}
{}

\titleformat{\subsection}
{\normalfont\fontsize{14}{18}\bfseries}
{\thesubsection}
{1em}
{}

\titleformat{\subsubsection}
{\normalfont\fontsize{12}{16}\rmfamily}
{\thesubsubsection}
{1em}
{}

\usepackage[utf8]{inputenc}
\usepackage{geometry}
\usepackage{amsmath}
\usepackage{amssymb}
\usepackage{amsthm}

\usepackage{physics}
\usepackage{longtable}
\usepackage{booktabs} 
\usepackage{amsmath, amssymb, amsthm, braket, bbold}
\usepackage{graphicx}
\usepackage{xspace}

\usepackage[titles]{tocloft} 
\usepackage[caption=false]{subfig}
\usepackage{float}
\usepackage{placeins}

\usepackage{setspace}
\newcommand{\QSL}{\textrm{QSL}\xspace}
\newcommand{\MT}{\textrm{MT}\xspace}
\newcommand{\DL}{\textrm{DL}\xspace}

\usepackage[backend=bibtex, style=numeric-comp, sorting=none]{biblatex}

\begin{document}
	

	\thispagestyle{empty} 

	\begin{abstract}

	\noindent\textbf{Abstract of thesis entitled}\\[0.5em]
	\textit{From Fundamental Dynamics to Applied Cryptography:}\\
	\textit{Studies on the Quantum Speed Limit and Fully Passive Quantum Key Distribution}\\[0.75em]
	\noindent\textbf{submitted by} \textit{Jinjie Li}\\
	\noindent\textbf{for the degree of} \textit{Doctor of Philosophy}\\
	\noindent\textbf{at} \textit{The University of Hong Kong}\\
	\noindent\textbf{in} \textit{May 2026}\\[1.5em]
	This thesis studies two distinct frontiers of quantum information processing: the
	fundamental physical limits of dynamical evolution and the practical realization of secure
	quantum communication networks.
	
	In the area of applied cryptography, I address the critical challenge of implementation
	security in quantum key distribution (QKD). While, in principle, QKD offers
	unconditional security, real-world devices often contain imperfections that open
	``side-channels'' for eavesdroppers. To close these loopholes, I propose and analyze
	fully passive sources that eliminate the need for active optical modulators, thereby
	removing a major class of source-side vulnerabilities. First, I develop a fully passive
	measurement-device-independent (MDI) QKD protocol. I construct a complete
	security framework involving a novel decoy-state analysis, and I show that secure
	communication is achievable over about $140\,\mathrm{km}$ of standard optical fiber.
	Second, I extend this passive paradigm to the multi-user setting by introducing a fully
	passive conference key agreement (CKA) protocol based on twin-field interference.
	Using high-dimensional numerical integration and linear-programming techniques,
	I show that secure conference keys can be established among four users over channel
	losses up to roughly $28\,\mathrm{dB}$, corresponding to more than $130\,\mathrm{km}$ of standard
	optical fiber, demonstrating that the protocol performs well even in such high-loss
	regimes.

	In the area of fundamental dynamics, I study the quantum speed limit (QSL),
	which sets the ultimate bound on how fast a quantum system can evolve. I identify
	limitations in existing bounds regarding their tightness and universality. To resolve this,
	I derive a new family of QSLs based on representation-dependent weighted
	$\ell_{w}^{p}$-seminorms. This framework is universal, applicable to both closed and open
	system dynamics, as well as both time-dependent and time-independent evolutions, and
	extends beyond density matrices to general linear operators. I demonstrate that these
	new bounds are computationally efficient to optimize and provide good performance in
	various scenarios, including spontaneous emission and high-fidelity quantum gates.
		
	\end{abstract}

	
	\begin{titlepage}
		\thispagestyle{empty}
		\centering

		{\LARGE\textbf{From Fundamental Dynamics to Applied Cryptography:}}\\[0.08cm]
		{\LARGE\textbf{Studies on the Quantum Speed Limit and Fully Passive Quantum Key Distribution}}\\[2cm]
		
		{\large\textbf{by}}\\[2cm]
		{\Large\textbf{Jinjie Li}}\\[9cm]

	{\large A thesis submitted in partial fulfilment of the requirements\\
		for the degree of}\\[0.5cm]
	{Doctor of Philosophy}\\[0.5cm]
	{ at The University of Hong Kong}\\[1.5cm]
		
		
		{\large May 2026}
		
	\end{titlepage}

\cleardoublepage
\pagenumbering{roman} 
\setcounter{page}{1}  
\phantomsection       
\addcontentsline{toc}{chapter}{Declaration} 

\chapter*{Declaration}
\noindent I declare that this thesis represents my own work, except where due acknowledgement is made, and that it has not been previously included in a thesis, dissertation or report submitted to this University or to any other institution for a degree, diploma or other qualifications.

\vspace{2cm}

\noindent\textit{Signed:} \rule{6cm}{0.4pt}\\[1em]
\noindent\textit{Name:} Jinjie Li

\cleardoublepage
\phantomsection
\addcontentsline{toc}{chapter}{Acknowledgements} 

\chapter*{Acknowledgements}

This thesis would not have been possible without the immense support of my family over the last four years. The decision to move to Hong Kong was sudden and, frankly, I was doubtful at the beginning. Yet, my wife and my parents supported me doubtlessly, standing by me even as we navigated the uncertainties of the pandemic.

I must pay a special tribute to my wife. Supporting my move to a new city was a brave decision on her part, and her presence was my anchor. The challenges during these years of study were significant, and without my family’s constant encouragement, I could not have persevered.

I also wish to extend my sincere thanks to my supervisors, Prof. H.F. Chau and Prof. H.K. Lo. Beyond their academic guidance, they provided me with an invaluable working environment. I am deeply grateful for the flexibility and lack of rigid constraints they afforded me. Their patience and support, particularly while I dealt with health issues, were crucial.

Lastly, I want to thank my family for understanding my need to rest and explore new directions after graduation. While this may diverge from the traditional path, their support allows me to embrace this next chapter with confidence. I understand this is a vote that places faith in me.

	
	\cleardoublepage
	\phantomsection
	\addcontentsline{toc}{chapter}{Table of Contents} 
	\tableofcontents

	\clearpage
	\pagenumbering{arabic}
	\setcounter{page}{1}
	\chapter{INTRODUCTION}
	\label{ch:introduction}
	
\section{The Dawn of a New Information Paradigm}

Information science in the twentieth century, founded upon Shannon’s classical theory~\cite{Shannon1948} and powered by digital computation, transformed human society. For decades, technological progress followed Moore’s Law~\cite{Moore1965}, which predicted an exponential increase in transistor density. Yet as devices shrink toward the atomic scale, quantum effects, such as tunnelling, charge discreteness, and excessive heat density, signal the approaching end of this classical roadmap~\cite{Dennard1974,Pop2010,Keyes2005}.

In the 1980s, Feynman and others proposed that these very quantum phenomena could be harnessed rather than feared~\cite{Feynman1982,Manin1980,Benioff1980}. Their insight gave rise to \emph{Quantum Information Processing} (QIP): the use of quantum mechanics to store, process, and transmit information. QIP marks a conceptual shift from the deterministic logic of bits to the probabilistic, wave-like logic of \emph{qubits}, which can exist in coherent superpositions and share non-classical correlations known as entanglement~\cite{Horodecki2009}.

Over subsequent decades, this idea branched into several major subfields. Quantum computation focuses on using quantum states and quantum logic to
accelerate tasks such as factoring~\cite{Shor1994} and unstructured
search~\cite{Grover1996}, and includes the study of quantum algorithms and
quantum computational complexity~\cite{NielsenChuang,Aaronson2016}. Quantum communication uses quantum entanglement and quantum measurements
to perform information–processing tasks that classical networks cannot
achieve~\cite{Bennett1993,Briegel1998,Kimble2008}. Quantum sensing and quantum metrology use quantum coherence to improve
measurement precision beyond classical limits~\cite{Giovannetti04,Giovannetti2011}. Related developments in quantum foundations continue to deepen our
understanding of the physical principles underlying these capabilities,
for example through studies of nonlocality and Bell inequalities~\cite{Brunner2014RMP}. 
Taken together, these directions form the broader landscape of emerging
quantum technologies~\cite{Acin2018,NielsenChuang}.

Within this broad framework of QIP, this thesis focuses on two major research directions. The first is the practical realization of secure quantum communication, addressing the challenges of maintaining cryptographic integrity in real quantum networks. The second is the study of fundamental quantum dynamics, specifically the characterization of the ultimate physical speed limits that govern how fast quantum systems can evolve.

\section{Quantum Key Distribution: Security from Physics}

The first focus addresses secure communication in the quantum era. Contemporary public-key cryptosystems such as RSA rely on computational hardness assumptions~\cite{Rivest1978}. A sufficiently powerful quantum computer running Shor’s algorithm~\cite{Shor1994} could break them, threatening global data security. This motivates approaches that remain safe even against quantum adversaries.

\emph{Quantum key distribution} (QKD) offers such security by grounding
cryptography in the laws of physics rather than mathematics. Through the
transmission and measurement of quantum states, two remote users can establish
a shared random key whose secrecy is guaranteed by quantum principles such as
measurement disturbance and the no-cloning theorem, \emph{as long as the known
	laws of quantum physics remain valid}~\cite{Bennett1984,Ekert1991,Wootters1982}.

Unlike classical encryption schemes, QKD’s security is information-theoretic: no computational advance can undermine it.  

Over three decades of development, QKD has progressed from tabletop demonstrations~\cite{Bennett1984,Ekert1991,Gisin2002,Scarani2009,Lo2014,Diamanti2016,Xu2020review} to satellite links and metropolitan fibre networks~\cite{Liao2017,Chen2021}.

While the theoretical foundation of QKD provides a path toward information-theoretic security, a significant ``implementation-security gap" exists between idealized proofs and real-world hardware. Standard security proofs assume that Alice and Bob have perfect control over their quantum states, yet physical devices often exhibit imperfections that an eavesdropper, Eve, can exploit through ``side-channel" attacks.

A fundamental motivation for this research is the need to simultaneously address implementation vulnerabilities at both the source modulators and detection sides of QKD. While the Measurement-Device-Independent (MDI) framework eliminates detector-side vulnerabilities, such as blinding attacks, traditional sources still rely on active modulators susceptible to Trojan-horse attacks and pattern-dependent leakage. To neutralize these loopholes concurrently, I integrate the MDI framework with a fully passive source paradigm. Detailed in Chapter 3, the primary applied contribution of this thesis is the development and analysis of a fully passive MDI-QKD protocol that closes the implementation-security gap at both source modulator and detector ends of the system.

As communication moves toward scalable networks, securing groups of users is a vital requirement for the emerging quantum internet. In Chapter 4, I developed this work by introducing a fully passive conference key agreement (CKA) protocol. By leveraging twin-field QKD techniques, this protocol allows the key rate to scale with the square root of channel transmittance, significantly extending communication distance. Because the CKA design is intrinsically MDI, it removes detector-side vulnerabilities. Furthermore, the fully passive nature of the source ensures there are no loopholes at the source modulator for any participant, providing a secure protocol for multi-party communication. 

\section{Quantum Speed Limit: The Pace of Quantum Evolution}

The second focus explores a fundamental question in physics: how fast can a quantum system evolve? The \emph{Quantum Speed Limit} (QSL) provides a universal lower bound on the time required for a state to evolve into another, setting ultimate performance constraints on quantum technologies~\cite{Mandelstam45,Fleming73,Anandan90,Vaidman92,Uhlmann92,Pires2016,Deffner17a}. Whether describing the switching rate of a quantum gate~\cite{Giovannetti03,Giovannetti04}, the response time of a sensor~\cite{Pang2017,Deffner17a}, or the charging power of a quantum battery~\cite{Binder2015}, the QSL establishes the physical horizon of achievable speed.

While numerous historical QSL bounds exist, as reviewed in Chapter 5, many are specialized formulations developed by physicists to address specific physical systems. However, no existing bound simultaneously satisfies three critical criteria: high performance (tightness), broad universality across diverse physical systems, and computational simplicity. The motivation of this research is to develop such a framework, establishing its theoretical foundation while demonstrating its utility across varied physical scenarios. As demonstrated in Chapter 6, the proposed bound successfully integrates all three of these essential characteristics.  

While the mathematical foundations of QSL bounds are presented in the separate background chapter, this concept complements QKD. Where QKD ensures \emph{secure} transmission of quantum information, QSL bounds delineate the \emph{temporal limits} within which such secure operations can occur. Together they span two fundamental axes of quantum information science: reliability and speed.

\section{Thesis Overview}

This thesis develops two main themes: (i) practical quantum cryptographic protocols—including Quantum Key Distribution (QKD) and Quantum Conference Key Agreement (CKA)—and (ii) the fundamental limits of quantum dynamics through quantum speed limits (QSLs).

\begin{itemize}
	\item \textbf{Studies on QKD and CKA}
	\begin{itemize}
		\item \textbf{Chapter~2} introduces the theoretical background of QKD and related cryptographic primitives, covering core concepts, representative protocols, and key developments in practical implementations.
		\item \textbf{Chapter~3} presents the author’s work on fully passive measurement-device-independent QKD (passive MDI-QKD), including the protocol construction, security analysis, and numerical performance.  The main results of this chapter are based on the journal publication in Ref.~\cite{Li2024}.
		\item \textbf{Chapter~4} extends the fully passive paradigm to the multipartite setting and introduces a fully passive quantum conference key agreement protocol.  Its design, security analysis, and performance evaluation are based on the preprint in Ref.~\cite{Li2024CKA}.
	\end{itemize}
	
	\item \textbf{Studies on Quantum Speed Limits (QSLs)}
	\begin{itemize}
\item \textbf{Chapter~5} reviews the background of quantum speed limits, with emphasis on their geometric formulations, connections to resource-like quantities, and generalizations to time-dependent and open-system dynamics.
		\item \textbf{Chapter~6} develops a family of QSL bounds based on weighted $\ell^{p}_{w}$-(semi)norms that depend on the representation basis, and demonstrates their performance in several open system scenarios by comparison with existing bounds. The main results are based on the work reported in Ref.~\cite{Chau2025} and on a contributed talk at AQIS 2025~\cite{LiChau_AQIS2025}.
	\end{itemize}
\end{itemize}

The original contributions presented in Chapters~3, 4, and~6 have been disseminated in the form of peer-reviewed journal articles, arXiv preprints, and an international conference presentation~\cite{Li2024,Li2024CKA,Chau2025,LiChau_AQIS2025,trss-qhc8}. In this thesis, they are reorganized and extended to provide a unified and self-contained exposition.

	\chapter{QKD THEORETICAL BACKGROUND AND TOOLKIT}
	\label{chap:background}
	
Quantum key distribution (QKD) is a central application of quantum information
science. Its special strength is \emph{information-theoretic security}: it is secure
against adversaries even when they have unbounded computational and quantum
resources, \emph{provided that the known laws of quantum physics remain valid}
~\cite{Shannon1949,Maurer1993}. QKD bases its security directly on the laws of
quantum mechanics, in contrast to classical cryptography, which often relies on
unproven computational assumptions (for example, the hardness of factoring
~\cite{Rivest1978,Shor1994}).

The goal of QKD is not to send secret \emph{messages} directly. Instead, it
establishes a shared provable secret \emph{key} between two distant users,
traditionally called Alice (A) and Bob (B). Once this provably secure key is in
place, they can use conventional symmetric encryption schemes to protect their
actual messages, for example, the \emph{Advanced Encryption Standard (AES)}
~\cite{Daemen1999}, or the \emph{one-time pad (OTP)}, which is well known to
achieve information-theoretic security~\cite{Shannon1949,Vernam1926}.

As an example, consider a \textbf{one-time pad (OTP)}~\cite{Vernam1926,Shannon1949}. Each bit of the message is added (mod 2) to the corresponding shared key bit. If Alice wants to send the 4-bit message $M = 1011$ and the secret key is
$K = 0110$, she computes the ciphertext $C = M \oplus K = 1101$. Bob, holding
the same key $K$, computes $C \oplus K$ to recover the original message.
Because the shared key is uniformly random and as long as the message, the OTP
provides \emph{perfect secrecy}~\cite{Shannon1949}.

	\section{A Primer on Quantum Key Distribution}
	
	\subsection{The BB84 Protocol: A Narrative Overview}
	
	The main workflow of QKD can be clearly seen in the famous \textbf{BB84 protocol}, proposed by Bennett and Brassard in 1984~\cite{Bennett1984}. In this protocol, Alice encodes random bits into quantum states from two complementary bases, such as ${|H\rangle, |V\rangle}$ (horizontal/vertical) and ${|+\rangle, |-\rangle}$, where $|{\pm}\rangle = (|H\rangle \pm |V\rangle)/\sqrt{2}$. Bob, without knowing Alice’s choice, measures each photon he receives using a randomly chosen basis. Because their bases are not always the same, only part of the data can be kept. The process unfolds in several steps, following the general security framework developed by Bennett, Brassard, and Ekert~\cite{Bennett1984,Ekert1991,Gisin2002,Scarani2009,Lo2014}:
	
	\begin{enumerate}
		\item \textbf{Channels.}
		Alice and Bob are linked by two channels:
		\begin{itemize}
			\item a \textbf{quantum channel}, usually an optical fiber or a free-space link, for sending quantum states (single photons or weak coherent pulses) \cite{Townsend1993,Gisin2002,Lo2014};
			\item a \textbf{classical channel}, which is assumed to be  \emph{authenticated but public}~\cite{Bennett1984}.  
			Eve can listen to all the messages sent on it, but she cannot change them or pretend to be Alice or Bob. Authentication typically relies on a short pre-shared secret key and universal hashing~\cite{Wegman1981}.
		\end{itemize}
		
		\item \textbf{Quantum transmission.}
		In each round, Alice picks a random bit and a random basis ($Z$ or $X$).
		She encodes the bit as:
		\begin{itemize}
			\item $0 \mapsto |H\rangle$, $1 \mapsto |V\rangle$ in the $Z$ basis,
			\item $0 \mapsto |+\rangle$, $1 \mapsto |-\rangle$ in the $X$ basis.
		\end{itemize}
		She sends the prepared state through the quantum channel. Bob, unaware of her choice, randomly selects a basis ($Z$ or $X$) and measures the photon. Because of the \emph{no-cloning theorem}~\cite{Wootters1982}, Eve cannot make a perfect copy of the photon. If she measures in the wrong basis, she disturbs the state and introduces errors. Alice and Bob can detect these errors later as an increase in the Quantum Bit Error Rate (QBER)~\cite{Bennett1992}.
		
		\item \textbf{Sifting.}
		After transmission, Alice and Bob publicly announce which bases they used. They keep only the rounds where their bases match and discard the rest. This produces a correlated string of bits called the \emph{sifted key}~\cite{Bennett1984,Scarani2009}.
		
		\item \textbf{Parameter estimation.}
		To estimate the error rate, they reveal a random sample of their sifted key and compare the results. Those revealed bits are discarded. The observed error fraction is the \emph{Quantum Bit Error Rate (QBER)}. A high QBER means eavesdropping or too much noise, so the protocol stops~\cite{Bennett1992,ShorPreskill2000}. A low QBER means they can move on to extract a secure key.
		
		\item \textbf{Error correction.}
		Because of imperfections in the channel or devices, Alice’s and Bob’s sifted keys may differ.
		This step ensures they end up with the same raw key.
		They exchange limited classical information (like parity checks on blocks of bits) to correct errors without revealing too much to Eve. Protocols such as Cascade~\cite{Brassard1993} or LDPC codes~\cite{Gallager1962} can be used.
		The goal is always to reconcile their strings with minimal information leakage.
		
		\item \textbf{Privacy amplification.}
		Even after error correction, Eve might still know something about the key. To erase any remaining knowledge, Alice and Bob apply \emph{privacy amplification}~\cite{Bennett1988,Renner2008}. They use a pre-agreed \emph{hash function} to compress the raw key into a shorter one. By choosing the output length carefully, they ensure that the final key looks perfectly random and independent of Eve's knowledge.
		
\item \textbf{Key usage.}
The final shared key can now be used for symmetric cryptography. In theory, it could serve directly as a \emph{one-time pad (OTP)}, which offers perfect secrecy but requires a key as long as the message~\cite{Shannon1949}.
\end{enumerate}

In short, BB84 provides a method to share secret keys over insecure channels~\cite{Bennett1984,Ekert1991,Scarani2009,Lo2014}.
Modern security analyses show that the protocol is not only unconditionally secure but also \emph{composable}.\footnote{I briefly mention composability here for completeness. In modern QKD theory, security is typically defined in the \emph{universally composable} (UC) framework, which guarantees that the final key can be safely used in any subsequent cryptographic application; see Refs.~\cite{ShorPreskill2000,Renner2008}.}
This guarantees a built-in ``fail-safe'' behavior: either Alice and Bob share a key that is secure against any adversary, or the protocol aborts without producing a weak key.

The process described above captures the essence of the \textbf{BB84 protocol}.
Many later protocols are extensions or refinements of this idea, designed to handle imperfect devices, detector loopholes, or large-scale multi-user networks~\cite{Gisin2002,Lo2014,Curty2014,Xu2020review,Wang2023a}.

	\section{The Language of Quantum Information}
	
	\subsection{Quantum States: Qubits and Density Operators}
	
	The basic unit of information in QKD is the \textbf{qubit}, a two-level quantum system described by
	\begin{equation}
		|\psi\rangle = c_0 |0\rangle + c_1 |1\rangle, \qquad c_0, c_1 \in \mathbb{C}, \quad |c_0|^2 + |c_1|^2 = 1 .
	\end{equation}
	
	A qubit can also exist as a \textbf{statistical mixture} of states, which I represent using a \textbf{density operator}:
	\begin{equation}
		\rho = \sum_i p_i |\psi_i\rangle\langle \psi_i| ,
	\end{equation}
	where the probabilities $\{p_i\}$ describe classical uncertainty~\cite{vonNeumann1932,NielsenChuang}.
	
	Two bases are commonly used:
	\begin{itemize}
		\item the \emph{computational basis} $\{|0\rangle, |1\rangle\}$,
		\item the \emph{diagonal basis} $\{|+\rangle, |-\rangle\}$, where $|\pm\rangle = (|0\rangle \pm |1\rangle)/\sqrt{2}$.
	\end{itemize}
	
	Measuring a qubit in one basis disturbs it if it was prepared in the other.
	This property, called measurement disturbance, is the key to detecting eavesdropping in QKD~\cite{Bennett1984,Ekert1991}.
	
	\subsection{Physical Realizations: Polarization and Phase Encoding}
	\label{sec:physical-encodings}
	
	So far, qubits have been treated abstractly as $|\psi\rangle = c_0|0\rangle + c_1|1\rangle$.
	In \emph{discrete-variable (DV)} photonic QKD, which is the focus of this thesis, the two main physical encodings are \emph{polarization} and \emph{time-bin/phase} encoding~\cite{Gisin2002,Scarani2009,Lo2014}.
	Both methods map the logical basis $\{|0\rangle, |1\rangle\}$ onto optical modes that can be prepared, transmitted, and detected using single-photon detectors.\footnote{Other DV options include path or frequency-bin encodings. A different family, called continuous-variable (CV) QKD, encodes information in optical field quadratures, but this is outside my current scope~\cite{Weedbrook2012}.}
	
	\subsubsection*{Polarization encoding}
	
	In polarization encoding, the logical basis states are mapped to orthogonal linear polarizations:
	\begin{equation}
		|0\rangle \equiv |H\rangle, \qquad |1\rangle \equiv |V\rangle,
	\end{equation}
	and the complementary ($X$) basis is
	\begin{equation}
		|+\rangle = \tfrac{1}{\sqrt{2}}(|H\rangle + |V\rangle), \qquad
		|-\rangle = \tfrac{1}{\sqrt{2}}(|H\rangle - |V\rangle).
	\end{equation}
	
	State preparation uses standard polarization optics, and measurement is done with a polarizing beam splitter (PBS) and two single-photon detectors.
	Polarization encoding is conceptually simple and widely used in free-space and satellite QKD~\cite{Gisin2002,Scarani2009}. However, in long optical fibers, random birefringence can scramble polarization, in those cases, active polarization tracking or an alternative encoding is required~\cite{Gisin2002}.
	
	\subsubsection*{Time-bin / phase encoding}
	
	In time-bin or phase encoding, information is stored in the \emph{relative phase} between two separated optical pulses—one early, one late.
	This can be represented as
	\begin{equation}
		|0\rangle \equiv |1,0\rangle, \qquad |1\rangle \equiv |0,1\rangle,
	\end{equation}
	where $|1,0\rangle$ means one photon in the early bin and $|0,1\rangle$ means one photon in the late bin.
	A general superposition takes the form
	\begin{equation}
		|\psi\rangle = \cos\frac{\theta}{2}\,|1,0\rangle + e^{i\phi}\sin\frac{\theta}{2}\,|0,1\rangle,
	\end{equation}
	up to a physically irrelevant global phase. Here, $\theta$ and $\phi$ are the standard angular coordinates of the Bloch sphere, which is discussed in the next section~\cite{NielsenChuang}.
	
	Detection uses an unbalanced Mach–Zehnder interferometer that recombines the two pulses and measures phase-dependent interference~\cite{Brendel1999,Tittel2000}.
	Time-bin encoding works especially well in fiber, since it is largely immune to slow polarization drift~\cite{Gisin2002,Scarani2009}.
	In practice, systems use active phase stabilization (for example, pilot tones or reference pulses) to maintain interference visibility~\cite{Tittel2000,Stucki2002}.
	
	\paragraph{A note on phase \emph{encoding} vs.\ phase \emph{randomization}.}
	The relative phase between bins carries information in phase encoding.
	By contrast, security proofs often require \emph{global phase randomization} of weak coherent pulses to ensure Poissonian photon-number statistics for decoy-state analysis~\cite{Lo2005,Wang2005,Ma2008}.
	These two procedures serve different purposes and should not be confused.
	Both are present in practical systems~\cite{Scarani2009,Lo2014}.
	
	\subsubsection*{The Bloch Sphere as a Unifying Picture}
	
	Both polarization and phase encodings can be visualized using the Bloch sphere.
	Each pure qubit corresponds to a point on the sphere:
	\begin{equation}
		|\psi\rangle = \cos\frac{\theta}{2} |0\rangle + e^{i\phi}\sin\frac{\theta}{2} |1\rangle,
	\end{equation}
	where $(\theta, \phi)$ are spherical coordinates~\cite{Bloch1946,Fano1957,NielsenChuang}.
	On this sphere:
	\begin{itemize}
		\item the poles represent the basis states ($|0\rangle$, $|1\rangle$),
		\item the equator represents equal superpositions with different relative phases,
		\item orthogonal states are separated by $180^\circ$.
	\end{itemize}
	
	This geometric view is intuitive.
	Polarization states such as 
	$|H\rangle$, $|V\rangle$, $|D\rangle$, $|A\rangle$, $|R\rangle$, and $|L\rangle$
	appear as fixed points on the sphere, where
	\[
	|H\rangle \equiv |0\rangle, \quad
	|V\rangle \equiv |1\rangle, \quad
	|D/A\rangle = \tfrac{1}{\sqrt{2}}(|H\rangle \pm |V\rangle), \quad
	|R/L\rangle = \tfrac{1}{\sqrt{2}}(|H\rangle \pm i|V\rangle),
	\]
	consistent with the standard Jones and Bloch-sphere descriptions~\cite{NielsenChuang}.
	
	Phase encoding corresponds to rotations along the equator. The Bloch sphere is especially useful when analyzing protocols that use arbitrary qubit states or mix polarization and phase degrees of freedom~\cite{NielsenChuang}.
	
	\subsubsection*{Relevance for Fully Passive Sources}
	
	The difference between abstract qubits and their physical encodings is key to understanding the \emph{fully passive} approach. In such systems, quantum states are not actively created by modulators. Instead, they appear randomly from interference between phase-randomized laser pulses. Users monitor part of the output locally and \emph{post-select} which polarization or phase state was actually emitted into the channel~\cite{Curty2010,Zapatero2023,Wang2023a,Hu2023,Lu2023}.
	
	To follow later chapters, it is important to be comfortable translating between abstract qubit notation and these physical realizations in polarization and phase encoding.

	\section{Performance and Security of QKD Protocols}
	
	\subsection{Figures of Merit}
	
	The performance of a QKD protocol is usually described by two main quantities:
	the \emph{quantum bit error rate} (QBER) and the \emph{secret key rate}. Together, they show whether Alice and Bob can extract a secure key, and how efficiently it can be done~\cite{Scarani2009,Xu2020review}.
	
	\subsubsection{Quantum Bit Error Rate (QBER)}
	
	The QBER is defined as
	\begin{equation}
		E = \frac{\text{Number of erroneous sifted bits}}{\text{Total sifted bits}} .
	\end{equation}
	
	It represents the fraction of mismatched bits between Alice’s and Bob’s sifted keys. In a perfect, noiseless setup, $E = 0$. In reality, several factors make $E$ nonzero:
	\begin{itemize}
		\item \textbf{Channel noise:} photons can be depolarized or phase-shifted during transmission.
		\item \textbf{Detector imperfections:} limited efficiency, dark counts, and timing errors may cause false detections.
		\item \textbf{Eavesdropping:} when Eve interacts with quantum signals, her disturbance creates errors.
	\end{itemize}
	
	The QBER is thus a direct sign of both natural noise and possible eavesdropping. Each QKD protocol defines a \emph{threshold} QBER.
	If the observed QBER is higher than this limit, the protocol stops and no key is kept. For BB84, the theoretical upper limit is about $11\%$ in the asymptotic regime~\cite{ShorPreskill2000}.\footnote{In BB84, the asymptotic key rate scales as $1 - 2H_2(E)$. The rate becomes zero around $E \approx 11\%$, meaning no secure key can be distilled beyond that.}
	
	\subsubsection{Secret Key Rate}
	
	In the limit of infinitely many signals and under collective attacks, the Devetak--Winter bound provides a lower bound on the secret key rate per channel use~\cite{DevetakWinter2005,Renner2008}:
	\begin{equation}
		R \ge q \left[ 1 - f_{\text{EC}} H_2(E) - I_E \right],
		\label{eq:DW}
	\end{equation}
	where
	\begin{itemize}
		\item $q$ is the sifting factor (for BB84, $q=1/2$),
		\item $f_{\text{EC}} \ge 1$ is the inefficiency factor of error correction,\footnote{This factor quantifies the inefficiency of the practical error correction protocol (e.g., LDPC codes) used for reconciliation. It represents the ratio of information \emph{actually} leaked compared to the theoretical minimum required by the Shannon limit, $H_2(E)$. A perfect, unattainable protocol would have $f_{\text{EC}} = 1$. Common values assumed in analyses and achieved by practical systems are often in the range of $f_{\text{EC}} \approx 1.1$ to $1.2$. Importantly, although many key-rate analyses treat $f_{\text{EC}}$ as a constant for simplicity, it is in fact \emph{dependent on the QBER} because practical reconciliation codes (such as LDPC or Cascade) are optimized differently at different error levels. The specific value depends on the chosen code and its parameters, which are optimized for the channel conditions (i.e., the QBER $E$) and the target secret key rate $R$.}
		
\item $H_2(E) = -E \log_2 E - (1-E) \log_2(1-E)$ is the binary entropy of the QBER $E$.
		\item $I_E$ is an upper bound on Eve’s information per sifted bit.
	\end{itemize}
	
	The Devetak--Winter bound is derived within the modern entropic approach to
	QKD security, which uses tools such as the quantum leftover-hash lemma,
	smooth min-entropy, and one-way classical post-processing via privacy
	amplification~\cite{DevetakWinter2005,Renner2008}. These techniques relate Eve’s
	information to the conditional von Neumann entropy of Alice’s key given Eve’s
	quantum system. In the specific case of BB84, this general entropic framework
	can be connected to earlier entanglement-based proofs, such as the
	Shor--Preskill argument, which expresses Eve’s information in terms of the
	phase-error rate in the complementary basis. This connection explains why the
	Shor--Preskill picture is often used to interpret the Devetak--Winter
	expression in prepare-and-measure protocols.

	In BB84, security proofs such as Shor--Preskill~\cite{ShorPreskill2000} show that Eve’s knowledge is tied to the \emph{phase error rate} $e_p$ (the error rate in the $X$ basis).
	Security analysis typically relies on the assumption of a symmetric channel, where this unmeasurable phase error rate $e_p$ is equal to the single-photon bit error rate, $e_1$. Thus, Eve's information is $I_E = H_2(e_p) = H_2(e_1)$, and the goal is to find $e_1$.
	
	To understand the fundamental limits, it is useful to first consider an \textbf{idealized model} (as in Shor–Preskill) that assumes a \textbf{perfect single-photon source}. In this ideal case, all signals are single photons, so the \emph{measured} QBER, $E$, is identical to the single-photon QBER, $e_1$. The symmetry assumption ($e_p = e_1$) therefore implies $e_p = E$. Eve's information is thus given directly by the \emph{measured} QBER: $I_E = H_2(E)$.
	
	The bound then simplifies to the well-known form:
	\begin{equation}
		R \ge q \big[\, 1 - f_{\text{EC}} H_2(E) - H_2(E) \,\big].
	\end{equation}
	
	This form makes the $11\%$ threshold easy to see:
	once the combined costs of error correction ($f_{\text{EC}} H_2(E)$) and privacy amplification ($H_2(E)$) exceed the information content of the raw key, no secret key remains. This occurs when $1 - (f_{\text{EC}} + 1)H_2(E) \le 0$, which for a perfect protocol ($f_{\text{EC}}=1$) is around $E \approx 11\%$.
	
	(This analysis also highlights the central problem for \emph{practical} systems: if my source is not a perfect single-photon source, e.g., a WCP, the measured $E$ is \emph{not} equal to $e_1$, and this simple formula no longer holds. We must then use other methods, such as the decoy-state technique, to find the true value of $e_1$~\cite{Lo2005,Wang2005,Ma2008,Scarani2009,Xu2020review}.)
	
	\subsubsection{Finite-Size Effects}
	
	In practice, the number of exchanged quantum signals is always finite. This means that observed values, such as QBER, are only statistical estimates of the true quantities. To ensure security, Alice and Bob must include worst-case corrections. These are typically derived using the smooth entropy framework~\cite{Renner2008,Tomamichel2012}, which generalizes standard entropic measures to account for finite statistics and failure probabilities.
	
	A typical finite-size key-length bound is written as~\cite{Renner2008,Tomamichel2012,Curty2014,Lim2014}
\begin{equation}
	\ell \le n_{\text{sifted}} \left( 1 - I_E^{\text{upper}} \right) 
	- \text{leak}_{\text{EC}} 
	- \Delta_{\text{PE}} 
	- \Delta_{\text{PA}},
\end{equation}
	where
	\begin{itemize}
		\item $n_{\text{sifted}}$ is the number of sifted bits,
		\item $I_E^{\text{upper}}$ is an upper bound on Eve’s information per bit, estimated from parameter estimation with high confidence (using, for example, Chernoff or Azuma bounds),
		\item $\text{leak}_{\text{EC}}$ is the information disclosed during error correction,
		\item $\Delta_{\text{PE}}$ and $\Delta_{\text{PA}}$ are correction terms for parameter estimation and privacy amplification, which usually scale logarithmically with $n_{\text{sifted}}$.
	\end{itemize}
	
	The key idea is that finite-size effects always shorten the final key compared with the ideal asymptotic case. These penalties disappear as $n \to \infty$, but for practical key blocks they must be included to ensure composable security~\cite{Renner2008,Tomamichel2012,Scarani2009,Xu2020review}. A more detailed discussion on the smooth entropy method and its role in the security proof is provided later in Section~2.6.

\subsection{Principles of Quantum Security}

The security of QKD relies on two kinds of ingredients that play different conceptual roles: 
(1) fundamental physical laws of quantum mechanics that constrain an eavesdropper’s capabilities, and 
(2) cryptographic requirements that define what it means for the generated key to be secure.

\paragraph{Fundamental quantum-mechanical principles.}
Two physical properties of quantum states ensure that eavesdropping inevitably leaves detectable traces:
\begin{itemize}
	\item \textbf{No-cloning theorem:} it is impossible to create a perfect copy of an unknown quantum state. This prevents Eve from duplicating photons without interacting with them~\cite{Wootters1982,Dieks1982}.
	\item \textbf{Measurement disturbance:} any attempt by Eve to distinguish non-orthogonal states necessarily introduces errors, observable as an increase in the QBER~\cite{Bennett1984,Ekert1991}. 
\end{itemize}

\paragraph{Cryptographic requirement: composable security.}
Beyond these physical constraints, a QKD protocol must satisfy a modern cryptographic notion of security.  
A protocol is called $\varepsilon_{\mathrm{sec}}$-secure if, except with probability $\varepsilon_{\mathrm{sec}}$, the joint state of Alice, Bob, and Eve is indistinguishable from an ideal situation in which Alice and Bob share a uniformly random secret key.  
This \emph{composable} definition ensures that the key remains secure even when used as a component inside any larger cryptographic system~\cite{Renner2008,Tomamichel2012}.  

The protocol must also guarantee \textbf{correctness}: Alice and Bob’s final keys match with all but a small probability $\varepsilon_{\mathrm{cor}}$. Combining both parts gives the overall security parameter:
\begin{equation}
	\varepsilon_{\mathrm{QKD}} = \varepsilon_{\mathrm{cor}} + \varepsilon_{\mathrm{sec}}.
\end{equation}

This thesis adopts the universally composable (UC) security framework described above, which is standard in modern QKD analysis~\cite{Renner2008,Tomamichel2012,Scarani2009,Xu2020review}.

\section{From Theory to Practice: Addressing Imperfections}
Real-world QKD systems differ greatly from the clean, idealized models found in textbooks. Practical devices come with \textbf{imperfections}, and if these are not properly handled, they can weaken or even break security. Quantum channels, for example, are impacted by both \textbf{loss} (which reduces the sifted key rate) and \textbf{noise} (which increases the QBER). While textbook protocols like BB84 are designed to handle channel noise, practical systems must contend with both. Detectors have limited efficiency and may be exposed to side-channel attacks. Furthermore, ideal single-photon sources are difficult to build. Most systems instead use \textbf{weak coherent pulses (WCPs)}, which are common because they are economical and offer much higher photon rates than many single-photon sources, but these sometimes emit multiple photons, enabling photon-number-splitting attacks. Active modulators used for encoding can also leak information, allowing Trojan-horse attacks on the source.

It is worth noting that a different security paradigm, \textbf{Device-Independent (DI) QKD}, aims to remove \emph{all} device assumptions by using a loophole-free Bell test. However, its extreme experimental requirements, such as near-unity detection efficiency, make it currently impractical for long-distance, high-loss networks. This thesis therefore focuses on practical, device-dependent security models that address the specific, known vulnerabilities mentioned above.

These imperfections are not just small engineering flaws—they fundamentally change the security picture of QKD.  
Over the years, researchers have developed specific methods to counter each major vulnerability:
\begin{itemize}
	\item The \textbf{decoy-state method} secures weak-coherent sources by statistically identifying and isolating single-photon contributions.
	\item \textbf{Measurement-device-independent (MDI) QKD} removes all detector side-channels by moving measurements to an untrusted relay.
	\item \textbf{Fully passive protocols}, including those presented in this thesis, close \emph{source-modulator} side-channels by eliminating active modulators entirely from the signal path.
\end{itemize}

The following subsections introduce these methods in more detail. Together, they form the foundation of secure and practical QKD implementations used today.

\subsection{Practical Implementations and Their Vulnerabilities}
\label{sec:vulnerabilities}

Real-life devices often behave slightly differently from their mathematical models, normally called imperfections, and such differences can open loopholes for an eavesdropper, Eve, to exploit. The main vulnerabilities in QKD can be grouped by where they occur: the \emph{source}, the \emph{channel}, and the \emph{detector}.

\subsubsection{Source Vulnerabilities}

The source, usually operated by Alice (and also by Bob in MDI-QKD), is responsible for encoding information into quantum states. If the source is imperfect, it might leak information even before the signal enters the quantum channel.

\begin{itemize}
	\item \textbf{Weak Coherent Pulses and the Photon-Number-Splitting (PNS) Attack:}  
	True single-photon sources are difficult to build, so practical QKD systems use strongly attenuated lasers to produce weak coherent pulses (WCPs). The number of photons in each pulse follows a Poisson distribution: most pulses contain one photon, some have none, and a few contain two or more. In a \textbf{PNS attack}, Eve takes advantage of channel loss by splitting off one photon from every multi-photon pulse and sending the rest to Bob. She can then store her photon and measure it later to learn the bit value—without introducing noticeable errors. This makes naïve WCP-based QKD insecure unless additional countermeasures, such as the decoy-state method, are applied.
	
	\item \textbf{Modulator Side-Channels and the Trojan-Horse Attack:}  
	Active phase or polarization modulators can unintentionally leak information about the encoding process. In a \textbf{Trojan-horse attack}, Eve shines bright probe light into the transmitter. The small amount of light reflected back carries clues about the internal modulator settings (phase or polarization), allowing Eve to infer the transmitted bit. This type of attack motivates the use of \emph{fully passive} transmitters, which remove active modulators from the signal path entirely.
\end{itemize}

\subsubsection{Detector Vulnerabilities}

Single-photon detectors are another major target.  
Security proofs assume that each detector “click” represents a genuine quantum measurement, but in practice, detectors may respond differently under certain conditions—creating exploitable weaknesses.

\begin{itemize}
	\item \textbf{Detector Side-Channel Attacks:}  
	A well-known example is the \textbf{detector blinding attack}.  
	Eve shines carefully tuned bright light onto the detectors, forcing them into a classical operating mode. She can then control which detector clicks by adjusting the light’s timing and power, effectively determining Bob’s measurement outcomes. This class of attacks inspired \textbf{MDI-QKD}, which removes trust from the detection device by shifting measurement to an untrusted relay.
\end{itemize}

\paragraph{From vulnerabilities to protocol design.}
These real-world weaknesses have led to a set of defense tools:
\emph{decoy states} protect weak-coherent sources from PNS attacks,  
\emph{MDI-QKD} removes detector side-channels by moving measurement to an untrusted node,  
and the \emph{fully passive} approach eliminates source-modulator side-channels by removing active modulators. Together, these techniques form the basis for the \emph{fully passive MDI-QKD} and \emph{fully passive CKA} protocols developed in this thesis.

	\subsection{The Decoy-State Method}
	
	\paragraph{The big idea.}
	The main idea behind the decoy-state method is simple:
	\emph{randomize the intensity of each pulse so Eve cannot take advantage of multi-photon emissions.} By comparing detection results from different intensity settings, Alice and Bob can statistically isolate the single-photon contribution that guarantees security~\cite{Hwang2003,Lo2005,Wang2005,Ma2008,Scarani2009,Xu2020review}.
	
	\medskip
	In real QKD systems, ideal single-photon sources are rare.
	Instead, most use attenuated lasers to generate \textbf{weak coherent pulses (WCPs)} with an average photon number $\mu < 1$.
	When the \emph{global optical phase} of each pulse is randomized, the photon number follows a Poisson distribution:
	\begin{equation}
		P_\mu(n) = e^{-\mu}\,\frac{\mu^n}{n!}\,.
	\end{equation}
	Each $n$-photon component behaves as an $n$-photon \emph{Fock} state.
	However, occasional multi-photon pulses enable the photon-number-splitting (PNS) attack described earlier~\cite{Brassard2000,Lutkenhaus2000}.
	
	The \textbf{decoy-state method}, proposed by Hwang and later by Lo, Ma, and Chen, solves this issue by having Alice (and Bob, in MDI-QKD) randomly vary the intensity of their pulses~\cite{Hwang2003,Lo2005,Wang2005}.  
	They choose between several intensity levels—one for ``signal'' states and one or more for ``decoy'' states.
	Since Eve cannot tell, during transmission, whether a given pulse is a signal or a decoy, she must treat all pulses with the same photon number $n$ in the same way.
	This leads to the key \emph{decoy assumption}: the conditional detection probability (\emph{yield}) $Y_n$ and the error rate $e_n$ depend only on $n$ and the channel, not on the intensity setting.
	
	Formally, for a given intensity $\mu$, the observed detection rate (\emph{gain}) and error gain satisfy
	\begin{equation}
		Q_\mu = \sum_{n=0}^{\infty} P_\mu(n)\,Y_n,
		\qquad
		Q_\mu E_\mu = \sum_{n=0}^{\infty} P_\mu(n)\,Y_n\,e_n,
	\end{equation}
	where $Q_\mu$ is the overall detection probability and $E_\mu$ is the measured QBER for that intensity.
	By combining the data from multiple intensities, Alice and Bob can calculate upper and lower bounds on the single-photon parameters $Y_1$ and $e_1$ using analytical methods or linear programming~\cite{Ma2008,Scarani2009,Xu2020review}.  
	These quantities are exactly what enter the key-rate formulas and security proofs.
	
	\paragraph{Why it works.}
	Only single-photon signals offer true information-theoretic security against PNS attacks.
	By placing a lower bound on $Y_1$ and an upper bound on $e_1$, Alice and Bob confirm that enough of their detections come from single photons with an acceptably low error rate—even though the actual source is imperfect.
	
	\paragraph{Finite-size remark.}
	In experiments, the estimates of $Y_1$ and $e_1$ include statistical confidence intervals (for example, based on Chernoff or Azuma bounds).
	These intervals depend on the block size and chosen security parameters and directly affect the final secure key length~\cite{Tomamichel2012,Lim2014}.
	
	\medskip
	In short, the decoy-state method turns an otherwise insecure weak-coherent source into a practically secure one by
	\emph{forcing Eve to act consistently across different intensity settings} and extracting reliable single-photon statistics.
	The same principle applies broadly:
	in \emph{MDI-QKD}, both Alice and Bob use decoys to estimate single-photon contributions at the relay~\cite{Lo2012,Xu2020review};
	in \emph{twin-field QKD}, decoys help bound interference visibility~\cite{Lucamarini2018};
	and in \emph{conference key agreement (CKA)} and the \emph{fully passive} schemes presented in this thesis, the same reasoning holds—with either active or passive intensity randomization.
	
\paragraph{Transition.}
While decoy states protect imperfect sources, the \emph{detectors} still remain major points of vulnerability.
This motivates the next advance,\linebreak \emph{measurement-device-independent QKD}, which removes the need to trust the measurement apparatus.
	
\subsection[Measurement-Device-Independent QKD (MDI-QKD)]{Measurement-Device-Independent \\ QKD (MDI-QKD)}
	\label{sec:mdi-qkd}
	
	\paragraph{The big idea.}
	The main goal of MDI-QKD is to remove the need to trust the detectors.  
	In this setup, Alice and Bob each send quantum states to an untrusted middle node that performs an interference measurement.  
	Only the \emph{correlations} between their bits (for example, whether the bits are the same or different) are made public, while the actual bit values remain secret.  
	Security is therefore based on the trusted \emph{sources} at Alice and Bob’s ends, not on the measurement device~\cite{Lo2012,Braunstein2012}.
	
	In practice, real detectors can behave differently from ideal models and are vulnerable to side-channel attacks such as detector blinding~\cite{Makarov2006,Lydersen2010,Gerhardt2011}.  
	MDI-QKD solves this by outsourcing all measurements to an untrusted relay (often called Charlie).  
	Alice and Bob generate keys only from Charlie’s public announcements, making the scheme immune to detector-based attacks~\cite{Lo2012}.
	
	\subsubsection{The Bell-State Measurement in Practice}
	
	At the relay, Charlie interferes the optical pulses sent by Alice and Bob using a balanced beam splitter.  
	The two output ports are connected to single-photon detectors, which can distinguish certain polarization or time-bin modes.  
	Due to the limits of linear optics, not all four Bell states can be perfectly distinguished—only two can be identified unambiguously~\cite{Braunstein2012}. 
	
	Alice and Bob’s sifting process depends entirely on what Charlie announces.  
	Only two-photon coincidence detections are kept, since they indicate a successful Bell-state projection.  
	Events with no clicks or only one detector click carry no correlation information and are discarded.  
	The type of coincidence determines whether Alice’s and Bob’s bits are the same or different, which in turn defines the bit assignment rule, as shown below~\cite{Lo2012}.

\begin{table}[h!]
	\centering
	\small 
	\begin{tabular}{p{3.2cm} p{4.0cm} p{1.5cm} p{3.3cm}}
		\toprule
		\textbf{Charlie's \newline Announcement} & \textbf{Matching Basis \& \newline Implied Correlation} & \textbf{Sift} & \textbf{Bit Assignment Rule} \\
		\midrule
		Unsuccessful BSM \newline (e.g., not 2 clicks) & None & Discard & N/A \\
		\addlinespace
		$|\psi^{-}\rangle$ projection \newline (Singlet state) & \textbf{Both Bases (Z or X):} \newline Anti-correlated & \textbf{Keep} & Apply bit flip \newline (e.g., Bob flips). \\
		\addlinespace
		$|\psi^{+}\rangle$ projection \newline (Triplet state) & \textbf{Rectilinear (Z):} \newline Anti-correlated & \textbf{Keep} & Apply bit flip. \\
		\addlinespace
		$|\psi^{+}\rangle$ projection \newline (Triplet state) & \textbf{Diagonal (X):} \newline Correlated & \textbf{Keep} & No bit flip. \\
		\bottomrule
	\end{tabular}
	\caption{Sifting and bit assignment logic in MDI-QKD based on the relay's public announcements. Only successful two-photon coincidence detections corresponding to valid Bell states are used.}
	\label{tab:mdi_sifting}
\end{table}

	Through this procedure, Alice and Bob can establish a correlated raw key using only Charlie’s public messages, without ever revealing their own bit values~\cite{Lo2012}.
	
	\subsubsection{The Source of Security}
	
	The strength of MDI-QKD lies in one key principle:  
	Charlie’s public announcement reveals only the \emph{correlation} between Alice’s and Bob’s bits, not the bits themselves.  
	Even if Charlie is controlled by Eve, she gains no direct information about the key unless she abandons the proper Bell measurement and performs a different attack~\cite{Lo2012,Braunstein2012}.
	
	Consider such a case:  
	Eve replaces the relay’s interferometer with her own detectors and performs an intercept–resend attack.  
	Because she does not know which basis Alice and Bob used, her measurement will inevitably disturb the quantum states.  
	For example, if she measures an $X$-basis photon using $Z$-basis detectors, she collapses the state and must resend a new one, introducing a high chance of error.
	
	These disturbances show up in the measurable statistics.  
	During parameter estimation, Alice and Bob publicly compare a subset of their sifted keys.  
	If Eve has interfered, the Quantum Bit Error Rate (QBER) will rise far above what is expected from normal channel noise.  
	A low QBER, on the other hand, means Eve’s potential information is tightly bounded.  
	Thus, even if Eve builds and operates the entire relay, Alice and Bob can still extract a secure key—so long as their sources are well-characterized and protected (for instance, using the decoy-state method to counter PNS attacks)~\cite{Lo2012,Scarani2009,Xu2020review}.
	
	This measurement-device-independent model forms the basis for the \emph{fully passive MDI-QKD} protocol developed later in this thesis.  
	MDI-QKD is a cornerstone of this research, and its theoretical structure will be examined in more depth in the next chapter.
	
	\subsection{The ``Fully Passive'' Paradigm: Eliminating Modulator Side-Channels}
	\label{sec:fully-passive}
	
	While MDI-QKD successfully removes vulnerabilities at the detector, weaknesses can still appear at the source.  
	Traditional QKD transmitters rely on active optical modulators (for intensity, phase, or polarization) to prepare quantum states.  
	These modulators are especially exposed to \textbf{Trojan-horse attacks}, in which Eve injects bright light into the transmitter and analyzes the weak back-reflections to learn the encoding choices of Alice or Bob~\cite{Gisin2006,Jain2014}.  
	Other imperfections in modulation can also leak unintended patterns of information.  
	The \textbf{fully passive paradigm} is designed to eliminate this entire class of side-channels by removing active modulators from the signal path~\cite{Curty2010}.
	
	\subsubsection{Core Concept: Intrinsic Randomness and Post-Selection}
	
	In a fully passive transmitter, quantum states are not created through active control but instead emerge from \emph{intrinsic randomness} in optical interference.  
	For example, in the fully passive MDI-QKD scheme explored in this thesis, each user employs four independent phase-randomized lasers.  
	Their outputs interfere through a network of beam splitters (BS) and a polarizing beam splitter (PBS), producing optical pulses with random intensities and polarizations that naturally correspond to the BB84 basis states and decoy classes.  
	A fraction of the emitted light is monitored locally so that the user can identify which state was generated in each round.  
	Only those rounds where the monitored outcome falls within a desired set (for instance, an $H$, $V$, $D$, or $A$ state with acceptable intensity) are \emph{post-selected} and used in the protocol.
	
	\subsubsection{Security and Practical Significance}
	
	By replacing active modulation with passive generation and classical post-selection, the transmitter becomes free of \emph{modulator-based side-channels}.  
	Since no active component directly encodes the bit or basis choice, Eve cannot obtain this information by probing the modulators~\cite{Gisin2006,Jain2014,Curty2010}. In effect, the choices of basis and decoy intensity are determined \emph{after} transmission, based on the local monitoring results.
	
	It is worth noting that fully passive sources do not eliminate every possible side-channel, for example, imperfections in lasers or other optical components can still pose risks and must be handled with isolators and filters. However, the most significant and experimentally confirmed loopholes—Trojan-horse attacks on modulators—are completely avoided by design~\cite{Jain2014}.
	
	When combined with MDI-QKD, the fully passive paradigm removes the main side-channels at \emph{both} ends of the system: the source and the detector.  
	This provides a much stronger and more realistic security framework for practical QKD implementations~\cite{Lo2012,Curty2010}.  
	This thesis builds upon this concept and presents two specific realizations: \emph{fully passive MDI-QKD} and \emph{fully passive CKA}.
	
	The primary disadvantage of this security-by-design approach is a significant trade-off in performance, namely a \textbf{much lower key rate}.
	Because state generation is probabilistic, it relies heavily on post-selection, a large fraction of the generated signals, which do not randomly fall into the desired state or intensity categories, must be discarded.
	This reduced sifting efficiency, when compared to deterministic active modulation, is the main cost of eliminating modulator-based side-channels~\cite{Curty2010}.

	\section{Advanced QKD Protocols}
	
	\subsection{Twin-Field QKD (TF-QKD) and Single-Photon Interference}
	\label{sec:tfqkd}
	
	\paragraph{Motivation.}
	In standard point-to-point QKD systems with trusted detectors (and even in MDI-QKD), the secret key rate scales as $\mathcal{O}(\eta)$, where $\eta$ is the total channel transmittance. This scaling limits how far secure communication can go without using quantum repeaters~\cite{Pirandola2017PLOB}. Twin-Field QKD (TF-QKD) breaks this repeaterless rate–distance barrier at the \emph{architectural} level~\cite{Lucamarini2018}. It does so by using \emph{single-photon interference at an untrusted central node}, achieving an effective scaling of $\mathcal{O}(\sqrt{\eta})$ under realistic conditions—without requiring repeaters or quantum memories~\cite{Lucamarini2018}.
	
	\paragraph{Core idea (conceptual).}
	In TF-QKD, Alice and Bob each send phase-randomized weak coherent pulses to an \emph{untrusted} central station (Charlie). Charlie interferes the two incoming optical modes on a balanced beam splitter and records which detector clicks. 
	
	A \emph{single click} indicates a single-photon–effective event. Crucially, the \emph{relative phase} between Alice’s and Bob’s pulses determines which detector fires. Charlie only announces which detector clicked, and based on this public message, Alice and Bob can correlate their local phase or basis choices to form sifted bits~\cite{Lucamarini2018}.
	
	\paragraph{Why it helps.}
	Each successful detection round depends on \emph{single-photon} interference, and since each pulse travels only half the total distance, the secret key rate now scales as $\mathcal{O}(\sqrt{\eta})$ instead of $\mathcal{O}(\eta)$~\cite{Lucamarini2018}. The measurement device remains untrusted, as in MDI-QKD, yet the achievable distance and performance are significantly improved.
	
	\paragraph{Core mechanisms and security intuition.}
	In practice, Alice and Bob send phase-randomized weak coherent pulses to an untrusted relay, Charlie. Charlie interferes the two pulses on a balanced beam splitter and publicly announces which detector clicked. This announcement reveals no actual key information—only a classical label identifying which output port detected a photon.
	
	The security of TF-QKD relies on two main mechanisms.  
	First, the decoy-state method isolates the \emph{single-photon–effective} events within the weak coherent pulses, guaranteeing that only one quantum of information is exposed per successful detection~\cite{Hwang2003,Lo2005,Wang2005,Ma2008}.  
	Second, the interference at the untrusted relay effectively performs an \emph{entanglement-swapping test}: each single click projects Alice’s and Bob’s modes into correlated quantum states without revealing their key bits~\cite{Curty2019}.  
	Any attempt by Eve to interfere with this process disturbs the interference visibility and increases the observed QBER~\cite{Lucamarini2018,Curty2019}.  
	
	Together, these mechanisms ensure measurement-device independence and enable the unique $\mathcal{O}(\sqrt{\eta})$ rate scaling that distinguishes TF-QKD from earlier QKD protocols~\cite{Lucamarini2018,Curty2019}.
	
	\subsection{Conference Key Agreement (CKA)}
	\label{sec:cka-background}
	
	Traditional QKD creates secure keys between two users. In many real-world networks, however, a group of $N$ users needs to share a \emph{single common key} for secure conferencing or broadcast encryption. \textbf{Conference Key Agreement (CKA)} protocols extend QKD to this multi-user setting while maintaining information-theoretic security~\cite{Grasselli2019CKA,murta}.
	
	\paragraph{Architectural picture.}
	A natural setup for CKA is a \emph{star-shaped network}: each user sends optical signals to a (possibly untrusted) central node that performs interference-based measurements and publicly announces the detection results~\cite{Grasselli2019CKA}. Depending on the approach:
	\begin{itemize}
		\item \emph{Entanglement-based CKA} distributes multipartite entangled states (such as GHZ-like states) and measures locally at each node.
		\item \emph{Interference-based CKA} (inspired by TF-QKD) has all users send phase-randomized weak coherent pulses toward the center, where \emph{single-photon} detection events indicate successful rounds and define shared correlations~\cite{Grasselli2019CKA}. 
		This one is which this thesis uses, and the following discussions will be solely based on this type of CKA. 
	\end{itemize}
	
	\paragraph{Two key properties.}
	(i) The detection patterns kept at the relay effectively project the system onto a \emph{single-excitation} (W-like) subspace.  
	Correlations within this subspace are enough to generate a shared key among all users~\cite{Grasselli2019CKA}.\\
	(ii) When an untrusted relay interferes all users’ weak coherent pulses using a balanced beam-splitter network, the achievable key-rate scaling improves from the coincidence-limited $\mathcal{O}(\eta^{k})$ with $k>1$ (typical of naïve multi-user links) to $\mathcal{O}(\eta)$.
	
	\paragraph{Single-photon interference for many users.}
	Extending the idea of twin-field interference, multi-user CKA ensures that a \emph{single click} at the relay projects the joint state of all users onto a correlated subspace defined by their \emph{relative phases or bases}.  
	The relay’s announcement contains no key information—only which event or detector pattern occurred.  
	Each user then uses this public information, together with their own local settings, to reconcile an identical bit value~\cite{Grasselli2019CKA}.
	
	\paragraph{Protocol workflow (high level).}
	\begin{enumerate}
		\item \textbf{State preparation.} Each user chooses an intensity (for decoy states) and a phase or basis (for encoding), while globally randomizing the optical phase.
		\item \textbf{Transmission and detection.} All users send their pulses to the relay.  The relay publicly announces single-click or specific multi-detector events.
		\item \textbf{Sifting.} Users keep only those rounds whose combination of relay events and local settings implies a well-defined correlation, the rest are discarded.
		\item \textbf{Parameter estimation.} Using decoy-state statistics and revealed samples, the users estimate single-photon (or single-excitation) yields and error rates relevant to security~\cite{Grasselli2019CKA}.
		\item \textbf{Post-processing.} Multilateral error correction ensures all $N$ users share the same raw string.  Privacy amplification then removes any information held by Eve or the relay, producing the final \emph{conference key}.
	\end{enumerate}
	
	\paragraph{Main challenges.}
	\begin{itemize}
		\item \textbf{Scaling and sifting.} As the number of users $N$ increases, strict coincidence and basis-matching conditions reduce the sifting efficiency.  Protocols must be carefully designed to mitigate this loss.
		\item \textbf{Phase coherence.} TF-inspired CKA requires stable phase relations across many users.  In practice, reference pulses and digital phase-tracking techniques are used to maintain synchronization.
		\item \textbf{Security analysis.} Extending decoy-state and entropy-based proofs to multiple users involves handling multi-party correlations and finite-size effects with care~\cite{Grasselli2019CKA}.
	\end{itemize}
	
	\emph{Benchmarking.}  
	Recent studies compare CKA protocols against “single-message multicast” network bounds (without a relay) and show that interference-based MDI-CKA can exceed these performance limits under realistic conditions.

	\paragraph{Decoy states and estimation.}
	As in BB84, MDI, and TF-QKD, multi-intensity decoy-state analysis is used\footnote{The standard decoy-state method relies on the ``decoy assumption’’: namely, that the choice of intensities is generated by truly random, independent, and private randomness, and that pulses prepared with different intensities are otherwise indistinguishable except for their photon-number statistics~\cite{Hwang2003,Lo2005,Wang2005,Ma2008}.  
		If perfect random numbers are used, this assumption is valid.  
		However, in practice, imperfect or biased randomness may introduce correlations between signal and decoy settings, potentially leaking side information to an adversary.  
		This sensitivity to randomness is one reason why fully passive source designs, such as those developed in this thesis, are attractive, since their decoy statistics emerge naturally from post-selection rather than active modulation.}
	to estimate the \emph{single-photon–effective} contributions that ensure secrecy.
	
	In interference-based CKA, these quantities are typically expressed as \emph{yields} and \emph{phase/error} parameters conditioned on the relay’s detection patterns~\cite{Grasselli2019CKA}.
	
	\paragraph{Bridge to this thesis.}
	This family of interference-based, TF-inspired protocols provides the theoretical foundation for the \emph{fully passive CKA scheme} introduced in Chapter~\ref{chap:passive-cka}. There, I integrate single-photon interference with \emph{fully passive sources} (no active modulators) and a tailored decoy/estimation framework compatible with passive, post-selected state preparation~\cite{Grasselli2019CKA}. The result is a protocol that removes both detector vulnerabilities at the relay and modulator side-channels at the users, aligning with the implementation-security goals outlined in Sec.~\ref{sec:fully-passive}.
	
\section{A Sketch of a QKD Security \\ Proof}

While full QKD security proofs are mathematically involved, the main logic can be
outlined in a few key steps. These steps show how Alice and Bob can estimate,
bound, and finally remove any information that Eve might have gained
~\cite{Renner2008,Scarani2009,Tomamichel2012}.

\begin{enumerate}
	\item \textbf{Parameter estimation:}
	Alice and Bob use the observed data (for example, from decoy-state analysis)
	to estimate the single-photon contributions and corresponding error rates.
	This gives them quantitative bounds on how much of their raw key comes from
	secure single-photon events~\cite{Scarani2009,Tomamichel2012}.
	
	\item \textbf{Entropy-based security bounds:}
	Using uncertainty relations or entropic arguments, they determine how much
	information Eve could possibly have about the raw key. In the finite-size
	setting, this is expressed using the \emph{smooth min-entropy}
	$H_{\min}^{\varepsilon}(A|E)$, which quantifies how unpredictable Alice’s raw
	key $A$ is from Eve’s perspective~\cite{Renner2008,Tomamichel2012}. A central
	tool is the quantum leftover-hash lemma:
	\begin{equation}
		\ell \le H_{\min}^{\varepsilon}(A|E) - 2\log\!\frac{1}{2\varepsilon_{\mathrm{PA}}},
	\end{equation}
	which states that after privacy amplification the final key of length
	$\ell$ is $\varepsilon_{\mathrm{PA}}$-secure. Intuitively, the smooth
	min-entropy measures how many ``almost uniform'' bits can be extracted from
	the raw key using universal hashing.
	
	\item \textbf{Error correction and privacy amplification:}
	They then perform classical post-processing to reconcile their bit strings
	and remove any remaining information that Eve might hold. After this step,
	they obtain a final key satisfying composable security
	~\cite{Renner2008,Tomamichel2012}.
\end{enumerate}

\paragraph{Finite-size effects.}
In practical implementations, Alice and Bob exchange only a finite number of
signals, so the observed QBERs, gains, and decoy statistics fluctuate around
their true values. Modern analyses therefore use statistical bounds (such as
Chernoff or Hoeffding bounds) to obtain worst-case estimates. A typical
finite-size key-length formula is
\begin{equation}
	\ell \le n_{\mathrm{sifted}}
	\left(\,1 - I_E^{\mathrm{upper}} \right)
	- \mathrm{leak}_{\mathrm{EC}}
	- \Delta_{\mathrm{PE}}
	- \Delta_{\mathrm{PA}},
\end{equation}
where the correction terms $\Delta_{\mathrm{PE}}$ and $\Delta_{\mathrm{PA}}$
scale logarithmically with the block size~\cite{Tomamichel2012,Lim2014}.

In most analyses, the proof first considers \emph{collective attacks}, where Eve
interacts with each signal independently but may store her quantum systems for
joint measurement later. Security against full \emph{coherent attacks} (where
Eve acts jointly on all signals) then follows from symmetry arguments and
de-Finetti–type reductions~\cite{Renner2008}.

In multi-user settings, the same logic applies, but with additional care in
tracking how information may be shared among several legitimate users and
Eve. These generalizations ensure that even in complex network
scenarios, the generated key remains provably secure under the composable
security framework.


	\chapter{FULLY PASSIVE MEASUREMENT-DEVICE-INDEPENDENT QKD}
	\label{chap:passive-mdi}

\section*{Statement of Contribution}

This chapter is based on the work published in Ref.~\cite{Li2024}:
\begin{quote}
	J. Li, W. Wang, and H.-K. Lo, ``Fully passive measurement-device-independent quantum key distribution,'' \textit{Physical Review Applied}, vol. 21, p. 064056, 2024.
\end{quote}

This work was a collaboration between the author, Dr.~Wenyuan Wang, and Prof.~Hoi-Kwong Lo. The authors jointly proposed the fundamental concept of integrating fully passive sources into the MDI-QKD architecture to simultaneously eliminate detector and modulator side-channels.

The author's contributions to this chapter are as follows:

The author extended the concept of a fully passive source to the MDI-QKD framework and developed a detailed theoretical model. This modeling work included deriving a new channel model for arbitrary polarization states and proving the validity of standard decoy-state analysis for the mixed single-photon ensembles produced by passive sources in the MDI-QKD setting.

Regarding numerical evaluation, the author performed extensive simulations comparing the passive MDI-QKD protocol with its active counterpart. To improve the key-rate performance, the author implemented a ``small-ring'' decomposition optimization technique based on a concept proposed by Dr. Wang.

Finally, the author developed the necessary C++ program for these evaluations, utilizing the \texttt{CUBA} package for multi-dimensional integration and the Gurobi package for the linear programming involved in the decoy-state analysis.


\section[Measurement-Device-Independent Quantum Key Distribution (MDI-QKD)]{Measurement-Device-\\Independent Quantum Key Distribution (MDI-QKD)}
	\label{sec:mdi-qkd}
	
	In real QKD systems, detector side-channels are among the most serious security risks. Physical detectors are never ideal—they can be tricked or biased through attacks such as detector blinding, time-shifting, or exploiting detection efficiency mismatches~\cite{Lydersen2010,Makarov2006}. Even a protocol that is theoretically secure can be compromised in practice if these vulnerabilities are not addressed.  
	
	To solve this problem, Lo, Curty, and Qi proposed \emph{Measurement-Device-Independent Quantum Key Distribution} (MDI-QKD) in 2012~\cite{Lo2012}. This protocol was designed to remove \textbf{all} detector side-channels without requiring the extremely demanding conditions of fully device-independent QKD.
	
	The main idea is straightforward: move the measurement process to a central, possibly \textbf{untrusted relay}, often called Charlie. Alice and Bob, the legitimate users, each send quantum states to Charlie. He performs a Bell-state measurement (BSM) and publicly announces the result. The key’s security comes entirely from the correlations between Alice’s and Bob’s prepared states, conditioned on Charlie’s announcements. Even if Charlie is malicious and controls the detectors, he cannot learn the key without creating detectable statistical changes in the data~\cite{Lo2012}.
	
	\begin{figure}[t]
		\centering
		\includegraphics[width=0.65\linewidth]{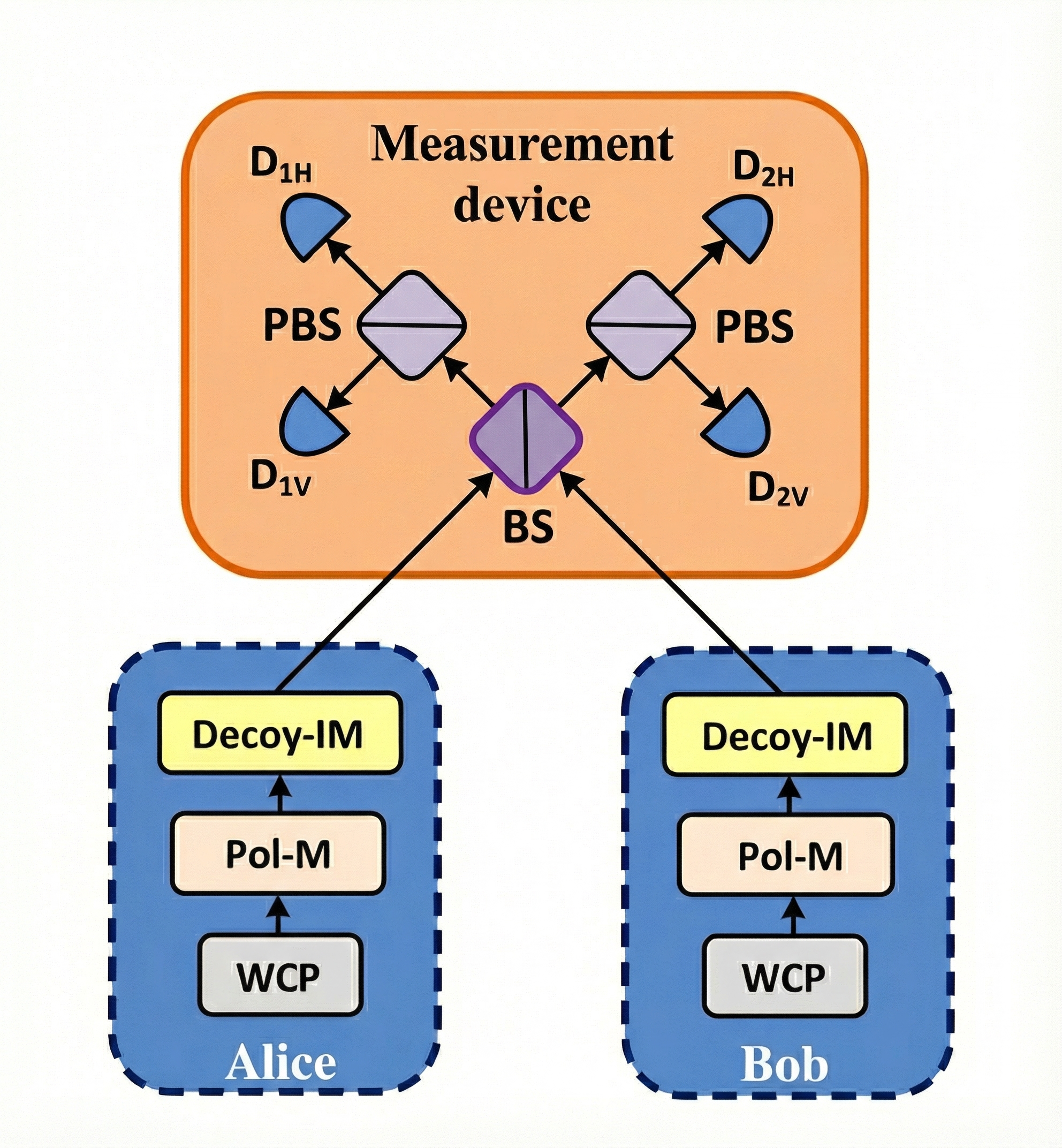}
		\caption{
			Schematic illustration of a measurement-device-independent
			quantum key distribution (MDI-QKD) implementation.
			Alice and Bob prepare phase-randomized weak coherent pulses (WCPs) in
			randomly chosen BB84 polarization states, created using a polarization
			modulator (Pol-M) and intensity settings controlled by a decoy-state
			modulator (Decoy-IM). Their optical signals are directed to an
			untrusted central relay, where they interfere on a 50:50 beam splitter.
			Each output arm is analyzed by a polarizing beam splitter (PBS) that
			resolves horizontal ($H$) and vertical ($V$) polarizations, and the
			resulting photons are detected by four single-photon detectors.  
			Coincidence events associated with orthogonal polarizations identify
			successful Bell-state measurements, enabling secure key generation even
			when the measurement device is untrusted.  
			\textit{Adapted from} Phys.\ Rev.\ Lett.\ \textbf{108}, 130503 (2012),
			DOI:\,\textit{https://doi.org/10.1103/PhysRevLett.108.130503}
			{10.1103/PhysRevLett.108.130503}.
		}
		\label{fig:mdi-qkd-setup}
	\end{figure}

	The protocol works as follows (see Fig.~\ref{fig:mdi-qkd-setup}).
	Alice and Bob each prepare phase-randomized weak coherent pulses (WCPs). In each round, they independently and randomly:
	\begin{enumerate}
		\item Choose a basis, either the rectilinear ($Z$) basis, $\{|H\rangle, |V\rangle\}$, or the diagonal ($X$) basis, $\{|+\rangle, |-\rangle\}$.
		\item Encode a bit in that basis (for example, in the $Z$ basis, $0 \to |H\rangle$ and $1 \to |V\rangle$).
		\item Select an intensity for the WCP from a set of decoy states, such as $\{\mu, \nu, \omega\}$, with $\mu$ as the signal intensity~\cite{Lo2005,Wang2005,Hwang2003}.  
		The photon number $n$ in a pulse of intensity $\mu$ follows a Poisson distribution:
		\[
		P_\mu(n) = e^{-\mu}\frac{\mu^n}{n!}.
		\]
	\end{enumerate}
	
	Both users send their pulses to Charlie, who sits between them.  
	At his station, Charlie interferes the two optical modes on a beam splitter and performs a Bell-state measurement. With linear optics, only two of the four Bell states can be distinguished, for instance, the singlet state 
	$|\psi^-\rangle = \frac{1}{\sqrt{2}}(|HV\rangle - |VH\rangle)$ and the triplet state 
	$|\psi^+\rangle = \frac{1}{\sqrt{2}}(|HV\rangle + |VH\rangle)$, identified by specific two-detector coincidence clicks~\cite{Lo2012}. Charlie then announces which Bell state was detected and at what time.  
	
	Alice and Bob keep only the rounds in which: 1) Charlie reported a successful BSM, and 2) they used the same basis. Depending on the basis and Charlie’s announcement, one of them (usually Bob) flips his bit to ensure their strings are correlated. The raw key is formed from these successful, matching events.
	
	The security of MDI-QKD can be understood in the \textbf{time-reversed entanglement picture}.  Imagine that Alice and Bob each create an entangled pair and send one photon from each pair to Charlie while keeping the other locally. When Charlie performs a successful BSM, this acts as an entanglement swap, projecting their local photons into an entangled state. From this view, MDI-QKD is equivalent to an entanglement-based BB84 protocol where the source of entanglement is untrusted but the local measurements are fully secure~\cite{Lo2012}. Because detectors lie outside the trusted regions, any side-channel attacks on them are irrelevant. The only assumption is that Alice’s and Bob’s state preparation devices are secure and well-characterized, which is a condition supported by the decoy-state method, which protects against multi-photon emissions from WCP sources~\cite{Lo2005,Wang2005}.
	
	To extract a secure key, Alice and Bob analyze their observed statistics. For each basis $b \in \{Z, X\}$ and each pair of intensities $(\mu_A, \mu_B)$, they measure the overall gain $Q^{(b)}_{\mu_A, \mu_B}$ (the probability of a successful BSM) and the quantum bit error rate (QBER) $E^{(b)}_{\mu_A, \mu_B}$. These observables are linked to the photon-number-resolved yields $Y^{(b)}_{nm}$ and error rates $e^{(b)}_{nm}$ through:
	\begin{align}
		Q^{(b)}_{\mu_A,\mu_B}
		&= \sum_{n,m=0}^{\infty} P_{\mu_A}(n)P_{\mu_B}(m)\,Y^{(b)}_{nm}, \label{eq:mdi-gain}\\[4pt]
		Q^{(b)}_{\mu_A,\mu_B}E^{(b)}_{\mu_A,\mu_B}
		&= \sum_{n,m=0}^{\infty} P_{\mu_A}(n)P_{\mu_B}(m)\,Y^{(b)}_{nm}e^{(b)}_{nm}. \label{eq:mdi-qber}
	\end{align}
	
	Here, $Y^{(b)}_{nm}$ is the probability that a successful BSM occurs when Alice sends $n$ photons and Bob sends $m$ photons in basis $b$. By comparing data from multiple decoy intensities, Alice and Bob can bound the single-photon yield $Y^{(Z)}_{11}$ and the single-photon phase-error rate $e^{(X)}_{11}$~\cite{Xu2013}. The $Z$ basis provides the raw key, while the $X$ basis is used for parameter estimation to limit Eve’s possible information.
	
	The asymptotic key rate per pulse pair is then:
	\begin{equation}
		R = Q^{(Z)}_{11}\!\left[1-H_2\!\big(e^{(X)}_{11}\big)\right]
		- Q^{(Z)}_{\text{obs}} f_\text{EC}  H_2\!\big(E^{(Z)}_{\text{obs}}\big),
		\label{eq:mdi-keyrate}
	\end{equation}
	where $Q^{(Z)}_{11} = P_\mu(1)P_\mu(1)Y^{(Z)}_{11}$ is the single-photon gain in the $Z$ basis, and $e^{(X)}_{11}$ is its corresponding error rate in the $X$ basis, both obtained from decoy analysis. The binary Shannon entropy is $H_2(x) = -x\log_2(x) - (1-x)\log_2(1-x)$. The first term in Eq.~\eqref{eq:mdi-keyrate} represents the net generation of secure bits after privacy amplification, while the second term quantifies the cost of error correction, where $f_\text{EC} \ge 1$ accounts for practical inefficiency~\cite{Lo2012}.
	\medskip
	
	It is natural to extend MDI-QKD by combining it with passive QKD ideas, so that both detector and source side-channels are eliminated simultaneously. This motivation leads directly to the development of the \emph{fully passive MDI-QKD} protocol, which forms a key contribution of this thesis~\cite{Li2024}.
	
	\section{Motivation and Overview}
	\label{sec:overview}
	
	MDI-QKD has become a key step toward building secure and practical QKD networks. By design, MDI-QKD eliminates all detector side-channels, closing one of the most critical loopholes in real implementations~\cite{Lo2012}. However, most current systems remain \emph{active} on the transmitter side, relying on optical modulators to prepare quantum states. These modulators are potential weak points: they can leak information through side-channels such as Trojan-horse probes or pattern-dependent emissions, which compromise the security of the source~\cite{Gisin2006,Yoshino2018}.
	
	To close this remaining loophole, I introduce and analyze a \emph{fully passive} MDI-QKD protocol~\cite{Li2024}. In this design, both Alice and Bob generate their quantum states through passive optical interference, without using any active modulators. As a result, the scheme simultaneously removes vulnerabilities on both the detector and transmitter sides, offering a significantly higher level of implementation security.
	
	\subsection{Fully Passive Sources and Detection}
	\label{subsec:passivesource}

	\begin{figure}[H]
		\centering
		\includegraphics[width=0.78\linewidth]{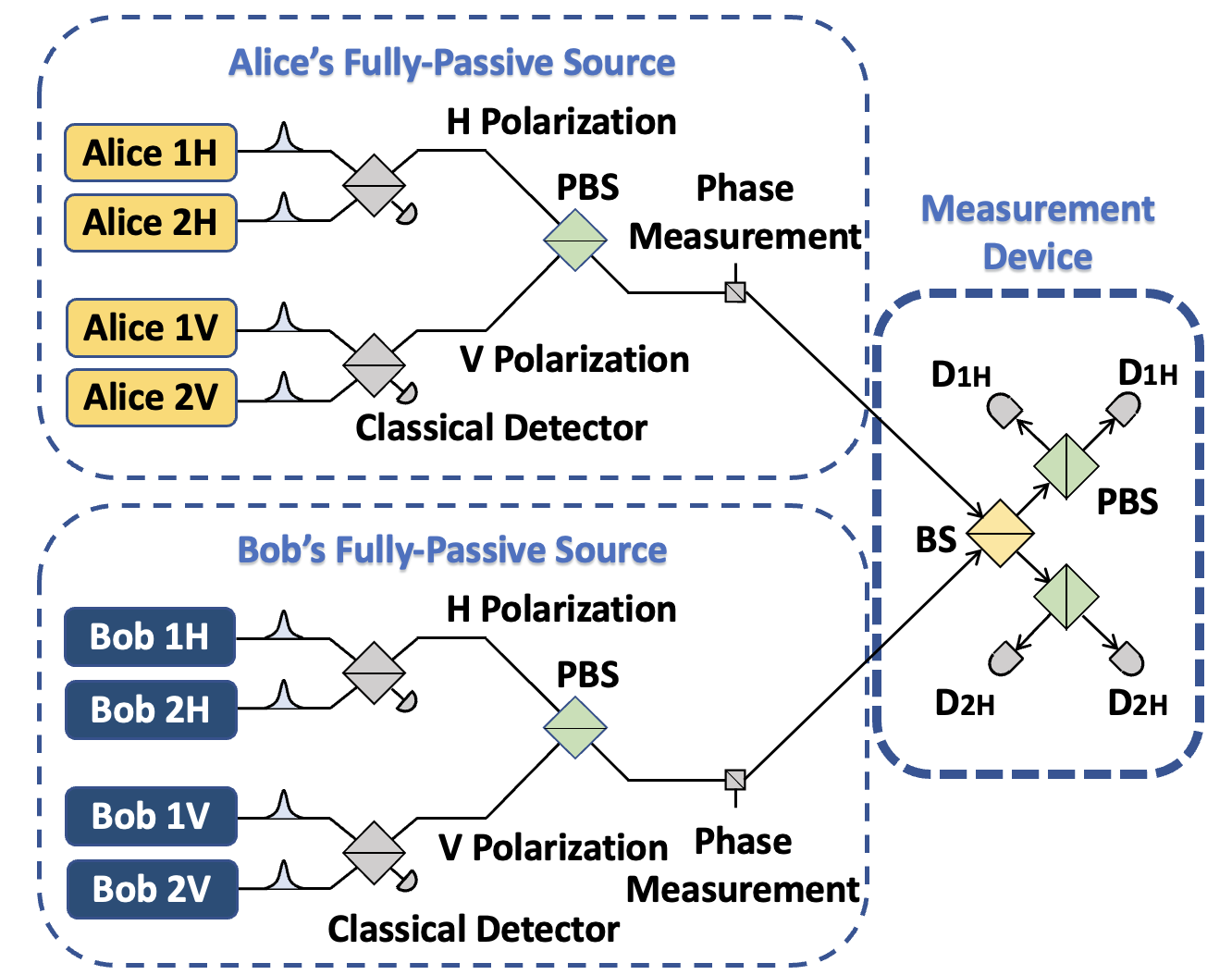}
		\caption{
			Optical layout of the fully passive source used in my passive MDI-QKD scheme.  
			Each user employs four independent lasers whose random phases are mixed through
			a network of 50{:}50 beam splitters (BSs) and a polarizing beam splitter (PBS),
			producing orthogonally polarized arms.  
			Local classical detectors monitor the arm intensities and phases, and the
			resulting classical data $(\mu_H,\mu_V,\phi_H,\phi_V)$ are used to infer the
			output state parameters for basis and decoy selection~\cite{Wang2023a,Li2024}.  
			\textit{Figure Adapted with modifications from Fig.~1 in
				W.~Wang \emph{et al.}, Phys.\ Rev.\ Lett.\ \textbf{130}, 220801 (2023),
				and subsequently used in my Phys.\ Rev.\ Applied manuscript, with
				permission of the American Physical Society.}
		}
		\label{fig:passivesource}
	\end{figure}

The fully passive source, which forms the core of my protocol, is shown schematically in Fig.~\ref{fig:passivesource}, following designs in Refs.~\cite{Wang2023a,Li2024}.

	Both Alice and Bob use identical setups to carry out passive MDI-QKD. \linebreak Each source consists of four independent lasers whose optical phases, $\phi_1, \phi_2$, $\phi_3, \phi_4$, are completely random and uniformly distributed. The laser outputs are combined pairwise at two 50:50 beam splitters (BSs), producing two orthogonally polarized arms, labeled $H$ and $V$. These arms are then combined at a polarizing beam splitter (PBS) to form the final output signal. The polarization of this output state is determined entirely by the interference process and is therefore \emph{passively} chosen.
	
	Each user locally monitors the intensity and relative phase of the $H$ and $V$ arms using classical photodetectors. These measurements yield classical quantities $(\mu_H, \mu_V, \phi_H, \phi_V)$, which map directly to the Bloch-sphere coordinates $(\mu, \theta_{HV}, \phi_{HV})$ and a global phase through the following relations~\cite{Wang2023a}:
	\begin{align}
		\theta_{HV} &= 2\cos^{-1}\!\left(\sqrt{\frac{\mu_H}{\mu_H+\mu_V}}\right), \label{eq:theta_hv} \\
		\phi_{HV} &= \phi_V - \phi_H. \label{eq:phi_hv}
	\end{align}
	Thus, the four initial random phases are transformed through the optical network as~\cite{Wang2023a}:
	\begin{equation}
		\phi_1, \phi_2, \phi_3, \phi_4
		\;\longrightarrow\;
		\mu_H, \mu_V, \phi_H, \phi_V
		\;\longrightarrow\;
		\mu, \theta_{HV}, \phi_{HV}, \phi_{\mathrm{global}}.
	\end{equation}
	Since all four initial phases are random, the resulting output can represent any point on the Bloch sphere~\cite{Wang2023a}. The corresponding polarization state can be expressed as
	\begin{equation}
		|\psi\rangle =
		\cos\!\frac{\theta_{HV}}{2}\,|H\rangle
		+ e^{i\phi_{HV}} \sin\!\frac{\theta_{HV}}{2}\,|V\rangle,
	\end{equation}
	where $(\theta_{HV}, \phi_{HV})$ locate the state on the sphere and $\mu = \mu_H + \mu_V$ denotes the total optical intensity.

	Because the state parameters are generated entirely by random interference, no active modulation is used to prepare the states. Instead, users record the outcomes of their local monitoring and \emph{post-select} those signals whose polarization and intensity fall within desired regions (for instance, states corresponding to the $Z$ or $X$ basis, and specific decoy intensities). This purely passive process handles both state preparation and decoy-state randomization without using any modulators, effectively removing a major class of source side-channels.
	
	Conceptually, the fully passive source merges the principles of passive decoy-state and passive encoding schemes. In a passive decoy-state setup, two signals interfere at a 50:50 BS, and the output intensity is determined by their random phase difference~\cite{Curty2010}. In a passive encoding setup, a PBS ensures that the output polarization depends on the phase difference between two input arms~\cite{Curty2010}. By combining both approaches, the fully passive source generates output signals whose intensity and polarization are jointly determined by the random phase relations among the four lasers. These signals can then be selectively used for QKD, providing a natural and secure way to integrate passive state generation within the MDI framework.

	\subsection{Postselection}
	\label{subsec:postselection}

	\begin{figure}[H]
		\centering
		\includegraphics[width=0.82\linewidth]{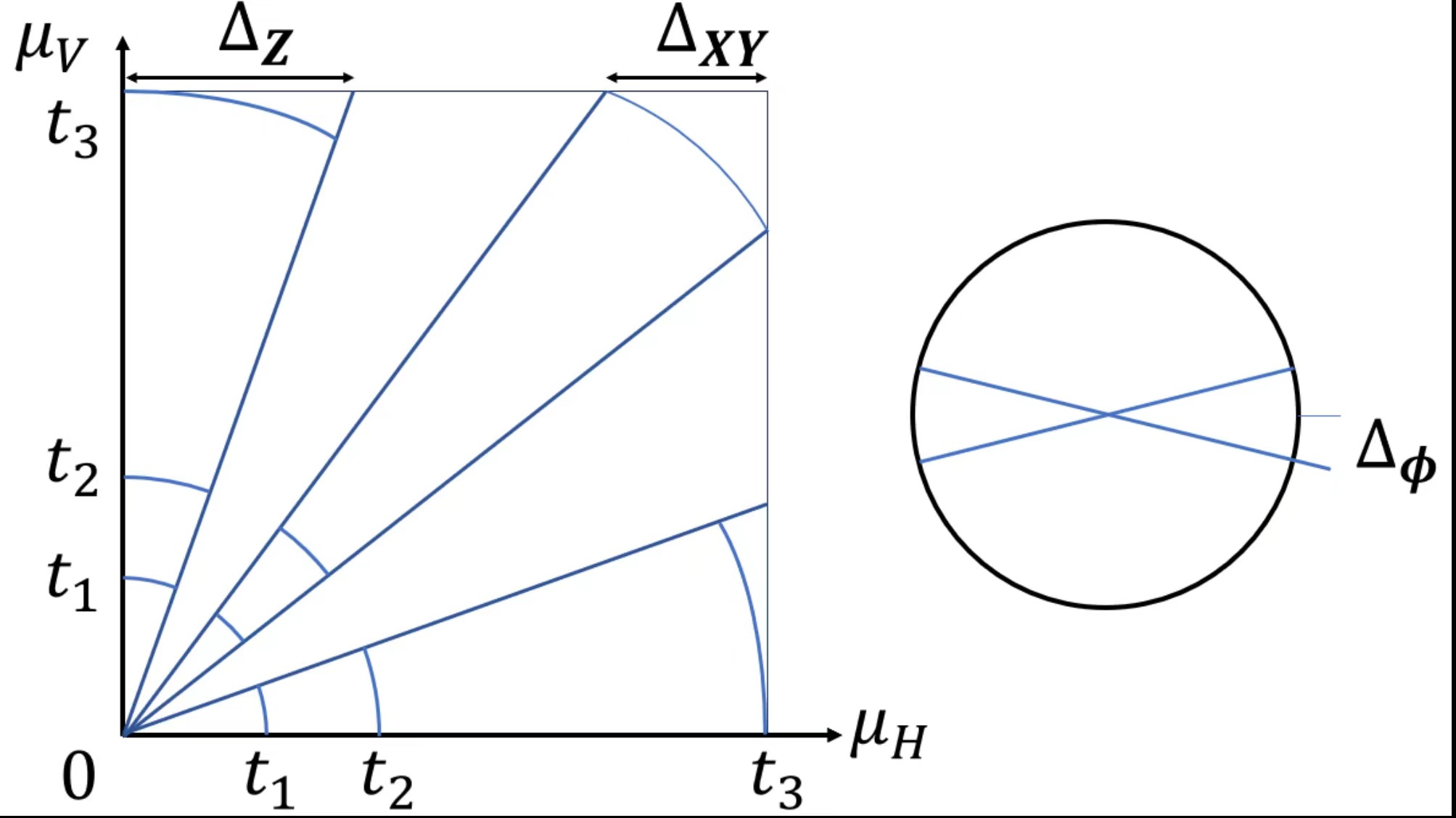}
		\caption{
			Illustration of the postselection domains used by Alice,
			represented in the $(\mu_H, \mu_V, \phi)$ parameter space together
			with their corresponding regions on the Bloch sphere. The basis-dependent partitions including the $Z$-basis window
			$\Delta_Z$ and the $X$-basis windows $\Delta_{XY}$ and $\Delta_{\phi}$ are
			further subdivided into several intensity bins, labeled
			$t_1, t_2, t_3$, which define the decoy settings employed in the
			protocol. Bob applies an identical partitioning.  
			\textit{Figure Adapted with modifications from Fig.~2 of
				W.~Wang \emph{et al.}, Phys.\ Rev.\ Lett.\ \textbf{130}, 220801 (2023),
				and subsequently incorporated into my Phys.\ Rev.\ Applied manuscript,
				with permission of the American Physical Society.}
		}
		\label{fig:regions}
	\end{figure}

	The output signals from the fully passive sources can, in principle, occupy any point on the Bloch sphere. To prepare usable states for the protocol, each user defines specific regions on the sphere corresponding to the desired BB84 states, such as $\{|H\rangle, |V\rangle, |+\rangle, |-\rangle\}$, as illustrated in Fig.~\ref{fig:regions}.

	Alice and Bob then post-select only those signals whose locally measured parameters fall within these predefined regions. These regions can also be subdivided to carry out the decoy-state analysis. As a result, any physical observable is described by its expectation value averaged over the selected regions.
	
	Alice and Bob then post-select only those signals whose locally measured parameters fall within these predefined regions. These regions can also be subdivided for decoy-state analysis, and any physical observable is evaluated as the expectation value averaged over the selected regions. All required multi-dimensional integrals, as can be seen in Eq. \ref{eq:exp7d},  are computed using the high-efficiency \texttt{CUBA} numerical integration library within my C++ implementation, ensuring fast and stable evaluation of these region-averaged quantities.

	For an arbitrary observable $A$, the expectation value within a given region pair—$S_i$ for Alice and $S_j$ for Bob—is calculated through a seven-dimensional integral~\cite{Li2024}:
\begin{align}
	\langle A \rangle_{S_i S_j}
	&= \frac{1}{2 P_{S_i S_j}}
	\int_{S_i} \int_{S_j} \int_0^{2\pi}
	p_A(\mu_{HA},\mu_{VA},\phi_{HVA})\,p_B(\mu_{HB},\mu_{VB},\phi_{HVB}) \notag\\
	&\quad \times
	A(\mu_{HA},\mu_{VA},\phi_{HVA},\mu_{HB},\mu_{VB},\phi_{HVB})\notag\\
	&\quad\times d\mu_{HA}\,d\mu_{VA}\,d\phi_{HVA}\,d\mu_{HB}\,d\mu_{VB}\,d\phi_{HVB}\,d\phi'_R,
	\label{eq:exp7d}
\end{align}
	where $P_{S_i S_j}$ is the joint probability that both users generate signals in regions $S_i$ and $S_j$:
	\begin{align}
		P_{S_i S_j}
		&=
		\int_{S_i}\int_{S_j}
		p_A(\mu_{HA},\mu_{VA},\phi_{HVA})
		p_B(\mu_{HB},\mu_{VB},\phi_{HVB})\notag\\&\quad\times
		\,d\mu_{HA}\,d\mu_{VA}\,d\phi_{HVA}\,d\mu_{HB}\,d\mu_{VB}\,d\phi_{HVB}.
	\end{align}
	The integral over $\phi'_R$ represents averaging over the random global phase difference between Alice’s and Bob’s sources. For Alice, the classical probability distribution $p_A$ can be written as a product of intensity and phase terms~\cite{Wang2023a}:
	\begin{align}
		p_A(\mu_{HA},\mu_{VA},\phi_{HVA})
		&= p_{\mu}^A(\mu_{HA},\mu_{VA})\,p_{\phi}^A(\phi_{HVA}),
		\label{eq:p_factorization}
	\end{align}
	where
	\begin{align}
		p_{\mu}^A(\mu_{HA},\mu_{VA})
		&= \frac{1}{\pi^2}
		\frac{1}{\sqrt{\mu_{HA}(\mu_{\max}-\mu_{HA})}}
		\frac{1}{\sqrt{\mu_{VA}(\mu_{\max}-\mu_{VA})}},
		\label{eq:pmu}\\[4pt]
		p_{\phi}^A(\phi_{HVA}) &= \frac{1}{2\pi}.
		\label{eq:pphi}
	\end{align}
	Bob’s distribution has the same form. These high-dimensional integrals are evaluated numerically using the \textsc{CUBA} library~\cite{Hahn2005}.
	
	\vspace{0.5em}
	\noindent\textbf{Intensity-Dependent Postselection.} To perform a valid decoy-state analysis, Alice and Bob apply an additional, intensity-dependent postselection step by probabilistically discarding some signals. This process is described by a function $q_{\mu}(\mu_{H(A/B)},\mu_{V(A/B)})$, chosen to shape the resulting intensity distribution. The function is defined as~\cite{Wang2023a}
\begin{align}
	q_{\mu}(\mu_{H(A/B)},\mu_{V(A/B)})
	=& C\pi^2
	\sqrt{\mu_{H(A/B)}[\mu_{\max}-\mu_{H(A/B)}]}
	\notag\\&\times
	\sqrt{\mu_{V(A/B)}[\mu_{\max}-\mu_{V(A/B)}]}
	e^{[\mu_{H(A/B)}+\mu_{V(A/B)}]},
	\label{eq:qmu}
\end{align}
	so that the overall intensity distribution takes a simple exponential form:
	\begin{align}
		p'_{\mu} = C\, e^{(\mu_H+\mu_V)}.
		\label{eq:expdist}
	\end{align}
	This “modulated” postselection step is crucial because it decouples the photon-number statistics from the polarization-dependent yields, a requirement for the linear programming used in the decoy-state analysis.
	
	\subsection{Channel Model}
	\label{subsec:channelmodel}
	
	In standard, actively modulated MDI-QKD, the prepared states usually lie on a two-dimensional plane of the Bloch sphere. In contrast, the signals generated by the fully passive scheme are distributed across the entire three-dimensional sphere. To accurately capture their transmission and interference, a complete 3D polarization model is therefore required~\cite{Li2024}.
	
	Each user's emitted state can be represented by a Bloch vector $\vec{s}_{U}$, where $U \in \{A, B\}$ refers to Alice and Bob, respectively. As each signal propagates through the optical fiber, it can undergo polarization rotation due to birefringence or small misalignments. This rotation is described by the Rodrigues formula~\cite{Friedberg2022}:
	\begin{align}
		\vec{s'}_{U} =
		\cos\gamma_U\,\vec{s}_U
		+ \sin\gamma_U\,(\vec{n}_U\times\vec{s}_U)
		+ (1-\cos\gamma_U)(\vec{n}_U\cdot\vec{s}_U)\vec{n}_U,
		\label{eq:rodrigues}
	\end{align}
	where $\vec{n}_U$ is the unit vector defining the rotation axis and $\gamma_U$ is the rotation angle. The resulting vector $\vec{s'}_{U}$ represents the polarization of the signal upon reaching the central relay, Charlie.
	
	At Charlie’s station, the two incoming signals interfere at a 50:50 beam splitter (BS). The interference is analyzed separately for the horizontal ($H$) and vertical ($V$) components. The rotated states are first decomposed back into their $H$- and $V$-leg intensities (e.g., $\mu'_{HA}$ and $\mu'_{VA}$ for Alice), which are then reduced according to the channel loss and detector efficiency.
	
	When two coherent states, $|\sqrt{\mu_{1}}e^{i\phi_{1}}\rangle$ and $|\sqrt{\mu_{2}}e^{i\phi_{2}}\rangle$, interfere at a beam splitter, the output port intensities are given by:
	\begin{align}
		\mu_{c} &= \frac{\mu_{1}}{2} + \frac{\mu_{2}}{2} - \sqrt{\mu_{1}\mu_{2}}\sin(\Delta\phi), \\
		\mu_{d} &= \frac{\mu_{1}}{2} + \frac{\mu_{2}}{2} + \sqrt{\mu_{1}\mu_{2}}\sin(\Delta\phi),
	\end{align}
	where $\Delta\phi$ is the phase difference between the two fields. This expression is applied independently to the $H$-leg intensities $(\mu'_{HA}, \mu'_{HB})$ and the $V$-leg intensities $(\mu'_{VA}, \mu'_{VB})$, determining the final light intensities incident on Charlie’s four single-photon detectors.
	
	The overall gain and quantum bit error rate (QBER) for a given pair of input states are obtained by calculating the resulting detection probabilities and averaging over the random relative global phase $\phi'_R$ between Alice and Bob. Finally, these quantities are integrated over the postselection regions, as defined in Eq.~(\ref{eq:exp7d}), to yield the parameters used in the decoy-state analysis.
	
A more detailed derivation of the channel model, including the Bloch-sphere rotations, beam-splitter interference formulas, and phase averaging, is collected in Appendix~\ref{sec:mdi_channel}.

	\section{Decoy-State Analysis}
	\label{sec:decoy}
	
	The decoy-state method is a key technique for securely estimating the single-photon yield ($Y_{11}$) and error rate ($e_{11}$) when weak coherent pulse (WCP) sources are used~\cite{Lo2005}. In the passive MDI-QKD protocol, however, the emitted signals form \emph{mixed} ensembles because the underlying optical parameters are random and continuous. This section outlines the theoretical framework that allows standard decoy-state analysis to be applied to such mixed ensembles. Detailed derivations, including the full linear-program construction, yield bounds, and error-yield bounds, are provided in Appendices~\ref{sec:mdi_lp}--\ref{sec:mdi_error}.

	\subsection{Decoupling the Photon-Number Yields}
	\label{subsec:decoupling}
	
	In a standard decoy-state analysis, the observed gain and QBER are expressed as linear sums over photon-number components. In the passive scheme, these quantities are intertwined within the seven-dimensional integral defined in Eq.~(\ref{eq:exp7d}). The experimentally observed gain $\langle Q\rangle_{S_i S_j}$ and error-gain $\langle QE\rangle_{S_i S_j}$ for a pair of postselection regions $(S_i, S_j)$ are written as:
	\begin{align}
		\langle Q\rangle_{S_i S_j}
		&= \sum_{n,m=0}^{\infty} \langle P_n^A P_m^B Y_{nm} \rangle_{S_i S_j}, \\
		\langle QE\rangle_{S_i S_j}
		&= \sum_{n,m=0}^{\infty} \langle P_n^A P_m^B e_{nm} Y_{nm} \rangle_{S_i S_j},
	\end{align}
	where the averages are taken over the joint probability distribution of the local optical parameters. To make this system solvable, the photon-number probabilities must be separated from the yields~\cite{Wang2023a}.
	
	This is achieved by using the intensity-dependent postselection function $q_{\mu}$ from Eq.~(\ref{eq:qmu}). This function reshapes the underlying intensity distribution into an exponential form, $p'_{\mu} \propto e^{\mu_H+\mu_V}$, effectively canceling the $e^{-(\mu_H+\mu_V)}$ factor in the Poisson photon-number probabilities. This key step \emph{decouples} the probabilities from the yields, allowing the equations to be written in a factorized form~\cite{Wang2023a,Li2024}:
	\begin{align}
		\langle Q\rangle_{S_i S_j}
		&= \sum_{n,m=0}^{\infty} \langle P_n^A \rangle_{S_i} \langle P_m^B \rangle_{S_j} Y_{nm}^{\mathrm{mix}}, \label{eq:gain_decoupled} \\
		\langle QE\rangle_{S_i S_j}
		&= \sum_{n,m=0}^{\infty} \langle P_n^A \rangle_{S_i} \langle P_m^B \rangle_{S_j} e_{nm}^{\mathrm{mix}} Y_{nm}^{\mathrm{mix}}, \label{eq:errgain_decoupled}
	\end{align}
	where $Y_{nm}^{\mathrm{mix}}$ and $e_{nm}^{\mathrm{mix}}$ represent the effective yields and error rates for the mixed photon-number components, averaged over the polarization degrees of freedom. A step-by-step derivation of Eqs.~\eqref{eq:gain_decoupled}--\eqref{eq:errgain_decoupled} and the explicit construction of the mixed yields $Y_{nm}^{\mathrm{mix}}$ can be found in Appendix~\ref{sec:mdi_lp}.

	\subsection{Linear Programming and Bound Estimation}
	\label{subsec:lp}

	While finding an explicit analytical expression for the bounds is difficult given the number of free parameters, the system is in a form that is perfectly suited for numerical optimization. To this end, standard linear programming (LP) is a reasonably effective and widely used tool for this task~\cite{Ma2008}. Once the observables are written in the linear form of Eqs.~(\ref{eq:gain_decoupled})--(\ref{eq:errgain_decoupled}), the system for different decoy-state settings $(S_i, S_j)$ can be solved using standard linear programming (LP). Alice and Bob use their measured gains and QBERs as input to the LP to obtain conservative bounds on the single-photon parameters. The optimization produces two key results:
	\begin{itemize}
		\item A lower bound on the single-photon yield, $Y_{11}^{\mathrm{mix, L}}$.
		\item An upper bound on the single-photon error yield, $(e_{11}Y_{11})^{\mathrm{mix, U}}$.
	\end{itemize} From these, the upper bound on the single-photon error rate, $e_{11}^{\mathrm{mix, U}}$, can be derived directly.
	
	\subsection{Mapping Mixed Ensembles to Ideal States}
	\label{subsec:mixed2ideal}
	
	A main challenge in proving security for the passive protocol is to connect the real states produced by the source to the idealized states assumed in standard QKD proofs. The linear program, being data-driven, provides bounds on the actual physical system that generates the measured statistics. Since the passive source emits an ensemble of states over finite postselection regions, the LP naturally provides bounds for a \textbf{mixed ensemble} (not perfect, as in the QKD proofs).
	
	However, the MDI-QKD key rate formula assumes idealized \textbf{perfect ensembles}, corresponding to pure BB84 states such as $|H\rangle$ and $|V\rangle$. To apply this formula safely, it is necessary to show that the bounds derived for the mixed ensemble are still valid (and conservative) for the perfect ensemble. This mapping is established through two crucial relationships~\cite{Wang2023a,Li2024}:
	
	\begin{enumerate}
		\item \textbf{Yield Equality:}  
		The yield of the mixed single-photon ensemble is shown to be \emph{identical} to that of the ideal single-photon state:
		\begin{equation}
			Y_{11}^{\mathrm{mix}} = Y_{11}^{\mathrm{perfect}}.
		\end{equation}
		This ensures that the lower bound $Y_{11}^{\mathrm{mix, L}}$ can be safely used as the effective single-photon yield in the standard key-rate formula.
		
		\item \textbf{Error Rate Conservativeness:}  
		The error rate of the mixed ensemble is always greater than or equal to that of the perfect state:
		\begin{equation}
			e_{11}^{\mathrm{mix}} \ge e_{11}^{\mathrm{perfect}}.
		\end{equation}
		This conservative inequality guarantees that using the upper bound $e_{11}^{\mathrm{mix, U}}$ in the key-rate calculation cannot overestimate the final key length.  
		In practice, this means the protocol may sacrifice some key rate, but never security.
	\end{enumerate}
	
Rigorous proofs of the yield equality and the conservative error-yield inequality are given in Appendix~\ref{sec:mdi_yield} and Appendix~\ref{sec:mdi_error}, respectively.

	\subsection{Key Rate}
	\label{subsec:keyrate}
	
	In this protocol, the $Z$ basis is used for key generation, while the $X$ basis is used to estimate the phase error and detect any eavesdropping. The final asymptotic secret key rate per pulse pair is derived from the bounds obtained through the decoy-state analysis~\cite{Li2024,Lo2012}.
	
	Let $P_Z^{A}$ and $P_Z^{B}$ denote the probabilities that Alice and Bob post-select a signal in the $Z$-basis region. Let $\langle P_1^A\rangle_{S_Z}$ and $\langle P_1^B\rangle_{S_Z}$ represent the average single-photon Poisson weights over these regions. Using the lower and upper bounds from the mixed ensemble analysis—together with the mapping between mixed and ideal ensembles—a conservative expression for the asymptotic key rate per pulse pair is:
	\begin{equation}
		R =
		P_Z^A P_Z^B
		\Big\{
		\langle P_1^A\rangle_{S_Z}
		\langle P_1^B\rangle_{S_Z}
		Y_{11}^{Z, \mathrm{mix, L}}
		\big[1 - H_2(e_{11}^{X, \mathrm{mix, U}})\big]
		-
		f_\text{EC}H_2(E_Z)
		\Big\},
		\label{eq:keyrate-final}
	\end{equation}
	where $H_2(x)$ is the binary Shannon entropy function and $f_\text{EC} \ge 1$ denotes the inefficiency of the classical error-correction algorithm.
	
	Each term in Eq.~\eqref{eq:keyrate-final} has a specific meaning:
	\begin{itemize}
		\item $Y_{11}^{Z, \mathrm{mix, L}}$ — lower bound on the single-photon yield in the $Z$ basis, obtained from the linear program.
		\item $e_{11}^{X, \mathrm{mix, U}}$ — upper bound on the single-photon error rate in the $X$ basis, used to estimate Eve’s potential information.
		\item $\langle Q_Z\rangle$ and $E_Z$ — observed $Z$-basis gain and quantum bit error rate (QBER), representing the rate of successful detections and the cost of reconciliation.  
		The QBER is calculated as $E_Z = \langle Q_Z E_Z \rangle / \langle Q_Z \rangle$.
		\item $f_\text{EC}$ — efficiency factor of the practical error-correction process, accounting for its deviation from the Shannon limit.
	\end{itemize}
	
	Because of the mixed-to-ideal mapping derived earlier, substituting the mixed bounds into this expression provides a valid and conservative estimate of the secure key rate for the ideal single-photon states.
	
	\paragraph{Remark (Small-Ring Sifting Method).}
	The key rate can be slightly improved by introducing a more refined sifting approach, called the \textbf{small-ring method} (see Fig.~\ref{fig:ringmethod}). Instead of treating the entire $Z$-basis region as one bin, it is divided into several narrow concentric rings, $\{\mathcal{R}_k\}$, based on the polar angle. A separate key rate is calculated for each ring pair (one from Alice, one from Bob), and the final key rate is obtained by summing over all pairs~\cite{Li2024}.
	
	The privacy amplification term is proportional to the number of single-photon detections and remains unaffected by this partitioning. However, the error-correction term involves the binary entropy function $H_2(E)$, which is \textbf{convex}. According to Jensen’s inequality~\cite{Jensen1906}, the average of $H_2(E)$ over all rings satisfies $\langle H_2(E) \rangle \le H_2(\langle E \rangle)$. This means that dividing the data into smaller rings reduces the effective reconciliation cost, leading to a modest improvement in the final key rate.
	
	\begin{figure}[H]
		\centering
		\includegraphics[width=0.85\linewidth]{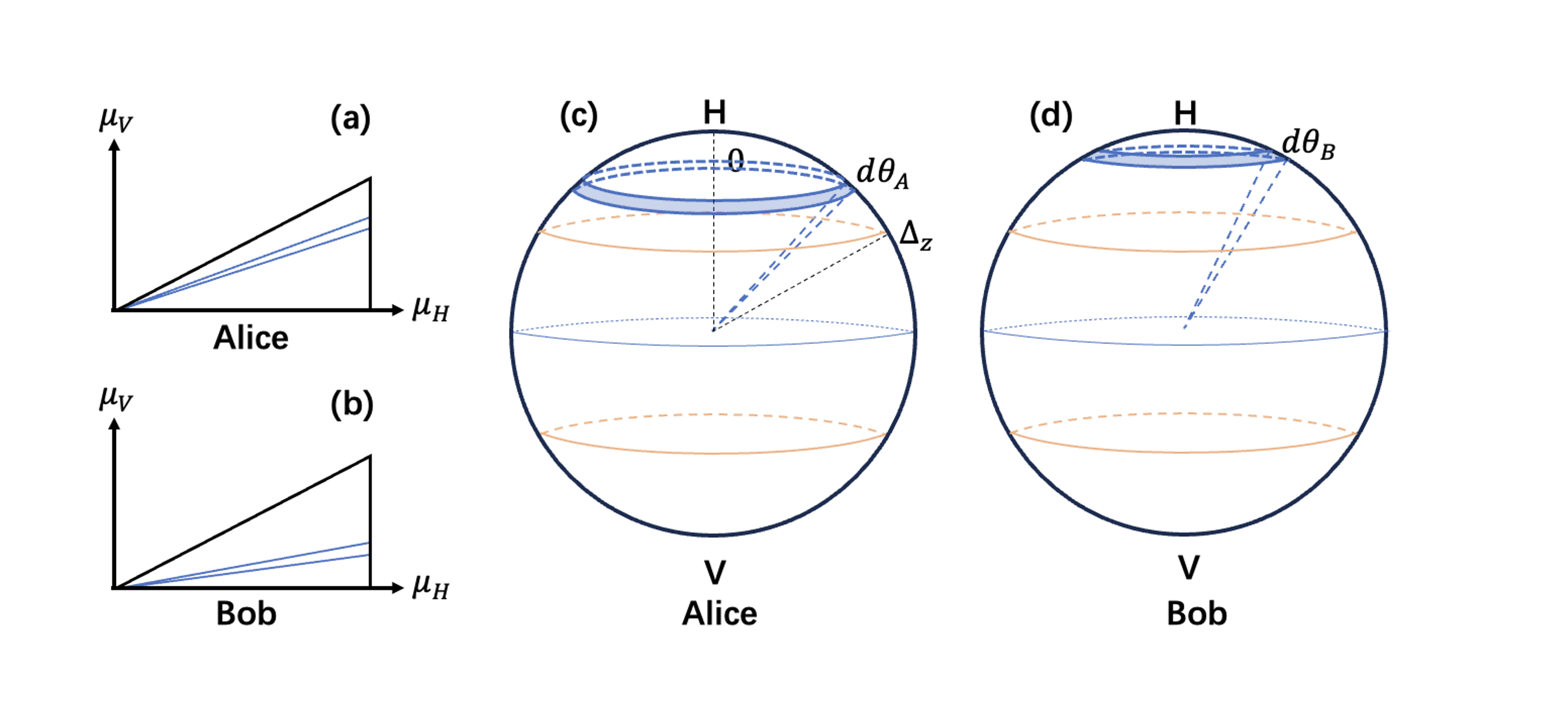}
		\caption{
			Illustration of the small-ring sifting strategy.  
			The key-generating regions associated with the $Z$ basis are divided
			into a sequence of thin concentric rings, indicated by the shaded
			segments between successive angular boundaries.  
			Panels~(a) and~(b) show how these rings appear in the
			$(\mu_H,\mu_V)$ parameter space for Alice and Bob, while
			panels~(c) and~(d) display the corresponding annular regions on their
			respective Bloch spheres.  
			Each ring is labeled by its polar angle~$\theta$ and has a differential
			width $d\theta$.  
			In the numerical evaluation of Eq.~(\ref{eq:exp7d}), Alice’s integration
			limits are marked as $0$ and~$\Delta_Z$, and Bob’s limits are identical.  
			\textit{Figure reproduced from my Phys.\ Rev.\ Applied manuscript
				\cite{Li2024}.}
		}
		\label{fig:ringmethod}
	\end{figure}

	\section{Results and Discussion}
	\label{sec:results}
	
	\subsection{Numerical Results}
	
	I numerically calculated and optimized the asymptotic key rate of the fully passive MDI-QKD protocol and compared it with a conventional active MDI-QKD scheme. The channel was modeled with a loss coefficient of $0.2~\mathrm{dB/km}$, a detector dark count probability of $10^{-6}$, and a total basis misalignment of $1\%$ ($0.5\%$ for each user in the $X$–$Z$ plane). Both protocols used three decoy settings. For the passive scheme, the postselection parameters were set to $\Delta_{XY}=\Delta_{\phi}=0.2$, while $\Delta_Z$ was optimized over $[0.01,0.05]$ and the decoy threshold $t_3$ over $[0.1,0.99]$. For the active scheme, the three decoy intensities were optimized numerically within the range $[0,1]$~\cite{Xu2013}.

	All numerical simulations in Fig.~\ref{fig:keyrate} were performed using a custom \texttt{C++} codebase. High-dimensional integrations (up to seven dimensions) required for evaluating the passive MDI-QKD channel model were computed using the \texttt{CUBA} numerical integration library~\cite{Hahn2005}, following the same methodology used in my published work~\cite{Li2024}. The decoy-state bounds for $Y_{11}$ and $e_{11}Y_{11}$ were obtained by solving the associated linear programs with the commercial solver \texttt{Gurobi}. A typical full key-rate evaluation at a single distance (including optimization over $\Delta_Z$ and $t_3$) requires several seconds to a few minutes on a standard desktop workstation, depending on the slice sizes and integration tolerances. These computational tools allow efficient and stable evaluation of the passive MDI-QKD performance over the full distance range shown.

	Figure~\ref{fig:keyrate} presents the simulated key rates. Clearly, it shows that the fully passive scheme achieves a maximum communication distance of approximately $143~\mathrm{km}$. Its key rate is two to three orders of magnitude lower than that of the active MDI-QKD protocol~\cite{Li2024}. This reduction is expected—it arises from the double-sifting loss due to postselection at both Alice’s and Bob’s sources. In return, the protocol eliminates all side-channels associated with active source modulators, offering a stronger and more robust implementation security.
	
\begin{figure}[H]
	\centering
	\includegraphics[width=0.82\linewidth]{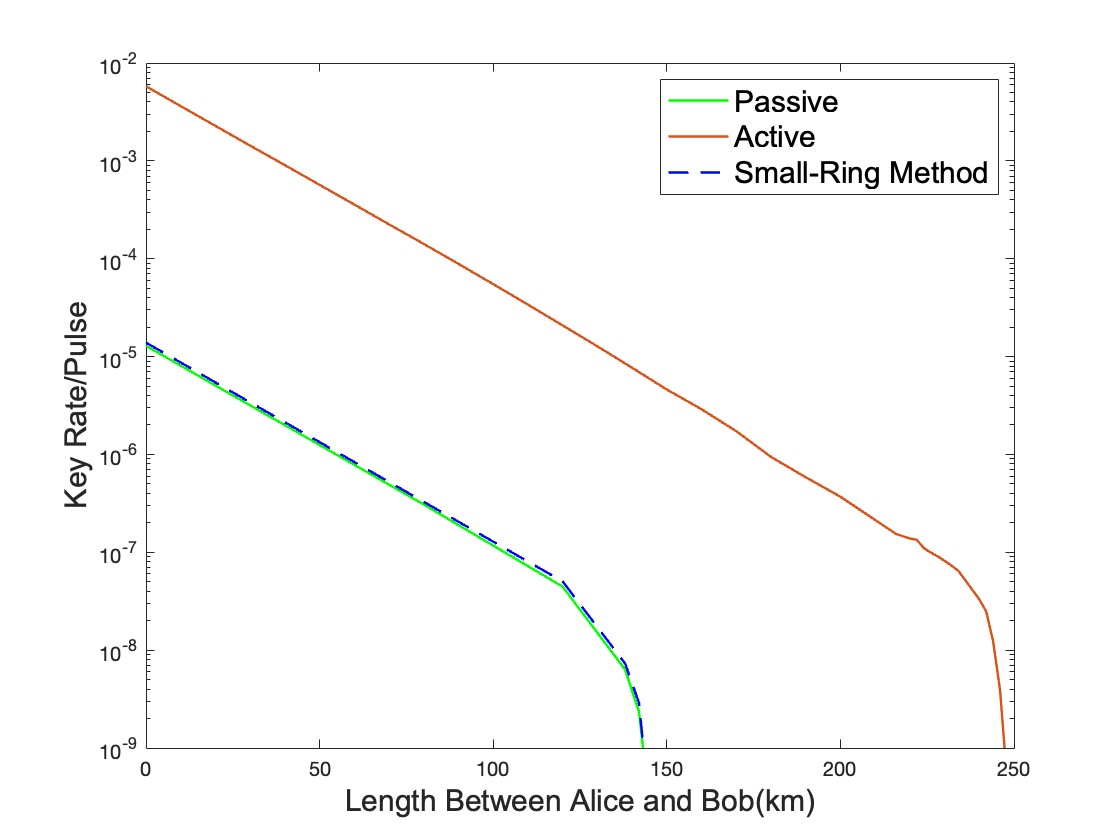}
	\caption{
		Simulated secret-key rates for the fully passive MDI-QKD protocol
		compared with a standard actively modulated MDI-QKD scheme.
		Both approaches employ three decoy settings. 
		For the passive protocol, the postselection parameters are fixed at
		$\Delta_{XY}=\Delta_{\phi}=0.2$, while $\Delta_Z$ is optimized over
		$[0.01,0.05]$ and the highest decoy threshold $t_3$ is optimized in the
		interval $[0.1,0.99]$.  
		In the active protocol, the three decoy intensities are optimized within
		$[0,1]$.  A channel loss of $0.2~\mathrm{dB/km}$, a detector dark count
		probability of $10^{-6}$, and a total misalignment of $1\%$ are assumed.
		The passive scheme achieves a maximum distance of approximately
		$143~\mathrm{km}$, with key rates that are two to three orders of
		magnitude lower than those of the active scheme due to double
		postselection, but with the key advantage that all source-modulator
		side-channels are eliminated.  
	}
	\label{fig:keyrate}
\end{figure}

	\subsection{Discussion}
	
	The proposed passive MDI-QKD protocol integrates fully passive sources into the MDI-QKD framework, creating a system that is simultaneously immune to all detector and modulator side-channels. This marks a significant advance toward bridging the gap between theoretical QKD security and real-world implementation. Compared to recently proposed fully passive twin-field QKD (TF-QKD) schemes~\cite{Wang2023}, the present protocol is experimentally simpler, as it does not require global phase stabilization across distant nodes.
	
	\paragraph{Practical Assumptions.}
	The security of the protocol relies on several practical assumptions. First, the four lasers in each source must have independent and uniformly random phases, without inter-pulse correlations. Second, the local classical detectors used for monitoring must measure the intensity and phase accurately and securely. Third, stable two-photon interference must be maintained both locally within each user’s source and remotely between Alice, Bob, and Charlie. Although the passive design removes modulator side-channels, care must still be taken to secure the classical detectors and to prevent leakage from fixed optical paths or reflection patterns~\cite{Wang2023a}.
	
	\paragraph{Finite-Size Considerations.}
	While my analysis focuses on the asymptotic regime, I can estimate requirements for a realistic finite-key scenario. With postselection widths of $\Delta_{XY}=\Delta_{\phi}=0.2$, only about $0.022\%$ of $X$-basis signals are retained. To achieve the same statistical confidence as an active MDI-QKD experiment with $10^{10}$ pulses, the passive protocol would require roughly $10^{12}$ pulses~\cite{Li2024}. This can be relaxed by increasing the postselection widths: for example, with $\Delta_{XY}=\Delta_{\phi}=0.5$, the retained fraction rises to $0.42\%$, reducing the required number of pulses to $10^{11}$, though the asymptotic key rate decreases by about $35\%$. A full finite-size analysis remains an important direction for future work.

	\paragraph{Outlook.}
	This work demonstrates the theoretical feasibility and performance characteristics of a fully passive MDI-QKD system. Future research will focus on extending the security proof to the finite-size regime, experimentally demonstrating the protocol, and improving sifting efficiency to enhance the overall key rate. The proposed architecture provides a clear path toward high-security, practical QKD networks suitable for real-world deployment.



	\chapter{FULLY PASSIVE TWIN-FIELD-BASED CONFERENCE KEY AGREEMENT}
	\label{chap:passive-cka}

\section*{Statement of Contribution}

This chapter is based on the work presented in the preprint~\cite{Li2024CKA}:
\begin{quote}
	J. Li, W. Wang, and H. F. Chau, ``Fully Passive Quantum Conference Key Agreement,'' arXiv preprint arXiv:2407.15761, 2024.
\end{quote}

This work was a collaboration between the author, Dr. Wenyuan Wang, and Prof. H. F. Chau. The authors jointly aimed to extend the principles of passive, interference-based Quantum Key Distribution (QKD) to the multi-user domain of Conference Key Agreement (CKA), with the goal of achieving both excellent performance and high implementation security.

The author's contributions to this chapter are as follows:

The author designed the technical details and developed the specific implementation of the fully passive CKA protocol, based on a conceptual architecture proposed by Dr. Wang, and provided theoretical feasibility by Prof. Chau.

Regarding the security analysis, the author adapted the ``corrected-yield'' technique to the CKA framework, which involved deriving the generalized correction integrals required for the multi-user scenario. To improve protocol efficiency, the author also adapted the ``phase mismatching'' strategy to the multi-user setting and derived the necessary modifications.

Finally, to enable the numerical simulation of this computationally intensive protocol, the author developed a high-performance program and a novel ``branch cutting'' algorithm. Using this program, the author conducted a detailed performance analysis for a four-user network.


\section{Introduction}
\label{sec:intro}

This work extends the principles of passive, interference-based Quantum Key Distribution (QKD) to the multi-user domain of Conference Key Agreement (CKA)~\cite{Li2024CKA}, a challenge that is not as well-investigated in the field. The foundational advantages of unifying the passive and Twin-Field (TF) methods are preserved~\cite{Wang2023}: the TF-QKD architecture overcomes the rate-distance limit and provides Measurement-Device-Independent (MDI) security~\cite{Lucamarini2018,Lo2012}, while passive sources eliminate user-side security loopholes~\cite{Wang2023a,Li2024}. My goal is to unify these solutions to create a multi-user QKD system that achieves both excellent performance and a high level of implementation security~\cite{Li2024CKA}.
	
	\section{Preliminary Protocols}
	\label{sec:prelim}
	\subsection{Twin-Field Quantum Key Distribution (TF-QKD)}
	Standard point-to-point QKD protocols are fundamentally limited by the repeaterless bound, which dictates that the secret key rate scales linearly with channel transmittance ($R \propto \eta$)~\cite{Pirandola2017PLOB}. Twin-Field QKD (TF-QKD) is a protocol designed to overcome this limit, achieving a more favorable scaling of $R \propto \sqrt{\eta}$ by using single-photon interference at an untrusted central relay~\cite{Lucamarini2018}.
	
	\subsubsection*{Protocol and Measurement}
	The TF-QKD protocol is a type of Measurement-Device-Independent (MDI) QKD where two users, Alice and Bob, send weak coherent pulses (WCPs) to an untrusted measurement station, Charlie~\cite{Lo2012,Lucamarini2018}. Charlie's station consists of a 50:50 beam splitter followed by two single-photon detectors, $D_c$ and $D_d$. A successful event is heralded when only one of these two detectors registers a ``click".
	
	The protocol distinguishes between two operational bases and associated rounds:
	
	\begin{enumerate}
		\item \textbf{Key Generation (KG) in the X basis:} To generate the raw key, Alice and Bob send WCPs where the bit is encoded in the phase. For example, they prepare coherent states $|\alpha\rangle$ (for bit `0`) or $|-\alpha\rangle$ (for bit `1`), which correspond to phases of $0$ and $\pi$. The global phase reference between them must be stable during these rounds~\cite{Lucamarini2018}.
		
		\item \textbf{Parameter Estimation (PE) in the Z basis:} To test for eavesdropping, Alice and Bob use the complementary Z basis. In these rounds, they send phase-randomized coherent states. The intensity of each pulse is randomly chosen from a predefined set of decoy values, e.g., $S = \{\beta_1, \beta_2, ...\}$, which is essential for the decoy-state method~\cite{Lo2005,Wang2005,Curty2019}.
	\end{enumerate}
	
	After each transmission, Charlie publicly announces his measurement outcome, for instance, $(k_c=1, k_d=0)$ if detector $D_c$ clicked. Alice and Bob only keep the single-click events. To correlate their raw keys from the KG rounds, they apply a simple rule: if $D_c$ clicks, their bits should be the same ($b_A = b_B$), and if $D_d$ clicks, their bits should be different ($b_A \neq b_B$). In the latter case, one party (e.g., Bob) flips their bit.
	
	\subsubsection*{Parameter Estimation and Key Rate}
	To distill a final, secure key from the raw data, the parties must quantify and remove two things: the errors in their key string and the information that an eavesdropper, Eve, might have gained. This leads to two distinct procedures: a direct measurement of the bit-error rate for error correction, and a more complex estimation of the phase-error rate for privacy amplification.
	
	\paragraph{The Bit-Error Rate (A Direct Measurement for Error Correction)}
	This process is straightforward. The bit-error rate, denoted $e_{X,k_c k_d}$, quantifies the disagreements in the raw key generated in the X basis. It is determined empirically:
	\begin{itemize}
		\item Alice and Bob publicly compare a randomly chosen subset of their raw key bits.
		\item The fraction of bits that do not match (after applying the bit-flip rule based on the detector click) is the measured bit-error rate, $e_{X,k_c k_d}$.
	\end{itemize}
	This measured value tells them how much information they must sacrifice for the classical error correction step, which is represented by the term $H_2(e_{X,k_c k_d})$ in the key rate formula.
	
	\paragraph{The Phase-Error Rate (An Indirect Estimation for Privacy Amplification)}
	This is the more subtle and critical part of the security proof. The phase-error rate, $e_{Z,k_c k_d}$, represents the error rate Eve would see if the parties had used the complementary Z basis for key generation. It quantifies Eve's potential information, and the parties must find a secure upper bound for it, $e_{Z,k_c k_d}^{\text{upp}}$, to know how much to shorten their key during privacy amplification~\cite{Curty2019}.
	
	The core challenge is that the key is generated from imperfect Weak Coherent Pulses (WCPs), not ideal single-photon sources. While the ``good" part of the key comes from single-photon interference events (`(1,0)` and `(0,1)` pairs), the WCPs also produce ``bad" multi-photon events (e.g., `(2,0)`, `(0,2)`, `(1,1)`), which Eve can attack to gain information. The decoy-state method is the tool used to meticulously account for the contributions of both good and bad events. The estimation follows a clear three-step logical chain:
	
	\begin{enumerate}
		\item \textbf{Step 1: Measure the Gains.} The process begins with raw experimental data from the Z-basis (PE) rounds. For each combination of decoy intensities $(\beta_A, \beta_B)$ that they sent, the parties measure the \emph{gain}: the probability of getting a specific single-click event. This observable quantity is denoted $p_{ZZ}(k_c, k_d | \beta_A, \beta_B)$.
		
		\item \textbf{Step 2: Estimate the Yields.} The measured gains are a statistical mixture of events from all possible photon-number pairs that could have been emitted. The decoy-state method ``unmixes" them. The gain is related to the underlying photon-number-resolved \emph{yields} ($Y_{n_A, n_B}$) by the equation:
		\begin{equation}
			p_{ZZ}(k_c, k_d | \beta_A, \beta_B) = \sum_{n_A, n_B=0}^{\infty} Y_{n_A, n_B} P_{\beta_A^2}(n_A) P_{\beta_B^2}(n_B).
		\end{equation}
		By creating a system of these equations using the data from their multiple decoy settings, the parties use Linear Programming (LP) to solve for the unknown yields. This single procedure provides secure bounds on all the necessary yields: lower bounds on the ``good" single-photon yields (like $Y_{10}$, $Y_{01}$) and upper bounds on the``bad" multi-photon yields (like $Y_{20}$, $Y_{02}$, etc.).
		
		\item \textbf{Step 3: Calculate the Phase-Error Rate Bound.} Finally, the parties take all the yield bounds they just estimated and insert them into the specific formula provided by the security proof (such as Eq. 15 in the work by Curty, Azuma, and Lo~\cite{Curty2019}). This formula combines the contributions from all photon-number components to produce a single, conservative upper bound on the phase-error rate, $e_{Z,k_c k_d}^{\text{upp}}$.
	\end{enumerate}
	
	\paragraph{Final Key Rate}
	With both the measured bit-error rate and the estimated phase-error rate bound, the parties can calculate the final secure key rate. The total rate is the sum of contributions from the two successful single-click outcomes, $(1,0)$ and $(0,1)$:
	\begin{equation}
		R_X = R_{X,10} + R_{X,01},
	\end{equation}
	where the rate for each outcome subtracts the costs of both error correction and privacy amplification:
	\begin{equation}
		R_{X,k_c k_d} \ge p_{XX}(k_c, k_d) \left[ \underbrace{1}_{\text{Raw bits}} - \underbrace{H_2(e_{X,k_c k_d})}_{\text{Error Correction cost}} - \underbrace{H_2(e_{Z,k_c k_d}^{\text{upp}})}_{\text{Privacy Amplification cost}} \right].
	\end{equation}
	Here, $p_{XX}(k_c, k_d)$ is the probability of a successful KG outcome and $H_2(x)$ is the binary entropy function.
	
\subsection[Interference-Based Conference Key Agreement (CKA)]{Interference-Based \\ Conference Key Agreement (CKA)}
	The principles of TF-QKD can be generalized from a two-party protocol to a multi-party network to achieve Conference Key Agreement (CKA). The goal of CKA is to establish a single, identical secret key among an arbitrary number of users, $N_U$. The protocol described here is a practical approach that uses only weak coherent pulses (WCPs) and linear optics, making it feasible with current technology while still being able to overcome the multipartite repeaterless bounds~\cite{Carrara2023}.
	
	\subsubsection*{Physical Setup and Measurement}
	In this protocol, all $N_U$ parties ($A_0, A_1, ..., A_{N_U-1}$) send their optical pulses to a central, untrusted relay. This relay houses a network of balanced beam splitters (BBS), which is a passive optical circuit designed to interfere all incoming signals. The network has $M$ input ports and $M$ output ports (where $M \ge N_U$ and $M$ is a power of two), and each output is monitored by a threshold single-photon detector. The number of optical components scales efficiently as $\mathcal{O}(M \log_2 M)$~\cite{Carrara2023,Grasselli2019CKA}.
	
	The BBS network performs a specific linear optical transformation that maps the creation operator for a photon from user $i$, $\hat{a}_i^\dagger$, to a superposition of creation operators at the output detectors, $\hat{d}_j^\dagger$. This relationship is given by~\cite{Carrara2023}:
	\begin{equation}
		\hat{a}_{i}^{\dagger} \rightarrow \frac{1}{\sqrt{M}}\sum_{j=0}^{M-1}(-1)^{\vec{j}\cdot\vec{i}}\hat{d}_{j}^{\dagger},
		\label{eq:BSnet-prelim-revised}
	\end{equation}
	where $\vec{i}$ and $\vec{j}$ are the binary vector representations of the user and detector indices, respectively.
	
	Similar to the bipartite TF-QKD case, a successful measurement event is heralded when exactly one of the $M$ detectors clicks. This event is labeled $\Omega_j$, where $j$ denotes the detector that fired. Such a single-click event post-selects a W-like quantum correlation among all users, from which a secret key can be distilled~\cite{Grasselli2019CKA}.
	
	\subsubsection*{Protocol Description}
	The protocol is executed in a series of rounds, with each round being designated for either key generation or parameter estimation:
	\begin{enumerate}
		\item \textbf{State Preparation:}
		\begin{itemize}
			\item In a \textbf{Key Generation (KG) round}, each party $A_i$ randomly chooses a bit $x_i \in \{0, 1\}$ and prepares a coherent state with a fixed amplitude $\alpha_i$, encoding the bit in the phase: $|(-1)^{x_i}\alpha_i\rangle$.
			\item In a \textbf{Parameter Estimation (PE) round}, each party prepares a phase-randomized coherent state (PRCS). The intensity $\beta_i$ of the state is randomly chosen from a public set of decoy values, $S_i$.
		\end{itemize}
		\item \textbf{Transmission and Measurement:} Each party sends their prepared state to the untrusted relay. The relay performs the interference measurement and publicly announces the detection pattern $\vec{k} \in \{0,1\}^M$. The round is discarded unless exactly one detector clicks (i.e., if the Hamming weight $|\vec{k}| \neq 1$).
	\end{enumerate}
	
	\subsubsection*{Parameter Estimation and Security}
	To distill a final, secure conference key, the parties must perform two essential tasks for each successful single-click outcome $\Omega_j$. First, they must correct the errors in their raw key (error correction). Second, they must remove any information an eavesdropper, Eve, might have gained (privacy amplification). These two tasks correspond to measuring the bit-error rate and estimating the phase-error rate, respectively.
	
	\paragraph{The Quantum Bit Error Rate (A Direct Measurement for Error Correction)}
	The ``Quantum Bit Error Rate (QBER)", denoted $Q_{X_{0},X_{i}}^{j}$, quantifies the disagreements in the raw key generated during the KG rounds. It is a ``directly measured" value that determines the cost of the error correction step. The process is straightforward:
	\begin{enumerate}
		\item The parties publicly compare a randomly chosen subset of their raw key bits from the KG rounds.
		\item For a given outcome $\Omega_j$, a reference party $A_0$ compares her bit value $x_0$ with the bit value $x_i$ from another party $A_i$.
		\item An error is counted if the bits are not properly correlated according to the optical network's transformation. The QBER is the frequency of these mismatches:
		\begin{equation}
			Q_{X_{0},X_{i}}^{j} = \text{Pr}\left(X_{0} \ne (-1)^{\vec{j}\cdot\vec{i}}X_{i} | \Omega_{j}, \text{KG}\right).
		\end{equation}
	\end{enumerate}
	The cost of fixing these errors via classical error correction is given by the term $H_2(Q_{X_0, X_i}^j)$ in the key rate formula.
	
	\paragraph{The Phase-Error Rate (An Indirect Estimation for Privacy Amplification)}
	The ``phase-error rate", $Q_Z^j$, is a theoretical quantity that bounds Eve's maximum possible information about the key. It represents the error rate that would have occurred if the parties had used a complementary basis for key generation. This quantity cannot be measured directly. Instead, the parties must calculate a secure upper bound, $\overline{Q}_Z^j$, by following a rigorous, three-step estimation chain using data from their PE rounds~\cite{Carrara2023}.
	
	\begin{enumerate}
		\item \textbf{Step 1: Measure the Gains.} The process begins with the raw experimental data from the PE rounds. For each combination of decoy intensities $\vec{\beta} = (\beta_0, ..., \beta_{N-1})$ sent by the parties, they measure the probability of observing a single-click event $\Omega_j$. This directly observable probability is called the multipartite \textbf{gain}, $G_{\vec{\beta}}^j$.
		
		\item \textbf{Step 2: Estimate the Yields.} The measured gains are a statistical mixture of events from all possible photon-number combinations that the parties' WCPs could have emitted. The gain is related to the underlying photon-number-resolved \textbf{yields} ($Y_{\vec{n}}^j = Y_{n_0,...,n_{N-1}}^j$) by the multipartite decoy-state equation:
		\begin{equation}
			G_{\vec{\beta}}^{j} = \sum_{n_{0},...,n_{N-1}=0}^{\infty} \left[ \prod_{i=0}^{N-1} P_{\beta_{i}}(n_{i}) \right] Y_{n_{0},...,n_{N-1}}^{j}.
		\end{equation}
		By creating a system of these equations for their different decoy settings, the parties use Linear Programming (LP) to solve for the unknown yields. This procedure provides secure upper bounds on the yields, $\overline{Y}_{\vec{n}}^j$, for all relevant photon-number components, including the ``bad" multi-photon states that are vulnerable to eavesdropping.
		
		\item \textbf{Step 3: Calculate the Phase-Error Rate Bound.} Finally, the parties take the yield bounds $\overline{Y}_{\vec{n}}^j$ they just found and insert them into the master formula for the phase-error rate bound, which is derived from the security proof. This formula has the structure~\cite{Carrara2023}:
		\begin{equation}
			\overline{Q}_{Z}^{j} \propto \sum_{v} \left\{ \left[ \sum_{\overline{n}} \left( \prod_{i=0}^{N-1} c_{i,n_{i}}^{(v_{i})} \right) \sqrt{\overline{Y}_{\overline{n}}^{j}} \right] + \Delta \right\}^{2}.
		\end{equation}
		This calculation combines the contributions of all bounded yields to produce a single, conservative upper bound on the phase-error rate, $\overline{Q}_Z^j$.
	\end{enumerate}
	This final bound tells the parties how much of their key they must ``shred" during privacy amplification to eliminate Eve's knowledge. The cost of this step is given by the term $H_2(\overline{Q}_Z^j)$.
	
	A full derivation of the multipartite decoy-state equations and the two-decoy
	linear-program construction used in this chapter is provided in 
	Appendix~\ref{sec:cka_lp}.

	\section{Methods}
	\label{sec:methods}
	
	\subsection{Setup}
	\label{subsec:setup}
The protocol is designed for a network of $N_U$ users, each equipped with a special type of optical source (see Fig.~\ref{fig:ckasetup}). 
	All users are connected via quantum channels to a central, untrusted relay station where the interference measurements take place.
	
\begin{figure}[t]
	\centering
	\includegraphics[width=0.72\linewidth]{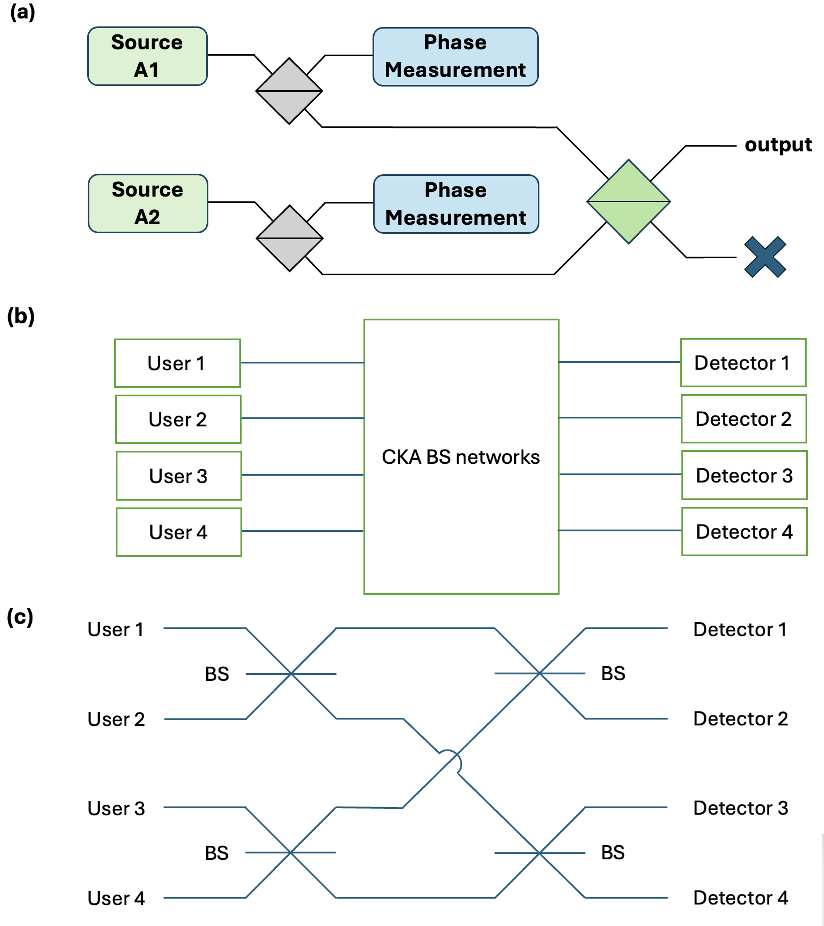}
	\caption{
		(a) Fully passive source for a single user: two independent lasers
		generate pulses with random phases, which are interfered on a local
		50{:}50 beam splitter to produce the outgoing signal, adapted from
		Ref. \cite{Wang2023}.  
		(b) Schematic architecture of the multiuser CKA relay, consisting of a
		layered network of 50{:}50 beam splitters that distribute each input
		mode over all output ports.  
		(c) Example of a four-user beam-splitter network implementing the
		linear optical transformation used in CKA.  
		Panels (b) and (c) are adapted from Ref.~\cite{Carrara2023}.  
	}
	\label{fig:ckasetup}
\end{figure}

	\subsubsection*{The Fully Passive Source}
	The core of this protocol is the unique design of each user's source, which is ``fully passive." This means it contains no active optical modulators, thereby eliminating a major class of security vulnerabilities. The design is as follows~\cite{Wang2023,Li2024CKA}:
	\begin{itemize}
		\item \textbf{Signal Generation:} Each user, $A_i$, employs two independent laser sources. In each round, these lasers produce optical pulses with random and uncorrelated phases, which I denote as $\phi_{i,1}$ and $\phi_{i,2}$. For simplicity, I assume the average intensity from both lasers is identical, at $\mu_{max}/2$.
		\item \textbf{Local Interference and Monitoring:} The two pulses are combined on a local 50:50 beam splitter (BS) inside the user's lab. One output port of this BS serves as the final optical signal that is transmitted to the relay. Crucially, the user employs classical detectors to locally measure and record the initial phases, $\phi_{i,1}$ and $\phi_{i,2}$, of the two lasers.
		\item \textbf{Resulting Output State:} Due to the random interference, the final transmitted signal has a phase and intensity that are randomly determined by the phase difference $(\phi_{i,1} - \phi_{i,2})$. The output signal's intensity, $\mu_i$, can take any value in the range $[0, \mu_{max}]$, while its phase, $\phi_i$, can be any value in $[0, 2\pi)$.
	\end{itemize}
	This setup generates a continuous distribution of quantum states. To make this useful for the protocol, the parties discretize this space during post-processing. The phase circle is partitioned into $M$ equal sectors, or ``slices," each of width $2\pi/M$. For parameter estimation, the continuous intensity range is similarly coarse-grained into a small number of bins, corresponding to the different decoy states~\cite{Wang2023}.
	
	\subsubsection*{The CKA Relay and Network}
	All $N_U$ users send their passively generated optical signals to a central, untrusted relay station. The purpose of this relay is to perform a multipartite interference measurement.
	\begin{itemize}
		\item \textbf{Network Architecture:} The relay contains a passive optical circuit composed of multiple layers of 50:50 beam splitters. For a system designed to handle up to $M=2^s$ users, this network consists of $s$ layers, with each layer containing $2^{s-1}$ beam splitters~\cite{Carrara2023}.
		\item \textbf{Optical Transformation:} This network is engineered to perform a specific linear transformation on the incoming optical modes. The creation operator of a photon from user $i$, $\hat{a}_i^\dagger$, is mapped to a superposition of creation operators across all $N_D = 2^s$ output detectors, $\hat{d}_j^\dagger$. This transformation is described by~\cite{Carrara2023}:
		\begin{equation}
			\hat{a}_{i}^{\dagger} \rightarrow \frac{1}{\sqrt{N_{D}}}\sum_{j=0}^{N_{D}-1}(-1)^{\vec{j}\cdot\vec{i}}\hat{d}_{j}^{\dagger},
		\end{equation}
		where $\vec{i}$ and $\vec{j}$ are the binary vector representations of the user's index and the detector's index, respectively.
	\end{itemize}
	
	\subsubsection*{Channel and Detection Model}
	\begin{itemize}
		\item \textbf{Quantum Channels:} Each link between a user and the central relay is modeled as a standard lossy optical fiber. The channel is characterized by a loss coefficient (e.g., $0.2 \text{ dB/km}$) and may also introduce small, random phase and polarization misalignments.
		\item \textbf{Detection:} The relay is equipped with $N_D$ standard threshold single-photon detectors, one for each output port of the beam splitter network. These detectors are characterized by a quantum efficiency $\eta_d$ and a dark count probability $p_d$. A successful experimental outcome for any given round is defined as a ``single-click event", where exactly one of the $N_D$ detectors fires. All other events (no clicks or multiple clicks) are discarded.
	\end{itemize}

	\subsection{Protocol}
	\label{subsec:protocol}
	The execution of the fully passive CKA protocol can be divided into two main stages: a physical transmission stage where quantum signals are exchanged, followed by a classical post-processing stage where the raw data is refined into a secret key~\cite{Li2024CKA}.
	
	\subsubsection*{1. State Transmission and Measurement}
	In each experimental round, the physical process is identical and straightforward:
	\begin{itemize}
		\item Each user, $A_i$, generates a random optical signal using their fully passive source as described in the setup. No active modulation is performed.
		\item While the quantum signal is sent to the central relay, each user locally measures and records the classical phases, $\phi_{i,1}$ and $\phi_{i,2}$, from their two internal lasers. This local classical data is kept private for now.
		\item The central relay, Charlie, performs the interference measurement and publicly announces the outcome. A round is considered successful only if a single-click event, $\Omega_j$, occurs, where $j$ identifies which detector fired. All other rounds are immediately discarded.
	\end{itemize}
	
	\subsubsection*{2. Classical Post-Processing}
	After a large number of transmission rounds have been completed, the parties begin the classical post-processing phase. This is where the raw detection events are sorted and analyzed to generate the key.
	
	\paragraph{Late Assignment of Roles (KG vs. PE)}
	A key advantage of the fully passive protocol is the concept of ``late basis choice." The physical signals sent for key generation and parameter estimation are identical. The decision on how to use each round's data is made only after the quantum communication is finished~\cite{Li2024CKA}:
	\begin{itemize}
		\item The reference party, $A_0$, publishes a long, random binary string.
		\item This string retroactively assigns each successful round to one of two groups: \textbf{Key Generation (KG)} rounds with probability $p_X$, or \textbf{Parameter Estimation (PE)} rounds with probability $1-p_X$.
	\end{itemize}
	This procedure completely eliminates the need for multi-user basis sifting, which is a major source of inefficiency in active CKA protocols. For $N_U$ users, this provides a sifting efficiency advantage of $1/p_X^{N_U-1}$ over a comparable active scheme~\cite{Li2024CKA}.
	
	\paragraph{Processing of Key Generation (KG) Rounds}
	For the rounds designated as KG, the parties use their locally recorded phases to distill a raw key:
	\begin{itemize}
		\item Each user $A_i$ takes their two measured phases, $\phi_{i,1}$ and $\phi_{i,2}$, and determines which of the $M$ phase ``slices" each phase falls into, yielding two integer indices, $k_{i,1}$ and $k_{i,2}$.
		\item In the simplest version of the protocol, a round is accepted for the raw key only if a strict alignment condition is met: for \textbf{all} users, their phase indices must be either $(k_{i,1}, k_{i,2}) = (1, 1)$ or $(k_{i,1}, k_{i,2}) = (M/2+1, M/2+1)$.
		\item A classical bit is then assigned: the ``all top slice" case ($k_{i,1}=k_{i,2}=1$ for all $i$) corresponds to bit `0`, while the ``all bottom slice" case ($k_{i,1}=k_{i,2}=M/2+1$ for all $i$) corresponds to bit `1`.
	\end{itemize}
	
	\paragraph{Processing of Parameter Estimation (PE) Rounds}
	For the rounds designated as PE, the parties use their local data to gather the statistics needed for the decoy-state analysis:
	\begin{itemize}
		\item Each user $A_i$ calculates the absolute difference between their two measured phases, $|\phi_{i,1} - \phi_{i,2}|$. This value is directly related to the intensity of the signal they transmitted in that round.
		\item Each user ``bins" their calculated intensity into one of a small number of predefined ranges (e.g., ``low," ``medium," ``high" for a three-decoy protocol).
		\item All users publicly announce which intensity bin their signal belonged to. The combination of these announcements (e.g., ``Alice sent low, Bob sent high, ...") defines a joint decoy state setting for that round, which is used to measure the gains for the multipartite decoy-state analysis~\cite{Wang2023}.
	\end{itemize}
	
	\subsubsection*{Strategy to Utilize All Slice Combinations}
	The strict post-selection rule for KG rounds, accepting only events where all users fall into either slice 1 or slice M/2+1, is conceptually simple but extremely inefficient, discarding the vast majority of successful detection events. A crucial insight of the fully passive framework is that this enormous sifting loss is not necessary. In principle, a secret key can be extracted from \emph{any} round, as long as the phase relationship between the users is known~\cite{Wang2023,Li2024CKA}.
	
	The underlying physics is that the interference outcome at the relay (i.e., which detector, $\Omega_j$, clicks) is deterministically governed by the relative phases of the signals sent by all users.
	\begin{itemize}
		\item In the strict case, when everyone sends a signal from slice 1 (bit `0'), the phase relationships are simple and well-defined, leading to a predictable interference pattern and a clear rule for correlating the final key bits.
		\item Now, consider a ``mismatched" event: for instance, in a two-user case, Alice's signal comes from slice 1, while Bob's comes from slice 2. Their signals still interfere at the relay. Because both Alice and Bob know which slice their signal came from, the relative phase between them is still known to within a small tolerance.
	\end{itemize}
This means that for this specific input combination (`Alice=slice 1', `Bob=slice 2'), the resulting interference pattern is just as predictable as in the perfectly aligned case. The correlation rule simply changes. The parties can use their knowledge of the slice combination and the detector outcome $\Omega_j$ to calculate what the expected correlation should be. They then apply the necessary classical corrections (bit-flips) during post-processing to align their key strings.
	
	This principle generalizes to all $M^{2N_U}$ possible combinations of slice choices. Each combination has its own unique ``recipe" for correlating the bits based on the detector outcome. Instead of a protocol limited by a tiny sifting factor, the total key rate becomes a large sum over the secure key contributions from every single combination. This transforms the sifting problem into a computational one, dramatically improving the protocol's overall efficiency and final key rate~\cite{Li2024CKA}.

\subsection{Security Framework}
\label{subsec:security}

The security proof for the fully passive CKA protocol largely follows that of its active counterpart, as the MDI-based architecture removes all detector-side vulnerabilities. 
A key technical step is to model the fully passive source using two equivalent formulations, illustrated in Fig.~\ref{fig:securitymodel}.  
The main challenge and modification arise from the imperfect state preparation inherent in the fully passive sources. 
This section details the framework used to rigorously account for these source imperfections.

The detailed mathematical forms of the imperfect-source model, the yield
correction integrals, and the phase-error estimation expressions used in
this chapter are collected in Appendix~\ref{sec:cka_security}.

\begin{figure}[H]
	\centering
	\includegraphics[width=0.85\linewidth]{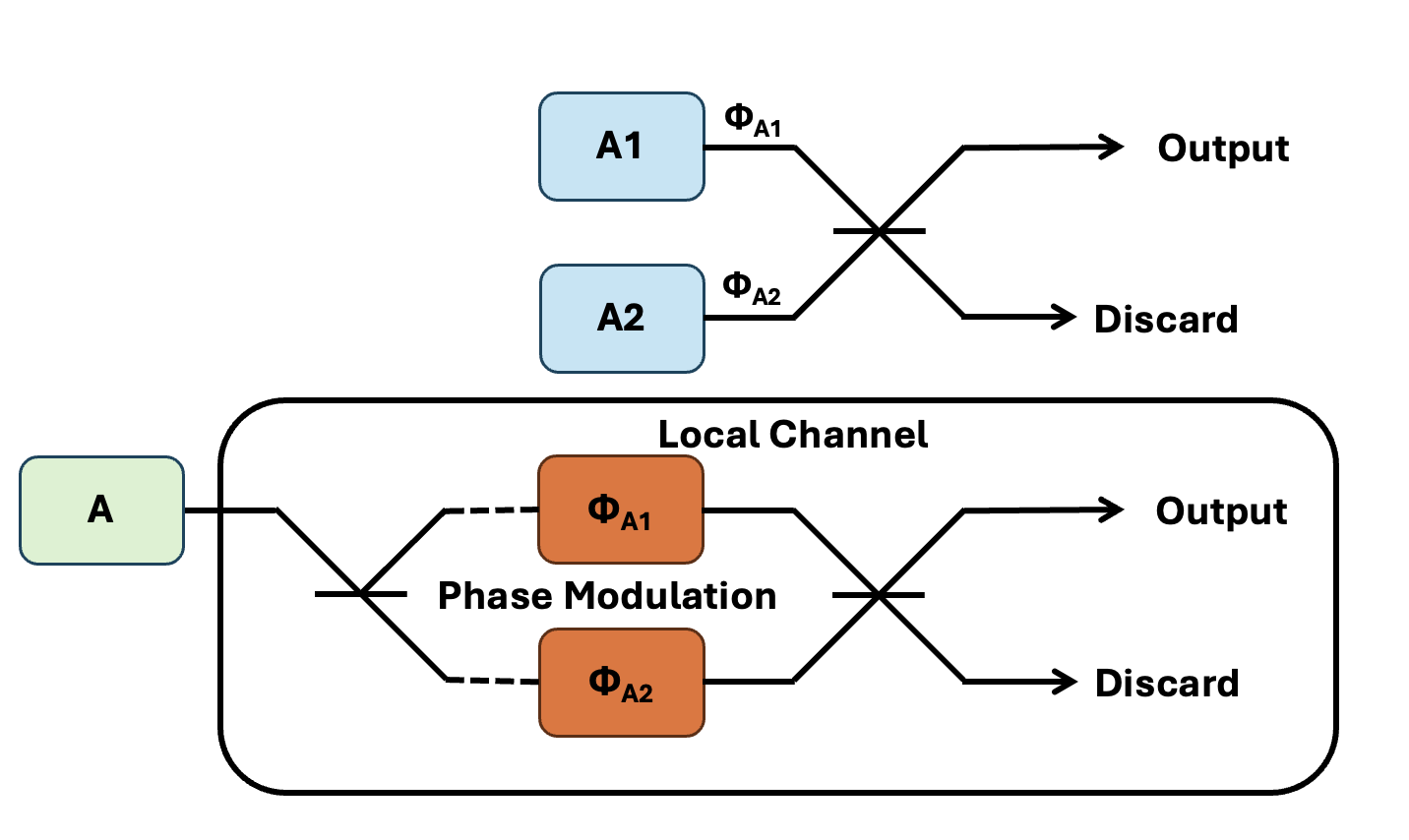}
	\caption{
		Two equivalent descriptions of the fully passive source model used in
		the security analysis. (a) In the first picture, the user employs two independent
		random-phase sources whose outputs interfere at a beam splitter to form
		the transmitted signal. (b) In the second, operationally equivalent formulation, a single
		optical field is split into two arms, each undergoing an independent
		random phase modulation before being recombined at a beam splitter.
		The pair of phase modulators and the beam splitters together define the
		effective channel $\mathcal{E}_{A}$. \textit{Figure reproduced from Figs.~3 and~4 of
			W.~Wang, R.~Wang, and H.-K.~Lo,
			``Fully-Passive Twin-Field Quantum Key Distribution,'' 
			(2023) \cite{Wang2023}, and subsequently incorporated into my passive CKA work
			\cite{Li2024CKA}.}
	}
	\label{fig:securitymodel}
\end{figure}

	\subsubsection*{The Challenge: Imperfect Source Preparation}
	In an active protocol, prepared states are assumed to be perfect. In my passive scheme, however, a state is defined by post-selecting a signal whose generating phases fall within a ``slice" of width $2\Delta$. This means the actual phase of the signal is not a single value but is randomly distributed within a small range, e.g., $[-\Delta, \Delta]$. This unavoidable phase fluctuation is the key source imperfection that the security proof must address.
	
	\subsubsection*{The Solution: The Equivalent Local Channel Model}
	To analyze this imperfection, I adopt a powerful conceptual tool: the equivalent local channel model~\cite{Wang2023}. I model each user's imperfect passive source as being mathematically equivalent to a two-step process:
	\begin{enumerate}
		\item A perfect, hypothetical source inside the user's lab emits a signal in a pure quantum state with a precise phase.
		\item This perfect signal then passes through a ``local channel," denoted $\mathcal{E}_{A_i}$, which is also inside the user's lab. This channel consists of random phase modulations and beam splitters that precisely replicate the physical effect of the phase uncertainty within a slice.
	\end{enumerate}
	This allows us to separate the ideal state preparation from the physical imperfection, which is now encapsulated entirely within the local channel $\mathcal{E}_{A_i}$.
	
	\subsubsection*{The Conservative Security Assumption}
	To guarantee security, I make the most pessimistic assumption possible: I yield this entire local channel $\mathcal{E}_{A_i}$ to the eavesdropper, Eve. In this picture, Eve not only controls the external fiber and the relay, but also the small amount of local optics that model my source imperfection. This has two immediate consequences:
	\begin{itemize}
		\item The Quantum Bit Error Rate (QBER) measured in the Key Generation rounds will naturally be higher, as the phase fluctuations in $\mathcal{E}_{A_i}$ will contribute to errors.
		\item The yields estimated from the Parameter Estimation rounds must be corrected to account for the effects of this local channel.
	\end{itemize}
	
	\subsubsection*{Correcting the Yields for Security Analysis}
	The yields estimated via the multipartite decoy-state analysis, which I can denote $Y_{m_A,m_B,...}^j$, only reflect the properties of the channel ``external" to the users' labs. However, the phase-error rate formula must be applied to the yields of the hypothetical perfect states before they entered the local channels.
	
	Therefore, I must perform a ``yield correction." I mathematically reverse the effect of the local channel $\mathcal{E}_{A_i}$ to find the true yields, $Y_{n_A,n_B,...}^j$, corresponding to the initial perfect states. This correction is a convolution that averages over all possible transformations within the local channels. For a four-user case, the corrected yield is given by~\cite{Wang2023}:
	\begin{equation}
		Y_{n_{A}n_{B}n_{C}n_{D}}^{j} = \int \sum_{m_{A}...m_{D}} \left[ \prod_{i=A}^{D} P_{\phi_{i1},\phi_{i2}}(m_{i}|n_{i}) \right] Y_{m_{A}m_{B}m_{C}m_{D}}^{j} d\phi_{A1}d\phi_{A2}\cdots d\phi_{D2}.
	\end{equation}
	Here, $Y_{m_A m_B m_C m_D}^{j}$ are the yields estimated from the PE data (characterizing the external channel), and $P_{\phi_{i1},\phi_{i2}}(m_i|n_i)$ is the probability that an initial state of $n_i$ photons becomes a state of $m_i$ photons after passing through user $i$'s local channel with random phase modulations $\phi_{i1}$ and $\phi_{i2}$. This conditional probability is given by~\cite{Wang2023}:
\begin{align}
	P_{\phi_{A1},\phi_{A2}}(m_{A}|n_{A}) 
	={}& \Biggl| \frac{1}{2!} \sqrt{\frac{n_{A}!}{m_{A}!(n_{A}-m_{A})!}} \notag \times \left\{ 1 + e^{i[\pi/2+(\phi_{A2}-\phi_{A1})]} \right\}^{m_{A}}  \\ 
	& \times \left\{ 1 - e^{i[\pi/2+(\phi_{A2}-\phi_{A1})]} \right\}^{(n_{A}-m_{A})} \Biggr|^2.
	\label{eq:prob_phi}
\end{align}
	After calculating these ``corrected yields" $Y_{n_A n_B n_C n_D}^{j}$, they are then inserted into the standard phase-error rate formula for active CKA. This procedure ensures that the calculated phase-error bound is secure and rigorously accounts for the imperfections of the passive sources.

	The explicit corrected-yield formulas, distributions $P_{\phi_{i1},\phi_{i2}}(m_i|n_i)$, and the full phase-error estimation
	procedure are provided in Appendix~\ref{sec:cka_security} for reference.
	
	\subsection{Key Rate Formulation}
	\label{subsec:keyrate}
	The overall asymptotic secret key rate, R, for the fully passive CKA protocol is an average over all possible experimental outcomes. The final formula, which utilizes all possible combinations of phase slices to maximize efficiency, is given by~\cite{Li2024CKA}:
	\begin{equation}
R = \frac{1}{M^{2N_{U}}} \sum_{k_{A1}=1}^{M} \cdots \sum_{k_{D2}=1}^{M} \max \left[ 0, \sum_{j=0}^{N_{D}-1} \mathrm{Pr}(\Omega_{j} | \mathrm{KG}, k) \cdot r_{j}(k) \right].
	\end{equation}
	To understand this expression, I can break it down into its constituent parts.
	
	For comparison, Appendix~\ref{sec:cka_active} summarizes the corresponding key-rate 
	expression for the active CKA protocol, which my passive formulation 
	generalizes.

	\subsubsection*{Breakdown of the Formula}
	\begin{itemize}
		\item R: This is the total secret key rate, the final performance metric of the protocol, given in bits per transmitted signal pulse.
		
		\item $\frac{1}{M^{2N_U}}$: This is the normalization factor. It is not a sifting factor that causes loss. Since each of the $2N_U$ locally measured phases can randomly fall into one of M slices, there are $M^{2N_U}$ possible outcomes in total. This pre-factor simply averages the key rate contributions over all these equally probable outcomes.
		
		\item $\sum_{k_{A1}=1}^{M} \cdots \sum_{k_{D2}=1}^{M}$: This is the summation over all slice combinations. This is the mathematical representation of the strategy to use all experimental data. Instead of discarding mismatched events, I calculate the key contribution from every single combination of slice choices, $\mathbf{k} = (k_{A1}, k_{A2}, ..., k_{D2})$, and sum them up.
		
		\item $\sum_{j=0}^{N_D-1}$: This is the summation over all detector events. For any given slice combination, a single click can occur at any of the $N_D$ detectors. The total key rate for that combination is the sum of the contributions from each possible single-click event $\Omega_j$.
		
		\item $\text{Pr}(\Omega_j|\text{KG}, \mathbf{k})$: This is the conditional probability of observing a single click at detector $j$, given that the parties' phases corresponded to the specific slice combination $\mathbf{k}$.
		
		\item $r_j(\mathbf{k})$: This is the secret fraction, which represents the amount of secure key (in bits) that can be generated from a single, specific event (a given slice combination $\mathbf{k}$ and a given detector click $j$). It is defined as~\cite{Carrara2023}:
		\begin{equation}
			r_{j}(k) = 1 - H_2(\overline{Q}_{Z}^{j}(k)) - \max_{i \ge 1} \left\{ H_2(Q_{X_{0},X_{i}}^{j}(k)) \right\}.
		\end{equation}
		The components of the secret fraction are:
		\begin{itemize}
			\item $H_2(Q_{X_0,X_i}^j(\mathbf{k}))$ is the cost of error correction, determined by the worst-case Quantum Bit Error Rate (QBER) for that specific event.
			\item $H_2	(\overline{Q}_Z^j(\mathbf{k}))$ is the cost of privacy amplification, determined by the upper bound on the phase-error rate for that event, which is calculated using the corrected yields as described in the security framework.
		\end{itemize}
		
		\item $\max[0, \dots]$: This function simply ensures that if, for a particular slice combination, the costs of error correction and privacy amplification exceed the raw bit of information, its contribution to the total key rate is set to zero.
	\end{itemize}

	\section{Results and Discussion}
	\label{sec:results}
	
	\subsection{Numerical Configuration}
	\label{subsec:numcfg}
	To evaluate the performance of the protocol, I conducted a numerical simulation of a four-user fully passive CKA network. Unless specified otherwise, the simulation uses the following baseline parameters.
	
	To evaluate the protocol's performance, I simulate a four-user fully passive CKA network. The channel between each user and the relay is modeled as a standard optical fiber with a loss of 0.2 dB/km, and the relay's single-photon detectors have a dark count probability of $p_d = 10^{-8}$. For the passive source post-selection, each user divides their phase circle into M=8 slices. I simulate a two-decoy state protocol for parameter estimation, where each user evenly divides their output intensity range into two bins, the joint announcements from all four users then define the decoy cells for the linear programming analysis. To ensure a fair comparison with the active CKA protocol, the mean input intensity of the passive sources is set to be the same as the optimized intensity used in the active CKA simulation. The baseline simulation assumes perfect alignment between all users.
	
	The numerical evaluation uses a custom C++ implementation incorporating
	the Cuba library for high-dimensional integration and Gurobi for the 
	linear-program steps.  The branch-cutting acceleration used to prune
	negligible slice combinations is described in Appendix~\ref{sec:cka_branch}.

	\subsection{Key Rate Performance}
	\label{subsec:perf}
	Figure \ref{fig:rate_vs_loss} shows the simulated asymptotic key rate as a function of channel loss for the four-user network. The performance of the fully passive CKA is compared against that of its active counterpart~\cite{Li2024CKA}.
	
	The results show that, as expected, the introduction of passive sources reduces the overall key rate. Specifically, when using theoretical (exact) yields, the key rate of the passive protocol is approximately two to three orders of magnitude lower than that of the active protocol. When using a more practical two-decoy linear program to estimate the yields, the key rate is reduced by an additional order of magnitude. This second reduction is due to the fact that using a limited number of decoys provides looser bounds on the photon-number yields, which translates to a more conservative (and thus lower) key rate.
	
	Despite this reduction in performance, the simulation confirms that the protocol remains viable. The four-user fully passive CKA can achieve a secure key rate over a communication loss of about 28 dB. This demonstrates that while some performance is sacrificed, it is in exchange for a significantly higher level of implementation security by removing all side-channels associated with source modulators.
	
	\begin{figure}[H]
		\centering
		\includegraphics[width=0.82\linewidth]{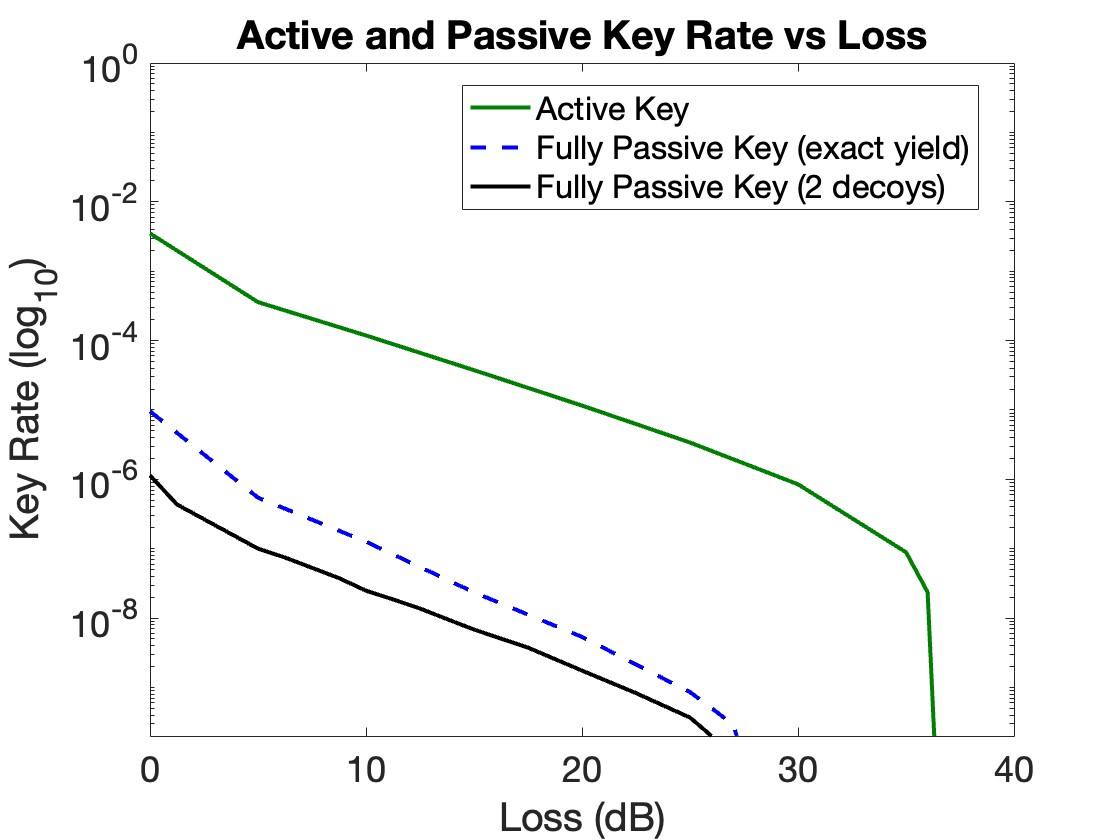}
		\caption{
Asymptotic secret-key rate versus channel loss for a four-user
conference-key-agreement network. The performance of the fully passive architecture is compared with that
of an actively modulated CKA scheme under identical system parameters. For consistency, the active protocol’s optimized intensities are reused
in the passive simulation. The passive model employs $M=8$ phase partitions and incorporates the
branch-cutting refinement for postselection. All results assume a detector dark-count probability of $10^{-8}$. 
		}
		\label{fig:rate_vs_loss}
	\end{figure}

	\subsection{Efficiency and Protocol Advantages}
	\label{subsec:efficiency}
	The fully passive CKA protocol offers several important advantages over active schemes, particularly concerning efficiency in a multi-user setting~\cite{Wang2023,Li2024CKA}.
	\begin{itemize}
		\item Absence of Basis Sifting: A significant advantage is the concept of ``late basis choice." Since the physical signals are identical for every round, the decision to use a round for key generation (KG) or parameter estimation (PE) is made classically after all quantum communication is finished. This completely eliminates the need for basis sifting among the users, a process that becomes increasingly inefficient as the number of users, $N_U$, grows.
		\item Resistance to Misalignment: As the final key rate is calculated by summing the contributions from all possible slice combinations, the protocol is highly robust against phase misalignment. If a small misalignment causes the phase of a signal to shift from one slice to a neighboring one, its contribution to the key rate is not lost. Instead, it is simply accounted for in a different term of the overall sum. This ``soft-landing" behavior ensures that the total key rate remains stable even with practical imperfections.
	\end{itemize}
	
	\subsection{Computational Acceleration: The Branch Cutting Method}
	\label{subsec:branchcut}

	Calculating the total key rate requires summing over all $M^{2N_U}$ possible combinations of phase slices. For a four-user protocol with M=8 slices, this amounts to $8^8$ terms, which is computationally prohibitive. To make the simulation feasible, I developed a ``branch cutting" method to intelligently prune the summation~\cite{Li2024CKA}.
	
	The key insight is that combinations where the phases are highly mismatched contribute negligibly to the final key rate. I therefore discard any slice combination that meets one of the following criteria:
	\begin{itemize}
		\item Large intra-user phase mismatch: The difference between a single user's two measured phase slices is too large. Formally, for any user $i$, I require $|k_{i1} - k_{i2}| \le x$.
		\item Large inter-user phase mismatch: The difference between the average phase of any two users is too large. Formally, for any pair of users $i$ and $j$, I require $|\frac{k_{i1}+k_{i2}}{2} - \frac{k_{j1}+k_{j2}}{2}| \le y$.
	\end{itemize}
	In my simulation, I used thresholds of $x=2$ and $y=2$ for $M=8$ slices. This method dramatically reduces the number of computations required while having a negligible impact on the calculated key rate, as it only removes terms that would have contributed near-zero key.
	
		A formal definition of the branch-cutting criteria used in the simulation
	is provided in Appendix~\ref{sec:cka_branch}.

	\subsection{Outlook}
	\label{subsec:practical}
	This work serves as a proof-of-concept for applying the fully passive architecture to multipartite QKD. Several avenues remain for future work and practical implementation.
	Future works may include: The key rate performance, particularly for the decoy-state case, could be improved by implementing a three-decoy (or more) analysis to obtain tighter bounds on the yields. With access to greater computational resources, the branch cutting thresholds could be relaxed, and a full optimization of the input intensity and number of slices (M) could be performed. Extending the analysis to larger numbers of users and conducting a full finite-size security analysis are also important next steps.

	\chapter{QUANTUM SPEED LIMITS: BACKGROUND AND THEORETICAL FOUNDATIONS}
	\label{chap:qsl_background}
	
	\section{The Fundamental Question}
	\label{sec:qsl_question}
	
	How fast can a physical system evolve? Can the evolution time be arbitrarily short? This is one of the most fundamental questions in physics. Quantum speed limits (QSLs) provide the answer by establishing rigorous lower bounds on the minimal time, $\tau_{min}$, required for a quantum system to evolve from an initial state $\rho_0$ to a target state $\rho_\tau$~\cite{Deffner17a}.
	
	Unlike relativistic constraints (e.g., the speed of light) which limit the propagation of information in spacetime, QSLs concern the rate of change in the system's abstract state space. They quantify the ultimate performance limits for any quantum process, thereby constraining the ``clock speed" of quantum computers~\cite{Lloyd00}, setting precision-time tradeoffs in quantum metrology~\cite{Giovannetti04}, and defining power constraints in quantum thermodynamics~\cite{Binder2015}. A rigorous understanding of QSLs is therefore indispensable both for foundational insight and for setting practical, physically impassable benchmarks for quantum technologies given the value of certain observable of the system.
	
	At a high level, most QSLs can be written in the schematic form
	\begin{equation}
		\label{eq:qsl_schematic}
		\tau \ge \frac{\text{distance}(\rho_0, \rho_\tau)}{\text{average dynamical speed.}}
	\end{equation}
	Here, the ``distance'' is a measure of distinctness or distinguishability between the initial and final states. It need not be a formal metric. The ``speed" is set by a suitable functional of the generator of dynamics (such as its energy variance or mean, or, in more general formulations, a norm of its action on the state). The specific choice of distance and speed functional determines the particular bound and its regime of applicability. This chapter develops the historical context, core concepts, and modern formulations necessary to situate the new class of speed limits introduced in the next chapter.

	\section{Historical Foundations}
	\label{sec:qsl_history}
	
	\subsection{From Time-Energy Uncertainty to Minimal-Time Bounds}
	\label{subsec:qsl_mt_ml}
	
	The earliest QSLs were derived for closed systems evolving unitarily under a time-independent Hamiltonian, $H$. These bounds are often seen as a formal expression of the time-energy uncertainty principle.
	
	In 1945, Mandelstam and Tamm (MT) derived a bound based on the system's energy variance. For a pure state $|\psi_t\rangle$ evolving to an orthogonal state, the minimal time $\tau$ is bounded by~\cite{Mandelstam45}:
	\begin{equation}
		\label{eq:mt_bound}
		\tau \ge \frac{\pi\hbar}{2\Delta E},
	\end{equation}
	where $\Delta E^2 \equiv \langle\psi_0|H^2|\psi_0\rangle - \langle\psi_0|H|\psi_0\rangle^2$ is the variance of the system's energy. This bound is state-dependent and is tightest for states with a large energy spread.
	
	Over half a century later, Margolus and Levitin (ML) derived a complementary bound, this time dependent on the system's mean energy relative to its ground state energy, $E_0$~\cite{Margolus98}:
	\begin{equation}
		\label{eq:ml_bound}
		\tau \ge \frac{\pi\hbar}{2(\langle H \rangle - E_0)},
	\end{equation}
	where $\langle H \rangle = \langle\psi_0|H|\psi_0\rangle$. This bound is most restrictive for systems with a high average energy. For a closed, time-independent system, the true minimal evolution time is constrained by the tighter (larger) of these two bounds, leading to the unified QSL~\cite{Levitin2009}:
	\begin{equation}
		\label{eq:unified_closed}
		\tau \ge \max\left[ \frac{\pi\hbar}{2\Delta E}, \frac{\pi\hbar}{2(\langle H \rangle - E_0)} \right].
	\end{equation}
These foundational results crystallized the connection between evolution time and energy, launching a line of inquiry that now encompasses general open dynamics and time-dependent control~\cite{Giovannetti03,Deffner13,Deffner17a}.

	\subsection{A Geometric Interpretation of Evolution Speed}
	\label{subsec:qsl_geometric}

	\begin{figure}[t]
		\centering
		\includegraphics[width=0.55\textwidth]{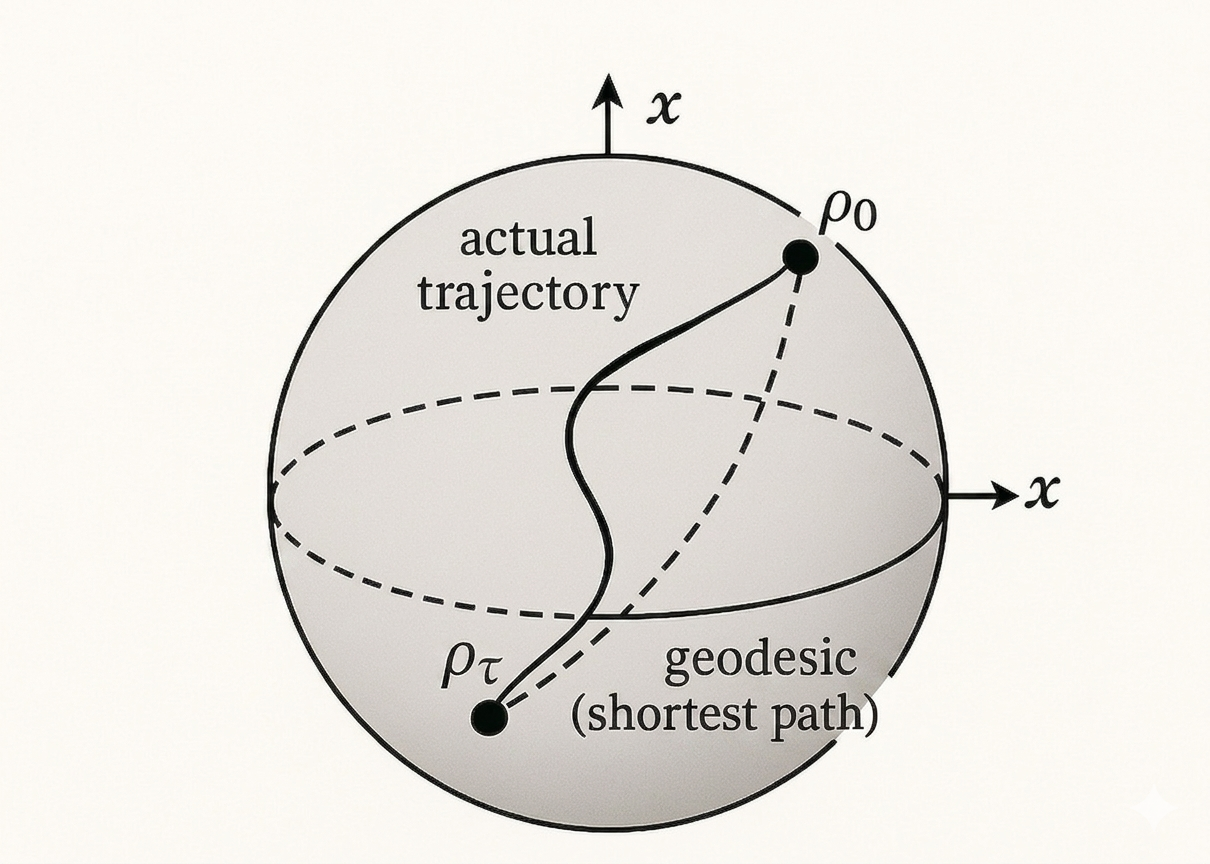}
		\caption{
			\textbf{Geometric interpretation of the Quantum Speed Limit on the Bloch sphere.}
			The system evolves from an initial state $\rho_0$ to a target state $\rho_\tau$ along
			its actual dynamical trajectory (solid curve). 
			The QSL follows from the requirement that the length of this actual path must be at
			least the geodesic distance (dashed curve), representing the minimal distance between
			states under the chosen quantum metric (e.g., Bures angle).
		}
		\label{fig:qsl_geodesic}
	\end{figure}

	A more general and intuitive perspective was introduced by Anandan and Aharonov, who reframed unitary dynamics as a path through the projective Hilbert space. The ``speed" of evolution is the rate at which the state vector traces this path~\cite{Anandan90}.
	
	For pure states, the natural measure of distance is the Fubini-Study metric. The (infinitesimal) squared distance $ds^2$ between a state $|\psi\rangle$ and a nearby state $|\psi + d\psi\rangle$ is given by
	\begin{equation}
		\label{eq:fs_metric}
		ds^2 = 4 \left( \langle d\psi | d\psi \rangle - |\langle \psi | d\psi \rangle|^2 \right),
	\end{equation}
	assuming $\langle \psi | \psi \rangle = 1$. The total distance $S$ between two arbitrary pure states $|\psi_0\rangle$ and $|\psi_\tau\rangle$ is the length of the shortest path (geodesic) connecting them, given by $S = 2 \arccos(|\langle \psi_0 | \psi_\tau \rangle|)$. The MT bound can be re-derived directly from this. Using the Schrödinger equation, the instantaneous speed of evolution $v(t) = ds/dt$ along the actual path is found to be~\cite{Anandan90}:
	\begin{equation}
		\label{eq:fs_speed}
		v(t) = \frac{ds}{dt} = \frac{2}{\hbar} \sqrt{\langle H^2 \rangle - \langle H \rangle^2} = \frac{2\Delta E_t}{\hbar}.
	\end{equation}
	The total length $L$ of the path traced by the state in time $\tau$ is $L = \int_0^\tau v(t) dt$. Since the actual path length must be greater than or equal to the shortest distance $S$ ($L \ge S$), an evolution to an orthogonal state ($S = \pi$) requires $\int_0^\tau (2\Delta E_t/\hbar) dt \ge \pi$. For a time-independent $H$, this simplifies to $\tau \ge \pi\hbar / (2\Delta E)$, which is the MT bound.
	
	The power of this geometric picture, however, is fundamentally limited to pure states. The Fubini-Study metric is ``only" defined for the projective Hilbert space (the space of pure states). For a qubit, this corresponds to the ``surface" of the Bloch sphere. An open system, however, is subject to dissipation and dephasing, which causes a pure state to evolve into a ``mixed state". This corresponds to a trajectory that leaves the surface of the Bloch sphere and moves into its ``interior". The Fubini-Study metric is not defined for this trajectory, as it cannot measure the distance between a point on the surface (the pure initial state) and a point in the interior (the mixed final state).
	
	This challenge is solved by finding a metric that (a) is defined for \textit{all} states (pure and mixed) and (b) is \textit{contractive} under the dynamics. A metric is contractive if the distance between any two states can only decrease (or stay the same) under noisy evolution: $D(\Phi(\rho_1), \Phi(\rho_2)) \le D(\rho_1, \rho_2)$, where $\Phi$ is a completely positive trace-preserving (CPTP) map. This is a crucial physical requirement, as it ensures that noisy processes, which make states ``less" distinguishable, do not appear to increase the distance between them~\cite{Pires2016}. The most common choice that satisfies these conditions is the Bures angle, $\Theta_{Bures}$, which is a true metric that defines its own geodesics. A \textbf{geodesic} is simply the shortest possible path between two points in a given space (e.g., a straight line on a flat plane, or a ``great circle" route on the curved surface of the Earth). The Bures angle is defined via the Uhlmann fidelity:
\begin{equation}
	\label{eq:bures_angle}
	\Theta_{Bures}(\rho_0, \rho_\tau) \equiv 
	\arccos\left( \Tr \sqrt{\sqrt{\rho_0} \rho_\tau \sqrt{\rho_0}} \right),
\end{equation}
which is the Bures angle associated with the Uhlmann fidelity~\cite{Braunstein96,Pires2016}.

	This angle is the proper generalization of the Fubini-Study distance for mixed states. In this view, the schematic QSL from Section \ref{sec:qsl_question} becomes:
	\begin{equation}
		\label{eq:qsl_geometric_form}
		\tau \ge \frac{\Theta_{Bures}(\rho_0, \rho_\tau)}{\bar{v}},
	\end{equation}
	which is the natural geometric form of a quantum speed limit in terms of the Bures angle and the time-averaged speed~\cite{Pires2016,Deffner17a}.
	
	Here, $\Theta_{Bures}(\rho_0, \rho_\tau)$ is the geodesic distance (the shortest path)
	between the initial and final states in the full space of density matrices
	(see Fig.~\ref{fig:qsl_geodesic}).
	 $\bar{v} = (1/\tau) \int_0^\tau v(t) dt$ is the time-averaged speed of the system's \textit{actual} path, as measured by the Bures metric. The central task of open-system QSLs, as will be discussed, is to find a bound for this instantaneous speed $v(t)$ (e.g., by using a norm of the dynamical generator $\mathcal{L}$) in order to make this inequality a computable bound.
	
\section[From Idealized to Realistic Dynamics: Open and Driven Systems]{From Idealized to Realistic Dynamics: \\ Open and Driven Systems}
	\label{sec:qsl_open_driven}
	
	Practical quantum devices are not closed or time-independent. They are \textit{open} systems that inevitably interact with an uncontrolled environment, and they are often \textit{driven} by time-dependent control fields, $H(t)$. Modern QSLs must therefore address two intertwined challenges:
	(i) identifying a suitable contractive distance in state space, and (ii) quantifying the dynamical speed in the presence of dissipation, dephasing, and time-dependent generators.
	
	\subsection{Operator-Norm Formulation for Open Systems}
	\label{subsec:qsl_operator_norm}
	A significant breakthrough came from extending the geometric approach to open systems described by a master equation, $\dot{\rho}_t = \mathcal{L}(\rho_t)$, where $\mathcal{L}$ is the (possibly time-dependent) generator of the dynamics (e.g., a Lindbladian). For such driven systems, the instantaneous speed of evolution $v(t)$ is no longer constant, and the QSL must be expressed as an integral over this changing speed.
	
	Deffner and Lutz bounded this instantaneous speed $v(t)$ (derived from the Bures angle) by relating it to standard operator norms of the generator's action on the state, $\mathcal{L}(\rho_t)$~\cite{Deffner13}. Integrating this speed leads to a powerful family of bounds. By defining the time-averaged norms:
	\begin{equation}
		\label{eq:avg_norm}
		\Lambda_{\tau}^{k} = \frac{1}{\tau} \int_0^\tau ||\mathcal{L}(\rho_t)||_k dt,
	\end{equation}
	where $k$ can be the operator norm ($op$), trace norm ($tr$), or Hilbert-Schmidt norm ($hs$), they derived a unified QSL for open systems~\cite{Deffner13}:
	\begin{equation}
		\label{eq:deffner_lutz}
		\tau \ge \max\left\{ \frac{1}{\Lambda_{\tau}^{op}}, \frac{1}{\Lambda_{\tau}^{tr}}, \frac{1}{\Lambda_{\tau}^{hs}} \right\} \sin^2 \Theta_{Bures}(\rho_0, \rho_\tau).
	\end{equation}
	This formulation explicitly generalizes both the MT and ML bounds. The Mandelstam-Tamm (MT) type bound is the one involving the Hilbert-Schmidt norm, $||\cdot||_{hs}$. The connection is not an approximation, for unitary evolution of a pure state $\rho_t = |\psi_t\rangle\langle\psi_t|$, the generator is $\mathcal{L}(\rho_t) = \dot{\rho}_t = -\frac{i}{\hbar}[H, \rho_t]$, and the Hilbert-Schmidt norm of this generator is exactly proportional to the energy variance:
\begin{align}
	||\mathcal{L}(\rho_t)||_{hs}^2 
	&= \Tr(\dot{\rho}_t^\dagger \dot{\rho}_t) = \Tr(\dot{\rho}_t^2) = - \frac{1}{\hbar^2}\Tr([H, \rho_t]^2) \notag \\
	&= \frac{2}{\hbar^2}(\langle H^2 \rangle_t - \langle H \rangle_t^2) = \frac{2}{\hbar^2}(\Delta E_t)^2 \label{eq:hs_variance_link}.
\end{align}
This family of bounds has become a standard benchmark for assessing performance in realistic open-system scenarios, particularly in the regime of quantum optimal control where QSLs set in-principle limits on achievable gate times~\cite{Caneva09,Campbell17,Deffner17a}.

	\subsection{Information-Geometric Formulation \\ (Quantum Fisher Information)}
	\label{subsec:qsl_qfi}
	
	A complementary formulation of quantum speed limits arises from the viewpoint of \emph{quantum information geometry}.  
	Instead of quantifying the speed of evolution using operator norms, one can describe it using the intrinsic metric structure of quantum state space—the \textbf{Quantum Fisher Information (QFI)}~\cite{Taddei13,Pires2016}.
	
	\paragraph{From local to global geometry.}
	The QFI, $\mathcal{F}_Q$, defines the \emph{infinitesimal} distance between two neighboring quantum states along an evolution $\rho_t$.  
	Mathematically, the QFI is defined via the \textbf{Symmetric Logarithmic Derivative (SLD)}, $L_{\rho_t}$, which is the Hermitian operator satisfying the Lyapunov equation:
	\begin{equation}
		\mathcal{L}(\rho_t) = \frac{1}{2}(L_{\rho_t} \rho_t + \rho_t L_{\rho_t}) ,
	\end{equation}
	where $\mathcal{L}(\rho_t) \equiv \dot{\rho}_t$ is the generator of the dynamics. The notation $\mathcal{F}_Q[\rho_t, \mathcal{L}(\rho_t)]$ specifies the QFI as a functional of both the current state $\rho_t$ and the generator $\mathcal{L}(\rho_t)$, calculated as:
	\begin{equation}
		\mathcal{F}_Q[\rho_t, \mathcal{L}(\rho_t)] = \mathrm{Tr}(\rho_t L_{\rho_t}^2) .
	\end{equation}
	
	If $\rho_t$ and $\rho_{t+dt}$ are infinitesimally close, the so-called Bures line element satisfies
\begin{equation}
	ds_{\mathrm{Bures}}^2 = \frac{1}{4}\,\mathcal{F}_Q[\rho_t, \mathcal{L}(\rho_t)]\, dt^2 ,
	\label{eq:bures-element}
\end{equation}
as dictated by the quantum Fisher information metric on state space~\cite{Braunstein96,Taddei13,Pires2016}.

	Integrating this local distance along the entire trajectory gives the total \emph{Bures length} between $\rho_0$ and $\rho_\tau$,
	\begin{equation}
		\mathcal{L}_{\mathrm{Bures}}(\rho_0,\rho_\tau)
		= \int_0^\tau \sqrt{\tfrac{1}{4}\,\mathcal{F}_Q[\rho_t, \mathcal{L}(\rho_t)]}\, dt ,
	\end{equation}
	whose geodesic value equals the Bures angle $\Theta_{\mathrm{Bures}}(\rho_0,\rho_\tau)$.
	In other words,
	\begin{itemize}
		\item the \textbf{QFI} describes the \emph{local curvature} or ``instantaneous speed'' of motion in state space, and  
		\item the \textbf{Bures angle} measures the \emph{total distance} accumulated along the entire path.
	\end{itemize}

	\paragraph{How they enter a QSL.}
	Because the Bures angle quantifies the total distance to be covered, it naturally appears in the \emph{numerator} of a QSL inequality. The QFI, by contrast, determines the instantaneous speed of evolution, and its time average sets the \emph{denominator}. Explicitly, defining the instantaneous ``Bures velocity''
	\begin{equation}
		v_{\mathcal{F}_Q}(t) \equiv 
		\sqrt{\tfrac{1}{4}\,\mathcal{F}_Q[\rho_t, \mathcal{L}(\rho_t)]},
	\end{equation}
	and its time average
	\begin{equation}
		\bar{v}_{\mathcal{F}_Q} = \frac{1}{\tau}\int_0^\tau v_{\mathcal{F}_Q}(t)\,dt,
	\end{equation}
	the information-geometric quantum speed limit takes the compact form~\cite{Pires2016}
	\begin{equation}
		\tau \;\ge\;
		\frac{\Theta_{\mathrm{Bures}}(\rho_0,\rho_\tau)}{\bar{v}_{\mathcal{F}_Q}} .
		\label{eq:qfi_bound}
	\end{equation}

	Hence, a large QFI corresponds to a higher ``statistical speed'' and therefore allows faster evolution. This connects the theory of QSLs directly to quantum metrology: the same QFI that governs the quantum Cramér–Rao bound,
	\((\Delta\theta)^2 \ge 1/\mathcal{F}_Q\),
	also quantifies the local velocity of state change in time~\cite{Braunstein96,Giovannetti2011}. Fast-evolving states are precisely those that can carry more information about a parameter—linking the limits of dynamics and of precision in a single geometric framework.

	{\small
		\renewcommand{\arraystretch}{1.5}
		\begin{longtable}{@{} p{14cm} @{}}
			\caption{Summary of major quantum speed limit (QSL) bounds discussed in this chapter.}
			\label{tab:qsl_summary} \\
			\hline
			\endfirsthead
			\hline
			\endhead
			
			\textbf{Mandelstam--Tamm (1945)~\cite{Mandelstam45}} \\
			\textbf{Basic Formula:}
			$$\tau \ge \frac{\pi\hbar}{2\,\Delta E}$$
			\textbf{Key Notes:}
			\vspace*{-0.5em}
			\begin{itemize}\setlength\itemsep{1pt}
				\item Closed, time-independent systems.
				\item Based on energy variance $\Delta E$.
				\item Fundamental lower bound for unitary evolution.
			\end{itemize} \\
			\hline
			
			\textbf{Margolus--Levitin (1998)~\cite{Margolus98}} \\
			\textbf{Basic Formula:}
			$$\tau \ge \frac{\pi\hbar}{2\,(\langle H\rangle - E_0)}$$
			\textbf{Key Notes:}
			\vspace*{-0.5em}
			\begin{itemize}\setlength\itemsep{1pt}
				\item Uses mean excitation energy.
				\item Complements Mandelstam–Tamm bound.
				\item Closed, time-independent dynamics.
			\end{itemize} \\
			\hline
			
			\textbf{Anandan--Aharonov (1990)~\cite{Anandan90}} \\
			\textbf{Basic Formula:}
			$$\tau \ge \frac{S_{\mathrm{FS}}}{\bar{v}}, \quad S_{\mathrm{FS}}=2\arccos|\langle\psi_0|\psi_\tau\rangle|$$
			\textbf{Key Notes:}
			\vspace*{-0.5em}
			\begin{itemize}\setlength\itemsep{1pt}
				\item Geometric reformulation via Fubini–Study metric.
				\item Interprets $\Delta E/\hbar$ as evolution speed.
				\item Equivalent to MT bound for pure states.
			\end{itemize} \\
			\hline
			
			\textbf{Deffner--Lutz (2013)~\cite{Deffner13}} \\
			\textbf{Basic Formula:} \\
			\textbf{(General unified form)}
			$$\tau \ge \max\!\Bigg\{ \frac{1}{\Lambda_{\tau}^{op}}, \frac{1}{\Lambda_{\tau}^{tr}}, \frac{1}{\Lambda_{\tau}^{hs}} \Bigg\} \sin^2\!\Theta_{\mathrm{Bures}}(\rho_0,\rho_\tau)$$
			\textbf{(MT-type)}
			$$\tau_{\mathrm{MT}}^{DL} \ge \frac{\sin^2\!\Theta_{\mathrm{Bures}}}{\Lambda_{\tau}^{hs}}$$
			\textbf{(ML-type)}
			$$\tau_{\mathrm{ML}}^{DL} \ge \frac{\sin^2\!\Theta_{\mathrm{Bures}}}{\Lambda_{\tau}^{op}}$$
			\textbf{Key Notes:}
			\vspace*{-0.5em}
			\begin{itemize}\setlength\itemsep{1pt}
				\item Applies to open and driven systems.
				\item Extends MT/ML to mixed-state Lindblad dynamics.
				\item MT-type uses Hilbert–Schmidt norm ($\propto\Delta E$).
				\item ML-type uses operator norm (mean-energy analog).
			\end{itemize} \\
			\hline
			
			\textbf{Pires \emph{et al.} (2016)~\cite{Pires2016}} \\
			\textbf{Basic Formula:}
			$$\tau \ge \frac{\Theta_{\mathrm{Bures}}(\rho_0,\rho_\tau)}{\bar{v}_{\mathcal{F}_Q}}, \quad \bar{v}_{\mathcal{F}_Q}= \frac{1}{\tau}\!\int_0^\tau \!\sqrt{\tfrac{1}{4}\mathcal{F}_Q[\rho_t,\mathcal{L}(\rho_t)]}\,dt$$
			\textbf{Key Notes:}
			\vspace*{-0.5em}
			\begin{itemize}\setlength\itemsep{1pt}
				\item QFI-based, information-geometric formulation.
				\item Connects QSLs with quantum metrology limits.
				\item Reduces to geometric MT for unitary evolution.
			\end{itemize} \\
			\hline
			\textbf{Giovannetti--Lloyd--Maccone (2003)~\cite{10.1117/12.507486}} \\
			\textbf{Basic Formula:}
			$$\tau \ge \max\left\{ \frac{\pi\hbar}{2\Delta E}, \frac{\pi\hbar}{2(\langle H\rangle - E_0)} \right\}$$
			\textbf{Key Notes:}
			\vspace*{-0.5em}
			\begin{itemize}\setlength\itemsep{1pt}
				\item \textbf{Unification:} Formally unifies the Mandelstam--Tamm (variance-based) and Margolus--Levitin (mean-energy-based) bounds into a single, tighter limit.
				\item \textbf{Mixed State Extension:} Extends the QSL framework to mixed states, moving beyond the pure-state constraints of earlier geometric formulations.
				\item \textbf{Multipartite Systems:} Establishes speed limits for complex, multipartite quantum systems, providing a fundamental bound for quantum information processing.
				\item \textbf{Entanglement Dynamics:} Specifically analyzes how quantum entanglement and correlations can be leveraged to enhance (speed up) dynamical evolution.
			\end{itemize} \\
			\hline

		\end{longtable}
	} 

	\section{Physical Significance and Applications}
	\label{sec:qsl_applications}
	
	QSLs have matured from formal inequalities into operational benchmarks with broad, practical implications.
	\begin{itemize}
		\item \textbf{Quantum Computing and Quantum Control:} QSLs impose the ultimate minimal-time constraint on high-fidelity quantum gates and state-transfer protocols~\cite{Caneva09}. They define the ``clock speed" of a quantum processor, separating limitations due to engineering (e.g., control pulse quality) from fundamental physical limits~\cite{Lloyd00}.
		
		\item \textbf{Quantum Metrology:} In sensing and metrology, QSLs quantify the fundamental tradeoff between interrogation time $\tau$ and estimation precision $\Delta\theta$ (e.g., in phase accumulation)~\cite{Giovannetti2011}. This leads to generalized time-energy uncertainty relations that limit the sensitivity of clocks, gravimeters, and other quantum sensors~\cite{Pang2017}.
		
		\item \textbf{Quantum Thermodynamics:} In the study of quantum engines and work extraction, QSLs constrain the maximum power output~\cite{Binder2015}. Thermodynamic processes driven at high speeds (approaching the QSL) incur costs, such as friction or non-adiabatic excitations, that limit efficiency. This connects QSLs to the ``shortcuts to adiabaticity" research program~\cite{Campbell17}.
		
		\item \textbf{Quantum Batteries:} In models of energy storage, many-body correlations can potentially lead to a ``super-extensive" charging speed. QSLs provide the ultimate cap on this charging power, guiding the design and assessing the feasibility of quantum batteries.
	\end{itemize}
	
	\section{Summary}
	\label{sec:qsl_summary}
	
	Quantum speed limits (QSLs) capture the ultimate temporal constraints imposed by quantum mechanics, linking the geometry of state space with the dynamical generator of evolution. From the early MT and ML bounds to geometric, information-based, and open-system formulations, the field has evolved into a unified framework describing how energy, coherence, and distinguishability jointly govern the speed of quantum change~\cite{Giovannetti03,Deffner13,Pires2016,Deffner17a}.

	Yet, despite this progress, existing formulations often rely on fixed norm distances and metrics that can be rigid or loose in realistic, representation-dependent dynamics. This tension between universality and tightness motivates the search for a more flexible formulation of QSLs—one that remains general but adapts naturally to the structure of a given system.

	The next chapter develops such a formulation by introducing a generalized norm-based framework that extends the traditional geometric picture. This new approach builds a family of QSL that is applicable universally. That means it can be applied to open, closed, time-dependent, time-independent systems, as well as for any linear operators. The family of QSL also involves tunability, there will a set of free parameters that can be tuned to match different quantum systems. Lastly, as my next chapter will demonstrate, the QSL can be computed efficiently.

		\newtheorem{theorem}{Theorem}[chapter] 
	\chapter{QUANTUM SPEED LIMITS BASED ON REPRESENTATION-DEPENDENT $\ell_{w}^{p}$-SEMINORMS}
	\label{chap:my_work_seminorm}
\section*{Statement of Contribution}

This chapter is based on the research presented in the following works:
\begin{quote}
	H. F. Chau and J. Li, ``Quantum Speed Limits For Open System Dynamics Based On Representation Basis Dependent $\ell_{w}^{p}$-Seminorm,'' arXiv preprint arXiv:2508.17053, 2025~\cite{Chau2025}.
\end{quote}
\begin{quote}
	J. Li and H. F. Chau, ``A New Quantum Speed Limit Based on Unitarily Invariant Norms,'' Contributed talk at the 25th Asian Quantum Information Science (AQIS) Conference, 2025~\cite{LiChau_AQIS2025}.
\end{quote}
\begin{quote}
	H. F. Chau and J. Li, ``Quantum Speed Limits for Open System Dynamics Based on a Representation-Basis-Dependent $\ell_{w}^{p}$-Seminorm,'' \emph{Phys. Rev. A} \textbf{113}, 042218, 2026~\cite{trss-qhc8}.
\end{quote}

This work was a collaboration between the author and Prof. H. F. Chau. The authors jointly identified the limitations of existing norm-based speed limits regarding universality and tightness, and proposed the fundamental concept of using representation-basis-dependent weighted norms to construct a more flexible and tunable QSL framework.

The author's contributions to this chapter are as follows:

The author derived the initial supremum form of the inequality based on an idea proposed by Prof. Chau. The author further introduced the concept of weighting factors in those inequalities and assisted in proving the associated mathematical properties.

Regarding the numerical analysis, the author and Prof. Chau jointly designed and performed the computations for the first case study involving qubit and qudit systems. The author independently designed and performed the numerical simulations for the remaining three case studies presented in Sec.~\ref{sec:seminorm_applications}: the spontaneously emitting qubit, the NV-center gate, and the quantumness measures of linear operators. 

Finally, the author demonstrated the universality of this QSL framework by extending the definitions to general linear operators.


\section{Introduction}
\label{sec:seminorm_intro}

In Chapter \ref{chap:qsl_background}, I reviewed the theoretical foundations of quantum speed limits (QSLs), from the foundational MT and ML bounds to modern open-system formulations~\cite{Mandelstam45,Margolus98,Levitin2009,Deffner13,Deffner17a}. As noted, a persistent challenge is to find a single QSL framework that is simultaneously \textit{universal} (applicable to all dynamics), \textit{tight} (providing a sharp, non-trivial bound), and \textit{computationally efficient}.

This chapter introduces such a framework, based on the author's work in Ref.~\cite{Chau2025,trss-qhc8}. I develop and analyze a new family of QSLs derived from a representation-basis-dependent weighted $\ell_{w}^{p}$-seminorm~\cite{Chau2025,trss-qhc8}. This approach provides a powerful and practical tool for benchmarking quantum dynamics, capable of setting limits on systems ranging from simple qubits to complex open-system environments.

This chapter is structured as follows: Section \ref{sec:seminorm_derivation} provides the complete mathematical derivation of the QSLs and analyzes their fundamental properties. Section \ref{sec:seminorm_applications} presents four distinct case studies to demonstrate the applicability and sharpness of the bounds. Finally, Section \ref{sec:seminorm_conclusion} summarizes the findings and discusses an outlook for future work.
	
	\section{The $\ell_{w}^{p,max}$-Norm QSLs and \\Their Properties}
	\label{sec:seminorm_derivation}
	
	
	\subsection{My QSLs and Their Proofs}
	\label{subsec:seminorm_proof}
	My derivation begins by representing the $n$-dimensional quantum state $\rho(\tau)$ at time $\tau$ in a fixed orthonormal basis $\mathcal{B}$~\cite{vonNeumann1932,Fano1957}. I ``vectorize" this operator into an $n^2$-dimensional complex vector by column-stacking \cite{Watrous_2018}:
	\begin{equation}
		vec_{\mathcal{B}}(\rho(\tau)) = (\rho(\tau)_{11}, \rho(\tau)_{12}, \dots, \rho(\tau)_{nn})^T \in \mathbb{C}^{n^2}.
	\end{equation}
	To quantify the change between the initial state $\rho_0$ and the state at time $\tau$, $\rho_\tau$, I define a time-dependent difference vector:
	\begin{equation}
		\Delta_{\mathcal{B}}(\tau) = vec_{\mathcal{B}}(\rho_{\tau}) - vec_{\mathcal{B}}(\rho_{0}).
	\end{equation}
	For clarity, I will omit the $\mathcal{B}$ subscript when the basis is clear from context or unimportant.
	
	I then equip this vector space with a specific norm derived from the weighted $\ell_w^p$-seminorm. For any $p \ge 1$ and a non-zero, non-negative real vector $w \in \mathbb{R}^{n^2}$, I define the $\ell_{w}^{p,max}$ norm as:
	\begin{align}
		\mathcal{D}_{p,w,\mathcal{B}}^{max}(\rho_{\tau},\rho_{0}) &= ||\Delta_{\mathcal{B}}(\tau)||_{p,w}^{max} \equiv \max_{P \in S_{n^2}} \left( \sum_{j=1}^{n^2} w_{P(j)} |\Delta_{\mathcal{B},j}(\tau)|^p \right)^{1/p} \label{eq:my_qsl_norm_perm} \\
		&= \left[ \sum_{j=1}^{n^2} w_{j}^{\downarrow} (|\Delta_{\mathcal{B}}(\tau)|_{j}^{\downarrow})^p \right]^{1/p} \label{eq:my_qsl_norm_sorted}.
	\end{align}
	Here, $S_{n^2}$ is the symmetric group on $n^2$ elements, and the $\downarrow$ superscript denotes that the elements of the vector are arranged in descending order. This quantity is a true norm even if $\ell_w^p$ is only a seminorm (i.e., if some elements of $w$ are zero)~\cite{Chau2025}. By continuity, for $p \to \infty$, I define $\mathcal{D}_{\infty,w,\mathcal{B}}^{max}(\rho_{\tau},\rho_{0}) = ||\Delta_{\mathcal{B}}(\tau)||_{\infty,w}^{max} = |\Delta_{\mathcal{B}}(\tau)|_{1}^{\downarrow}$.
	
	Let us assume $\mathcal{D}_{p,w,\mathcal{B}}^{max}(\rho_{\tau},\rho_{0}) \neq 0$. I define $\overline{w}$ as the vector $(w_{P(j)})_j$ that achieves the maximum in Eq. \ref{eq:my_qsl_norm_perm}, where $P$ is the maximizing permutation. This $\overline{w}$ may depend on $\Delta_{\mathcal{B}}(\tau)$ and need not be unique. I then define the $\ell_{\overline{w}}^p$ seminorm of the vector of absolute values, $|\Delta_{\mathcal{B}}(t)|$, as:
	\begin{equation}
		\mathcal{D}_{p,\overline{w},\mathcal{B}}(\rho_t,\rho_0) = \left( \sum_{j=1}^{n^2} \overline{w}_j |\Delta_{\mathcal{B},j}(t)|^p \right)^{1/p}.
	\end{equation}
	Treating $\mathcal{D}_{p,w,\mathcal{B}}^{max}$ and $\mathcal{D}_{p,\overline{w},\mathcal{B}}$ as functions of time $t$, I find $\frac{d}{dt}\mathcal{D}_{p,w,\mathcal{B}}^{max} = \frac{d}{dt}\mathcal{D}_{p,\overline{w},\mathcal{B}}$. For $p > 1$, this derivative is:
	\begin{equation}
		\frac{d\mathcal{D}_{p,\overline{w},\mathcal{B}}}{dt} = \mathcal{D}_{p,\overline{w},\mathcal{B}}^{1-p} \sum_{j=1}^{n^2} \overline{w}_j |\Delta_{\mathcal{B},j}|^{p-1} \text{sgn}(\Delta_{\mathcal{B},j}) \frac{d}{dt}vec_{\mathcal{B}}(\rho_t)_j.
	\end{equation}
	I now set $1/p + 1/q = 1$ and apply the Hölder inequality to the vectors $u_j = \overline{w}_j^{1/p} \frac{d}{dt}vec_{\mathcal{B}}(\rho_t)_j$ and $v_j = \overline{w}_j^{1/q} |\Delta_{\mathcal{B},j}|^{p-1} \text{sgn}(\Delta_{\mathcal{B},j})$. Multiplying the result by $\mathcal{D}_{p,\overline{w},\mathcal{B}}^{1-p}$ yields the key differential inequality:
	\begin{equation}
		\left| \frac{d\mathcal{D}_{p,\overline{w},\mathcal{B}}}{dt} \right| \le \left|\left| \frac{d}{dt}vec_{\mathcal{B}}(\rho_t) \right|\right|_{p,\overline{w}} = ||vec_{\mathcal{B}}(\mathcal{L}\rho_t)||_{p,\overline{w}},
	\end{equation}
	where $||\cdot||_{p,\overline{w}}$ is the $\ell_{\overline{w}}^p$ norm and $\mathcal{L}$ is the Lindblad super-operator governing the dynamics $\dot{\rho}_t = \mathcal{L}\rho_t$~\cite{Deffner17a}. By taking limits, this inequality is also shown to hold for $p=1$ and $p=\infty$.
	
	Integrating this expression from $t=0$ to $t=\tau$ gives my two fundamental QSL bounds~\cite{Chau2025}:
	\begin{align}
		\tau &\ge \tau_{p,w,\mathcal{B}}^{int} \equiv \frac{\mathcal{D}_{p,w,\mathcal{B}}^{max}(\rho_{\tau},\rho_{0})}{\frac{1}{\tau} \int_{0}^{\tau} ||vec_{\mathcal{B}}(\mathcal{L}\rho_{t})||_{p,\overline{w}} dt} \label{eq:thesis_qsl_int} ,\\
		\tau &\ge \tau_{p,w,\mathcal{B}}^{sup} \equiv \frac{\mathcal{D}_{p,w,\mathcal{B}}^{max}(\rho_{\tau},\rho_{0})}{\sup_{t\in[0,\tau]} ||vec_{\mathcal{B}}(\mathcal{L}\rho_{t})||_{p,\overline{w}}} \label{eq:thesis_qsl_sup}.
	\end{align}
	I refer to $\tau^{int}$ as the \textbf{basic integral form QSL} and $\tau^{sup}$ as the \textbf{basic supremum form QSL}.
	
	\paragraph{Several important remarks are in order:}
	\begin{enumerate}
		\item Both inequalities (Eq. \ref{eq:thesis_qsl_int} and \ref{eq:thesis_qsl_sup}) hold trivially if the initial and final states are identical, i.e., $\mathcal{D}_{p,w,\mathcal{B}}^{max}(\rho_{\tau},\rho_{0}) = 0$.
		\item These QSLs are, in general, \textit{not} unitarily invariant. A change of basis in $\rho_t$ corresponds to a change in $\mathcal{B}$ for $\mathcal{D}_{p,w,\mathcal{B}}^{max}$, which will change the value of the bound. This basis-dependence is a key feature distinguishing my bounds from others~\cite{Pires2016,Taddei13,Giovannetti03}.
		\item The quantities $\mathcal{D}_{p,w,\mathcal{B}}^{max}(\rho_t,\rho_0)$ (the distance) and $||vec_{\mathcal{B}}(\mathcal{L}\rho_t)||_{p,\overline{w}}$ (the speed) are, in principle, observables that can be determined at any time $t$, though this may require quantum state tomography or POVMs~\cite{NielsenChuang}.
		\item The bounds are invariant to a positive scaling of the weight vector, $w \to \lambda w$ for $\lambda > 0$. I can thus normalize $w$ without loss of generality (e.g., $w_1^{\downarrow} = 1$). Furthermore, if $w$ has only one non-zero component ($w_j^{\downarrow} = 0$ for $j>1$), then Eq. \ref{eq:my_qsl_norm_sorted} shows the bound becomes independent of $p$.
		
		\item While $\tau^{int}$ generally provides a tighter bound than $\tau^{sup}$, calculating it requires full knowledge of the state $\rho_t$ for all $t \in [0, \tau]$ and performing an integration, which can be computationally intensive or experimentally unfeasible. In contrast, $\tau^{sup}$ is simpler, more robust to uncertainties in the evolution, and, as I will see, can still provide a strong QSL bound.
	\end{enumerate}
	
	These basic QSLs can be further optimized over the free parameters: the norm index $p$, the weight vector $w$, and the representation basis $\mathcal{B}$. I denote this (fully) optimized QSL as $\tau_{opt}^{int}$. If the norm $||vec(\mathcal{L}\rho_t)||_{p,\overline{w}}$ is sufficiently piecewise smooth, optimization over $p$ (for local maxima) is efficient. Similarly, optimization over the basis $\mathcal{B}$ is a variational problem on a compact Lie group, which can also be solved efficiently for local solutions. As I will show in Section \ref{subsubsec:seminorm_efficient_opt}, optimization over $w$ is also highly efficient. Therefore, finding the locally optimized bound $\tau_{opt}^{int}$ is computationally tractable~\cite{Boumal14,Chau2025}.
	
	It is critical to note that my derivation did not rely on any specific properties of density matrices (e.g., Hermiticity, unit trace, positivity). Consequently, my QSLs are valid for the evolution of \textit{any} linear operator, including non-Hermitian ones~\cite{Shrimali25}.
	
	Finally, this framework can be generalized to infinite-dimensional systems. I extend the definition of the norm (Eq. \ref{eq:my_qsl_norm_perm}) by considering an operator $A$ on $\mathcal{H}$ represented in a basis $\mathcal{B}$ with indices $x, y \in \mathbb{R}$. I equip $\mathbb{R}^2$ with a measure $\mu$ and define $||vec(A)||_{p,\mu} = [\int_{\mathbb{R}^2} |A(x,y)|^p d\mu]^{1/p}$. The $max$ norm becomes $||vec(A)||_{p,\mu}^{max} = \sup_{g \in G} [\int_{\mathbb{R}^2} |A(x,y)|^p d(g^{-1}\mu)]^{1/p}$, where $G$ is the category of all $\mu$-measure preserving maps on $\mathbb{R}^2$. The entire derivation carries over, and numerical computation can be performed using careful, uniformly convergent discretization of the system~\cite{Chau2025}.

	\subsection{Properties of My Basic QSLs}
	\label{subsec:seminorm_properties}
	
	\subsubsection{Sufficient Condition for My QSLs to be Representation Basis Independent}
	\label{subsubsec:seminorm_basis_independent}
	
	As noted in Remark 2, my QSLs $\tau^{int}$ and $\tau^{sup}$ are generally representation-basis-dependent. This is a crucial distinction from unitarily invariant bounds~\cite{Pires2016,Taddei13}. There is, however, one significant exception.
	
	Consider any linear operator $A$ acting on $\mathcal{H}^n$. The sum of the squared magnitudes of its matrix elements is equal to the trace of $A A^\dagger$, which is also the sum of its squared singular values, $s_j^\downarrow(A)$:
	\begin{equation}
		\sum_{j,k} |A_{jk}|^2 = \sum_{j,k} \langle j | A | k \rangle \langle k | A^\dagger | j \rangle = \text{Tr}(A A^\dagger) = \sum_j (s_j^\downarrow(A))^2.
	\end{equation}
	As a result, if I choose $p=2$ and $w = \mathbb{I}_{n^2}$ (a vector of $n^2$ ones), my $\ell_{w}^{p,max}$ norm is equivalent to the Frobenius norm $||\cdot||_2$:
	\begin{equation}
		||vec_{\mathcal{B}}(A)||_{2, \mathbb{I}_{n^2}} = \left( \sum_{j,k} |A_{jk}|^2 \right)^{1/2} = \left( \sum_j (s_j^\downarrow(A))^2 \right)^{1/2} = ||A||_2.
	\end{equation}
	Since the Frobenius norm (and the singular values) of an operator are independent of the basis, my QSL bounds in Eqs. \ref{eq:thesis_qsl_int} and \ref{eq:thesis_qsl_sup} become representation-basis-independent for the specific choice $p=2$ and $w = \mathbb{I}_{n^2}$~\cite{Chau2025}.
	
	In this specific case, for a sufficiently small time interval $[t, t+\Delta t]$, the generator's action can be approximated by the Hamiltonian component, $\mathcal{L}\rho_t = -i[H(t), \rho_t]/\hbar + O(\Delta t^2)$, where $H(t) = \sum_j E_j(t) |E_j(t)\rangle \langle E_j(t)|$. This allows us to relate the Frobenius norm of the generator to the energy variance $\Delta E_t$: $|| \mathcal{L}\rho_t ||_2 = \sqrt{2} \Delta E_t / \hbar$. My basis-independent QSLs then take the form:
	\begin{align}
		\tau_{2, \mathbb{I}_{n^2}}^{int} &= \frac{||\rho_{\tau} - \rho_0||_2}{\frac{1}{\tau} \int_0^\tau ||\mathcal{L}\rho_t||_2 dt} = \frac{\hbar ||\rho_{\tau} - \rho_0||_2}{\frac{\sqrt{2}}{\tau} \int_0^\tau \Delta E_t dt} \label{eq:thesis_frob_int}, \\
		\tau_{2, \mathbb{I}_{n^2}}^{sup} &= \frac{||\rho_{\tau} - \rho_0||_2}{\sup_{t\in[0,\tau]} ||\mathcal{L}\rho_t||_2} = \frac{\hbar ||\rho_{\tau} - \rho_0||_2}{\sqrt{2} \sup_{t\in[0,\tau]} \Delta E_t}. \label{eq:thesis_frob_sup}
	\end{align}
	This demonstrates a close formal relationship between my integral form QSL and the MT bound (Eq. \ref{eq:mt_bound})~\cite{Mandelstam45,Giovannetti03}. The key difference lies in the numerator: my bound uses the Frobenius distance, $||\rho_\tau - \rho_0||_2$, whereas the MT bound uses the Bures angle, $\Theta_{Bures}(\rho_\tau, \rho_0)$~\cite{Taddei13,Braunstein96}. This confirms my earlier remark that my QSLs are fundamentally different in nature from unitarily invariant bounds, even when their forms appear similar. This also distinguishes my bounds from others, as in the time-independent case, $||vec_{\mathcal{B}}(\mathcal{L}\rho_t)||_{p, \mathbb{I}_{n^2}}$ cannot generally be expressed in terms of $\langle |E|^p \rangle$~\cite{Pires2016}.
	
	One might wonder if other simple choices of $w$, such as $w = \mathbb{I}_1$ (a vector with one 1 and the rest 0), also lead to basis-independent bounds. By Ky Fan's maximum principle, $s_1^\downarrow(A)^2 = \max_{|j\rangle} \langle j | A A^\dagger | j \rangle = ||vec(A A^\dagger)||_{p, \mathbb{I}_1}^{max}$. This suggests $\tau_{p, \mathbb{I}_1}^{int}$ and $\tau_{p, \mathbb{I}_1}^{sup}$ could be basis-independent. However, this is not true. The state $|j\rangle$ that maximizes the ``distance" term $\mathcal{D}_{p,w,\mathcal{B}}^{max}(\rho_\tau, \rho_0)$ is not necessarily the same state that maximizes the ``speed" term $||vec_{\mathcal{B}}(\mathcal{L}\rho_t)||_{p, \overline{w}}$ at all times $t$. My numerical investigations confirm that the only basis-independent QSLs in my framework occur for $p=2$ and $w = \mathbb{I}_{n^2}$~\cite{Chau2025}.
	
	\subsubsection{Efficient Optimization Over $w$}
	\label{subsubsec:seminorm_efficient_opt}
	
	A significant practical advantage of my QSL framework is that the optimization over the weight vector $w$ can be performed very efficiently. For a fixed time $t$, representation basis $\mathcal{B}$, and $p \in [1, \infty)$, I can find the maximum of my QSL bounds (Eqs. \ref{eq:thesis_qsl_int} and \ref{eq:thesis_qsl_sup}) by checking only $n^2$ specific weight vectors. The $p=\infty$ case follows by continuity~\cite{Chau2025}.
	
	First, observe that my basic bounds are invariant to the ordering of the elements in $w$. Therefore, without loss of generality, I only need to consider non-negative weight vectors $w = (w_i)$ whose elements are arranged in descending order.
	
	Now, consider two such vectors, $w$ and $w'$, and their convex combination $\lambda w + \overline{\lambda} w'$, where $\lambda \in [0, 1]$ and $\overline{\lambda} = 1 - \lambda$. The right-hand side of both QSL bounds (Eq. \ref{eq:thesis_qsl_int} and \ref{eq:thesis_qsl_sup}) can be written in the functional form $f(\lambda) = \left( \frac{\lambda N + \overline{\lambda} N'}{\lambda D + \overline{\lambda} D'} \right)^{1/p}$, where $N, N' \ge 0$ and $D, D' > 0$. By elementary algebra, this function $f(\lambda)$ is quasi-convex (or quasi-concave, depending on the values) and always attains its maximum value at the boundaries of the domain $\lambda \in [0, 1]$, i.e., at $\lambda=0$ or $\lambda=1$.
	
	This implies that the maximum value of my QSLs must be attained by a weight vector $w$ that cannot be decomposed as a weighted sum of two other linearly independent, non-negative, descending-order vectors. The only vectors that satisfy this condition are the ``extremal" vectors $w = \mathbb{I}_j$ for $j = 1, 2, \dots, n^2$, where $\mathbb{I}_j$ is the vector whose first $j$ components are 1 and the remaining $(n^2 - j)$ components are 0.
	
	Consequently, to maximize my basic QSLs for a fixed basis $\mathcal{B}$ and parameter $p$, I do not need to perform a complex search over the continuous, high-dimensional space of $w$. I only need to evaluate the right-hand side of my bounds for these $n^2$ different extremal weight vectors $\mathbb{I}_j$. This result is crucial, as it demonstrates that I can improve the sharpness of my QSL by introducing the weightings $w$ while fully maintaining the computational efficiency of the optimization~\cite{Chau2025}.

	\subsubsection{Efficient Optimization Over $w$}
	\label{subsubsec:seminorm_efficient_opt}

	Another important advantage of my QSL framework is its efficient tunability. The tightest bound is found by optimizing over the parameters $p$, the basis $\mathcal{B}$, and the weight vector $w$. While optimizing $p$ and $\mathcal{B}$ involves standard numerical approaches, the optimization over the high-dimensional continuous parameter $w$ might seem computationally intractable. In this subsection, I prove that this is not the case: the optimal $w$ can be found by checking only $n^2$ discrete extremal vectors~\cite{Chau2025}.

	\begin{theorem}[Optimal Weight Vector]
		For a fixed basis $\mathcal{B}$, parameter $p \in [1, \infty)$, and time $\tau$, the maximum value of the QSL bounds $\tau_{p,w,\mathcal{B}}^{int}$ (Eq. \ref{eq:thesis_qsl_int}) and $\tau_{p,w,\mathcal{B}}^{sup}$ (Eq. \ref{eq:thesis_qsl_sup}) is attained when the weight vector $w$ is one of the $n^2$ extremal vectors $\mathbb{I}_j$, defined as the vector whose first $j$ components are 1 and the remaining $(n^2 - j)$ components are 0.
	\end{theorem}
	
	\begin{proof}
		The QSL bounds are invariant to the ordering of elements in $w$, so I can consider only non-negative vectors $w = (w_i)$ with elements in descending order. The set of such normalized vectors is convex.
		
		Let $w$ and $w'$ be two such vectors. The QSL bounds, as a function of their convex combination $\lambda w + \overline{\lambda} w'$ (where $\lambda \in [0, 1], \overline{\lambda} = 1 - \lambda$), take the functional form:
		\begin{equation}
			f(\lambda) = \left( \frac{\lambda N + \overline{\lambda} N'}{\lambda D + \overline{\lambda} D'} \right)^{1/p}, \quad \text{where } N, N' \ge 0 \text{ and } D, D' > 0.
		\end{equation}
		By elementary algebra, this function $f(\lambda)$ always attains its maximum value at the boundaries of the domain, i.e., at $\lambda=0$ or $\lambda=1$.
		
		This implies that the maximum value for the QSL must be attained at a vector $w$ that cannot be decomposed as a weighted sum of two other linearly independent, non-negative, descending-order vectors. These are precisely the extremal vectors of the set, which are the vectors $\mathbb{I}_j$ for $j = 1, 2, \dots, n^2$.
	\end{proof}

	Consequently, to find the optimal bound for a given $p$ and $\mathcal{B}$, I do not need to perform a complex search over the continuous space of $w$. I only need to evaluate my bounds for these $n^2$ different extremal vectors $\mathbb{I}_j$. This result is crucial, as it proves that the $w$-optimization step is computationally efficient, allowing the full optimization (including $p$ and $\mathcal{B}$) to remain practical~\cite{Boumal14,Chau2025}.

	\section{Applications of My QSLs}
	\label{sec:seminorm_applications}
	
	To demonstrate the wide applicability and sharpness of my new bounds, I now present four case studies across a range of quantum systems, from simple closed dynamics to complex, driven open systems. The numerical results for these examples were generated using the Manopt package in Matlab, which is an optimization toolkit~\cite{Boumal14}.  Each data point typically requires less than 10 minutes of computation on a standard laptop.
	
	\subsection{Time-Independent Hamiltonian Evolution}
	\label{subsec:seminorm_closed_system}
	
	I begin my analysis with the foundational case of pure state evolution under a time-independent Hamiltonian.
	
	\subsubsection{The Qubit Case}
	\label{subsubsec:seminorm_qubit_case}
	
	As a first warm-up, I consider the evolution of the initial state $|\psi_0\rangle = (|E_0\rangle + |E_1\rangle)/\sqrt{2}$ under the simple Hamiltonian $H = \sum_j j |E_j\rangle\langle E_j|$, where $\{|E_j\rangle\}$ are the energy eigenstates. For this system, the energy standard deviation is $\Delta E = 1/2$, and the MT bound $\tau_{MT}$ is known to be tight for evolution times $\tau \in [0, \pi\hbar]$.
	
	I find that my bound can also be saturated. If I fix $\tau \in [0, \pi\hbar]$ and choose the basis $\mathcal{B}_d$ that diagonalizes the difference vector $\Delta(\tau)$, it is straightforward to show that $\mathcal{D}_{1,\mathbb{I}_1,\mathcal{B}_d}^{max}(\rho_\tau, \rho_0) = |\sin(\tau/2\hbar)|$. The corresponding instantaneous speed is $||vec_{\mathcal{B}_d}(\dot{\rho}(t))||_{1,\overline{\mathbb{I}}_1} = |\cos[(t-\tau/2)/\hbar]| / 2\hbar$. Integrating this speed from 0 to $\tau$ yields exactly $|\sin(\tau/2\hbar)|$.
	Therefore, the numerator and denominator of my integral bound match, and I find:
	\begin{equation}
		\tau_{1,\mathbb{I}_1,\mathcal{B}_d}^{int} = \tau, \quad \text{for } \tau \in [0, \pi\hbar],
	\end{equation}
	My basic integral form QSL is also tight for this system~\cite{Chau2025}.
	
	I also note that for a randomly chosen basis $\mathcal{B}$, there is at least a $1/4$ probability that $\tau_{1,\mathbb{I}_1,\mathcal{B}}^{int}$ remains tight, as the relevant trigonometric factors cancel in both the numerator and denominator.
	
	More revealing is the case when $\tau > \pi\hbar$. For an evolution time of $\tau = 4\pi\hbar/3$, I find my fully optimized bound is $\tau_{opt}^{int} = 3.20\hbar$. This is significantly tighter than the MT bound, which is $\tau_{MT} = 2\pi\hbar/3 \approx 2.09\hbar$. My bound also surpasses the optimized Chau-Zeng (CZ) bound \cite{CZbound}, which for this system is also $2\pi\hbar/3$. This optimum is achieved for $w = \mathbb{I}_1$ (and thus for any $p \ge 1$), and I find that it can be reached by numerous different representation bases, making it discoverable via simple Monte Carlo sampling.
	
	\subsubsection{The Four-Dimensional Qudit Case}
	\label{subsubsec:seminorm_qudit_case}
	
	To demonstrate the power of my QSL beyond simple systems, I analyze a 4D qudit. It is known that most QSLs are only saturated in qubit or qutrit systems~\cite{Deffner17a,Pires2016}. I consider a randomly chosen initial state $|\psi_0\rangle = \sqrt{0.2}|E_0\rangle + \sqrt{0.4}|E_1\rangle + \sqrt{0.3}|E_\pi\rangle + \sqrt{0.1}|E_4\rangle$ evolving under the same diagonal Hamiltonian $H$.
	
	At an evolution time of $\tau_c = 3.43\hbar$, I find that the optimized CZ bound is only $\tau_{CZ} = 0.87\hbar$~\cite{CZbound}. A non-optimized version of my bound (using the energy eigenbasis) gives $\tau^{int} = 1.98\hbar$, which is better but still not tight.
	
	However, my \textit{fully optimized} bound, $\tau_{opt}^{int}$, tells a different story. I find numerically that $\tau_{opt}^{int}$ is equal to $\tau$ up to at least 7 significant figures for \textit{any} $\tau \in [0, \tau_c]$. This strongly suggests my QSL is \textbf{saturated} for this system, demonstrating its power in higher-dimensional cases where other bounds, even optimized ones, are not tight~\cite{Chau2025}. As with the qubit case, I observe that the optimal parameters appear to be degenerate, which facilitates the optimization process.

	\subsection{A Qubit Undergoing Spontaneous Emission}
	\label{subsec:seminorm_spontaneous}
	
	I now turn to open system dynamics, using the canonical example of a qubit undergoing spontaneous emission. The system is initialized in the state $(|0\rangle + |1\rangle) / \sqrt{2}$ and decays to the ground state $|0\rangle$. The dynamics are governed by the Lindblad master equation in the interaction picture:
	\begin{equation}
		\frac{d\rho_{t}}{dt} = \mathbb{L}(\rho_{t}) = \gamma \left( |0\rangle\langle1|\rho_{t}|1\rangle\langle0| - \frac{1}{2}\{|1\rangle\langle1|,\rho_{t}\} \right),
	\end{equation}
	where $\gamma$ is the emission rate. The density matrix at time $t$, written in the energy eigenbasis $\mathcal{B}_E = \{|0\rangle, |1\rangle\}$, is
	\begin{equation}
		\rho_{t} = \frac{1}{2}
		\begin{pmatrix}
			2-e^{-\gamma t} & e^{-\gamma t/2} \\
			e^{-\gamma t/2} & e^{-\gamma t}
		\end{pmatrix},
	\end{equation}
	
	I compare my QSLs against the well-known Deffner-Lutz (DL) bound, $\tau_{DL} = \max\left\{ \frac{1}{\Lambda_{\tau}^{op}}, \frac{1}{\Lambda_{\tau}^{hs}} \right\} \sin^2 \Theta_{Bures}(\rho_\tau, \rho_0)$~\cite{Deffner13}. The Hilbert-Schmidt norm ($\Lambda^{hs}$) term serves as the open-system generalization of the MT bound.  For this system, the operator norm ($\Lambda^{op}$) term provides the tighter bound, which I refer to as the DL bound.
	
	I first analyze my bound with the simple parameter choice $p=1$, $w=\mathbb{I}_1$, and $\mathcal{B}=\mathcal{B}_E$. From the expression for $\rho_t$, the distance is $\mathcal{D}_{1,\mathbb{I}_1,\mathcal{B}_E}^{max}(\rho_\tau, \rho_0) = (1 - e^{-\gamma\tau}) / 2$. The instantaneous speed is $||vec_{\mathcal{B}_E}(\dot{\rho}_t)||_{1,\overline{\mathbb{I}}_1} = \gamma e^{-\gamma t} / 2$. Integrating this speed from 0 to $\tau$ gives an average speed of $\frac{1}{\tau} \int_0^\tau \dots dt = (1 - e^{-\gamma\tau}) / (2\tau)$.
	Substituting these into my integral bound (Eq. \ref{eq:thesis_qsl_int}) reveals a remarkable result:
	\begin{equation}
		\tau_{1,\mathbb{I}_1,\mathcal{B}_E}^{int} = \frac{(1 - e^{-\gamma\tau}) / 2}{(1 - e^{-\gamma\tau}) / (2\tau)} = \tau,
	\end{equation}
	My basic integral form QSL, even with a simple, non-optimized parameter choice, is \textbf{exactly tight} and saturates the bound for all $\tau \ge 0$.
	
	The corresponding supremum form bound is also easily calculated. The speed is maximized at $t=0$, so $\sup ||\cdot|| = \gamma/2$. This gives the bound $\tau_{1,\mathbb{I}_1,\mathcal{B}_E}^{sup} = (1 - e^{-\gamma\tau}) / \gamma$.
	
	In Figure \ref{F:qubit_case}, I plot these bounds for a fixed evolution time $\tau=1$ ns against varying decay rates $\gamma$.
	\begin{itemize}
		\item Panel (a) clearly shows $\tau_{1,\mathbb{I}_1,\mathcal{B}_E}^{int}$ (blue line) saturating the actual evolution time (yellow line). It also shows that the simpler $\tau_{1,\mathbb{I}_1,\mathcal{B}_E}^{sup}$ (solid black) provides a tighter bound than both the MT (dashed red) and DL (dashed black) bounds for $\gamma < 2.5 \text{ ns}^{-1}$.
		\item Panel (b) confirms this performance for the basis-independent choice $p=2, w=\mathbb{I}_4$.
		\item Panels (c) and (d) illustrate the effects of varying $p$ and $w$, respectively, demonstrating that my bounds are effective over a wide range of parameters.
	\end{itemize}
	A significant advantage demonstrated here is that even without full optimization, a sub-optimal parameter choice yields a QSL that is extremely close to saturation, affirming the practical utility and robustness of my framework~\cite{Chau2025}.

	\begin{figure*}[t!]
		\centering
		\subfloat[\label{fig:qubit_a}]{
			\includegraphics[width=0.47\textwidth]{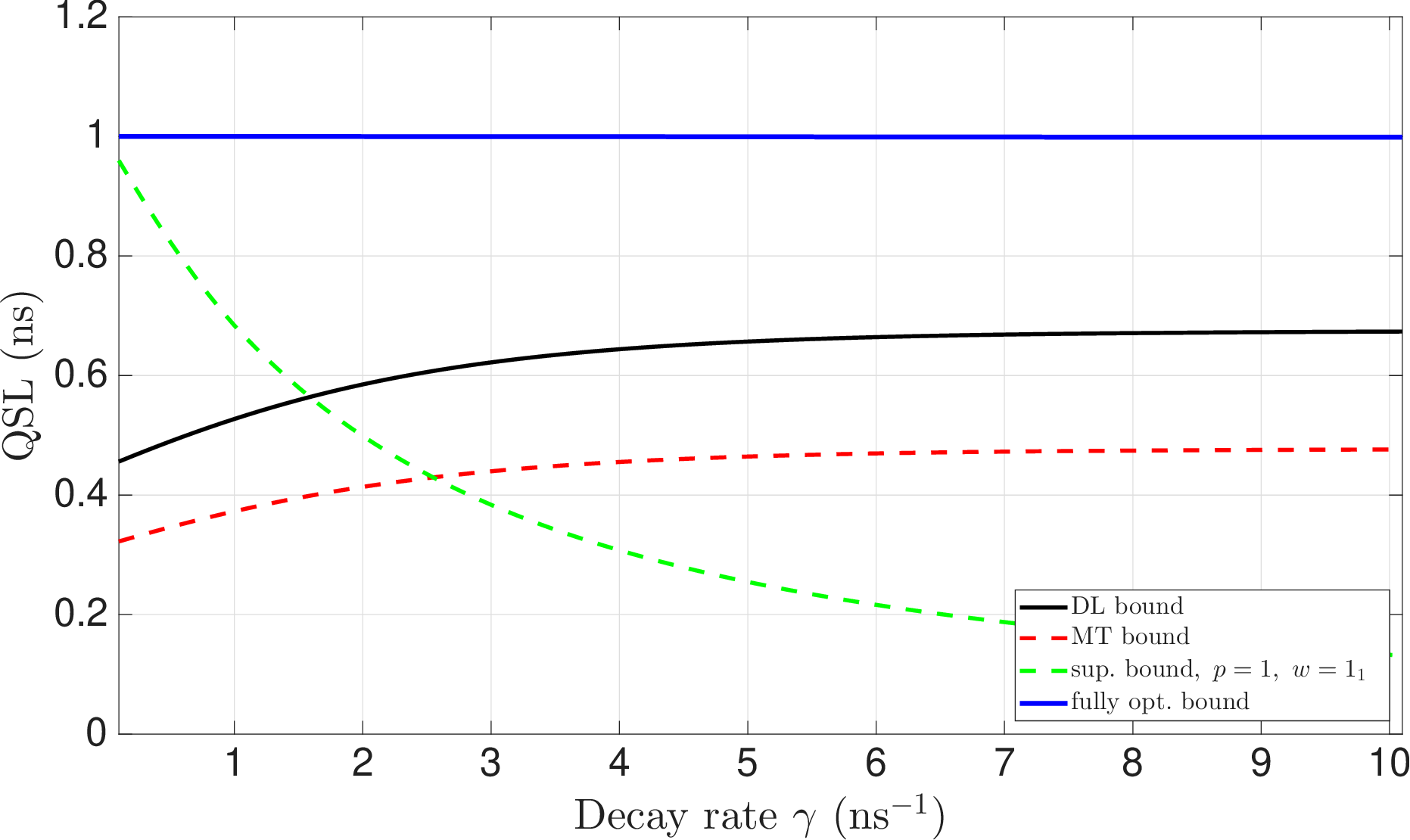}
		}
		\hfill 
		\subfloat[\label{fig:qubit_b}]{
			\includegraphics[width=0.47\textwidth]{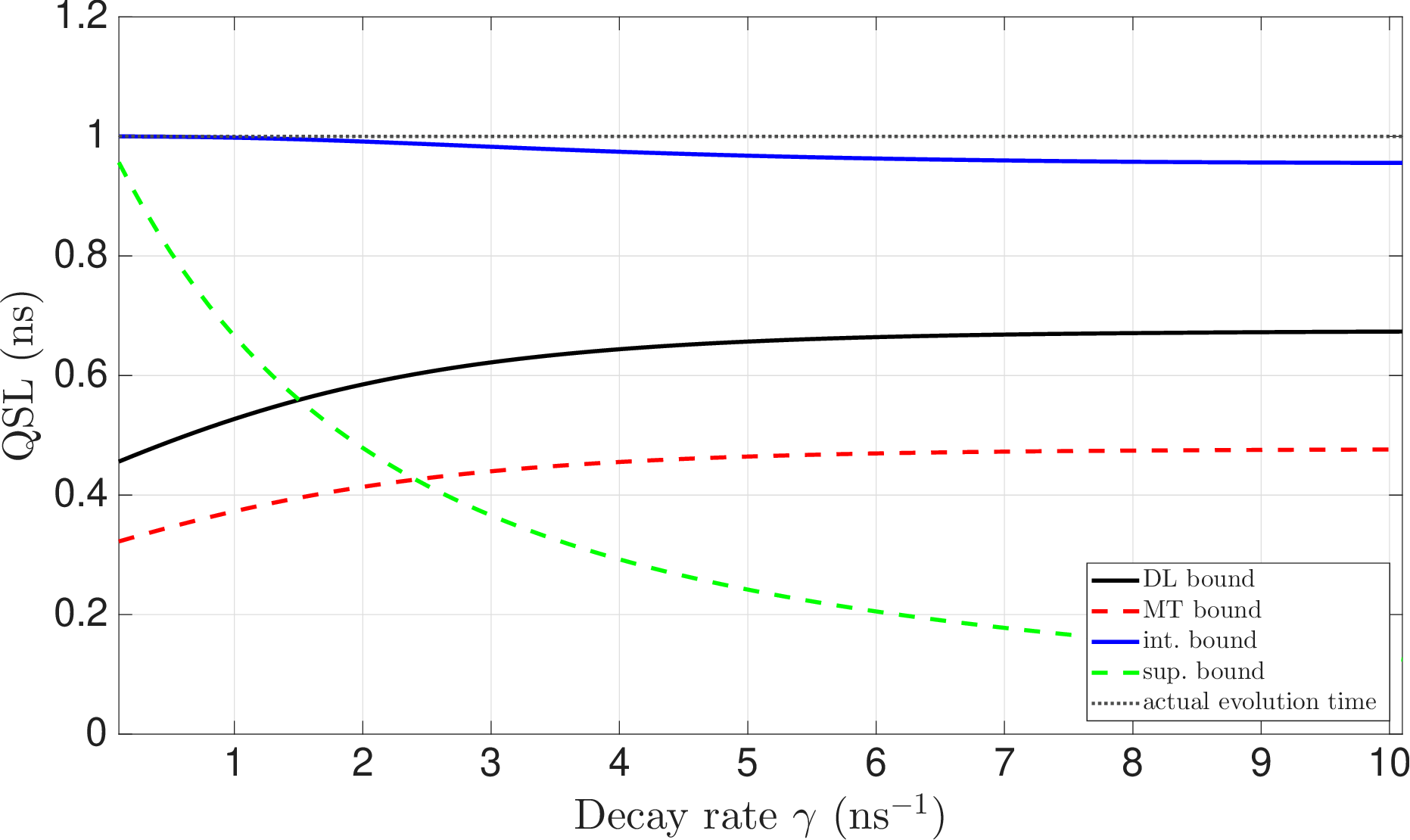}
		}
		
		\vspace{1em} 

		\subfloat[\label{fig:qubit_c}]{
			\includegraphics[width=0.47\textwidth]{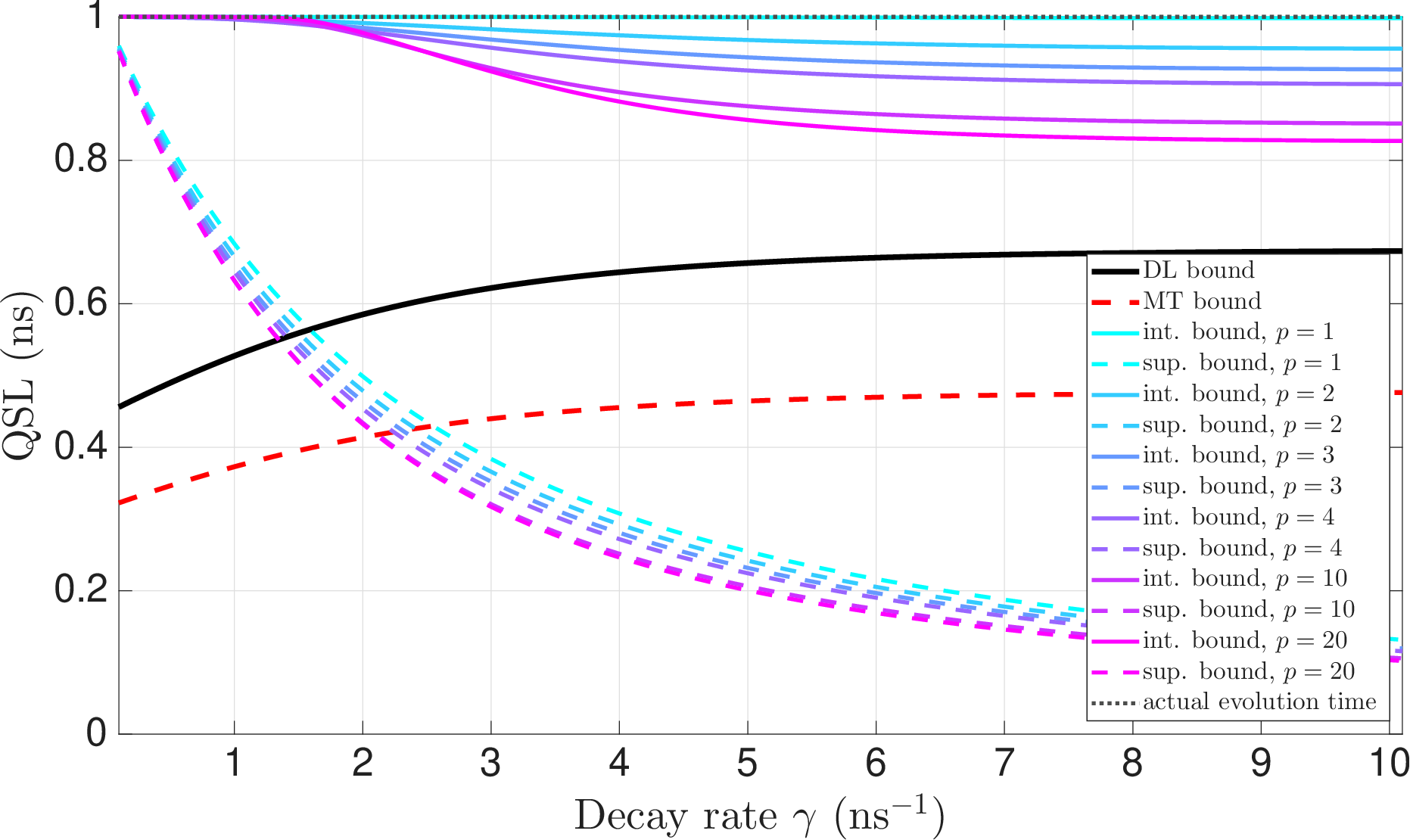}
		}
		\hfill 
		\subfloat[\label{fig:qubit_d}]{
			\includegraphics[width=0.47\textwidth]{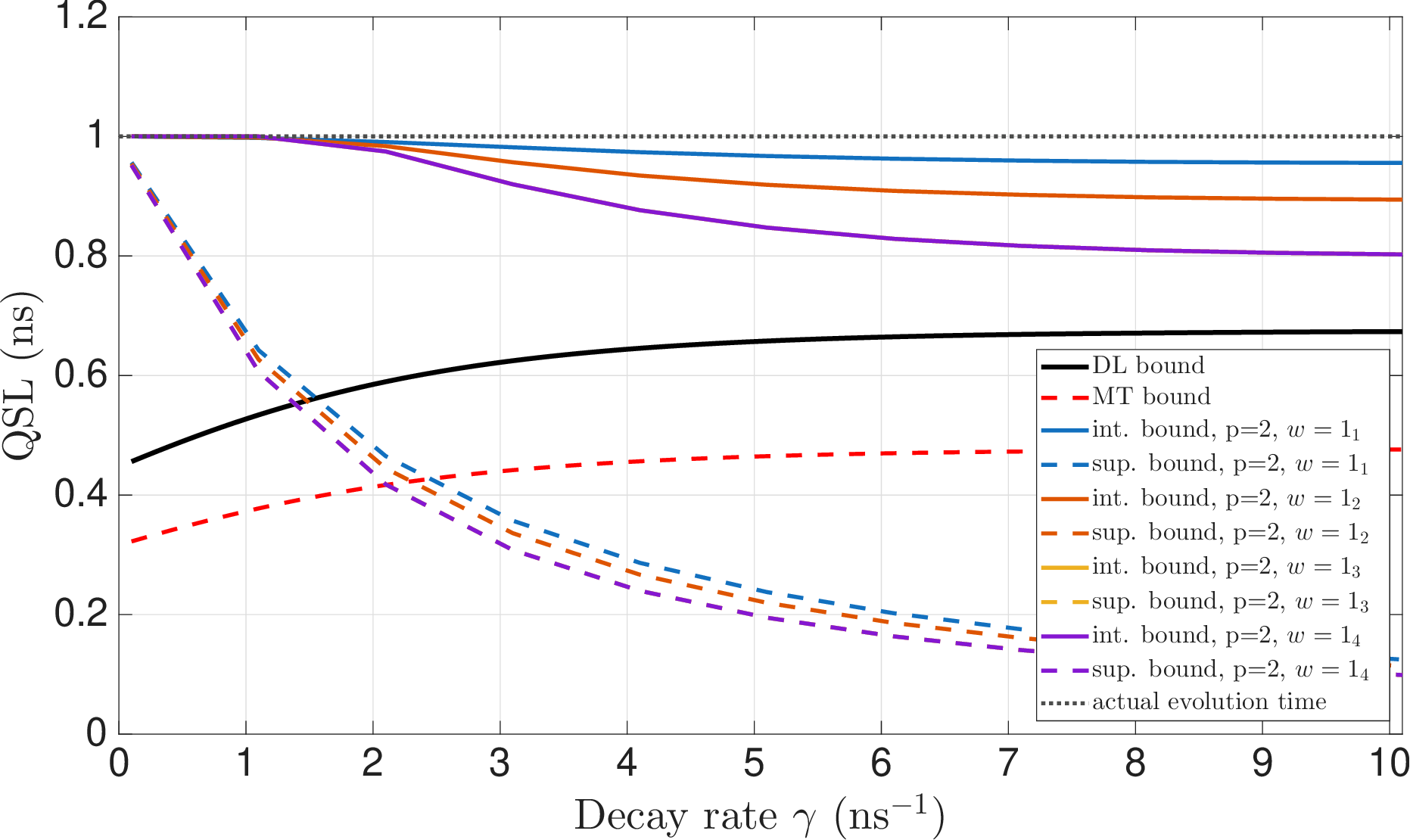}
		}
		
		\caption{
		Comparison of different \QSL bounds for a spontaneously emitting qubit plotted as a function of the decay rate $\gamma$ (with total time fixed at $\tau=1$).
		Panel~(a) displays the supremum ($\tau_{1,\mathbb{1}_1,\mathcal{B}_E}^{\text{sup}}$) and integral ($\tau_{1,\mathbb{1}_1,\mathcal{B}_E}^{\text{int}}$) bounds alongside the standard \MT and \DL limits. A key observation here is that the integral bound $\tau_{1,\mathbb{1}_1,\mathcal{B}_E}^{\text{int}}$ exactly saturates the actual evolution time $\tau$.
		Panel~(b) presents a similar comparison but utilizes the parameter set $p=2$ and $w=\mathbb{1}_4$.
		Finally, Panels~(c) and (d) demonstrate the sensitivity of the bounds by varying the norm parameter $p$ and the weight vector $w$, respectively, while holding all other parameters constant at the values used in Panel~(b).
			\label{F:qubit_case}
		}
	\end{figure*}

	\subsection{High-fidelity Gates Using an NV-center in Diamond}
	\label{subsec:seminorm_nv_center}
	
	Next, I test my framework on a more complex and experimentally relevant scenario: the dynamics of a spin-1 nitrogen-vacancy (NV) center in diamond. This system is a promising platform for quantum information processing, and I analyze the time required to perform a high-fidelity quantum gate, comparing my QSLs with the MT bound from~\cite{Deffner13}.
	
	The spin-1 operators in the basis $\{|m_s=+1\rangle, |m_s=0\rangle, |m_s=-1\rangle\}$ are given by:
	\begin{align}
		\hat{S}_x &= \frac{1}{\sqrt{2}}\begin{pmatrix} 0 & 1 & 0 \\ 1 & 0 & 1 \\ 0 & 1 & 0 \end{pmatrix}, \quad
		\hat{S}_y = \frac{1}{\sqrt{2}}\begin{pmatrix} 0 & -i & 0 \\ i & 0 & -i \\ 0 & i & 0 \end{pmatrix}, \quad
		\hat{S}_z = \begin{pmatrix} 1 & 0 & 0 \\ 0 & 0 & 0 \\ 0 & 0 & -1 \end{pmatrix},
	\end{align}
	The dynamics are governed by a time-dependent Hamiltonian $\hat{H}(t) = \hat{H}_0 + \hat{H}_c(t)$. The static Hamiltonian is:
	\begin{equation}
		\hat{H}_0 = \hbar(D\hat{S}_z^2 + \gamma_e B_0 \hat{S}_z),
	\end{equation}
	where $D$ is the zero-field splitting and $\gamma_e$ is the electron gyromagnetic ratio. The control Hamiltonian, representing a time-dependent magnetic field $B_1$ applied perpendicular to the NV axis, is:
	\begin{equation}
		\hat{H}_c(t) = \hbar\gamma_e B_1 [f_x(t)\hat{S}_x + f_y(t)\hat{S}_y].
	\end{equation}
	The control field switches direction at the halfway point:
	\begin{equation}
		(f_x(t), f_y(t)) = \begin{cases} (1, 0) & \text{if } t < \tau/2 \\ (0, 1) & \text{if } t \ge \tau/2 \end{cases}.
	\end{equation}
	I initialize the system in the state $|\psi_0\rangle = |m_s = +1\rangle$ and use the parameters $D=2\pi \times 2.87$ GHz, $\gamma_e = 2\pi \times 28.0345$ GHz/T, and $B_0 = 0.05$ T.
	
	The results are shown in Figure \ref{fig:nv_center}.
	\begin{itemize}
		\item Panel (a) shows my fully optimized bound, $\tau_{opt}^{int}$, compared to the MT bound. My bound is demonstrably tighter across the entire range of magnetic field ratios $B_0/B_1$, ``often by a wide margin." The optimal parameters $(p, w, \mathcal{B})$ themselves vary with the field ratio in a complex manner.
		\item Panels (b), (c), and (d) show the effect of varying each parameter ($p$, $w$, $\mathcal{B}$) individually, while fixing the others to a baseline (e.g., $p=2, w=\mathbb{I}_9$, computational basis). These plots show that optimizing even a single parameter is effective in achieving a good bound.
		\item I observe that all bounds show significant fluctuations with the $B_0/B_1$ ratio. However, my QSLs, particularly for larger $p$, tend to exhibit a \textbf{smaller variance}. This is a key practical advantage: a QSL with low variance provides a more robust estimate of the minimal gate time, making it less susceptible to minor experimental parameter drifts~\cite{Chau2025}.
	\end{itemize}
	This analysis demonstrates that by optimizing over $p$ and $w$, my QSL framework can offer both superior sharpness and a less fluctuating, more robust bound than existing methods for complex, time-dependent systems.

	\begin{figure*}[t!]
		\centering
		\subfloat[\label{fig:nv_a}]{
			\includegraphics[width=0.48\textwidth]{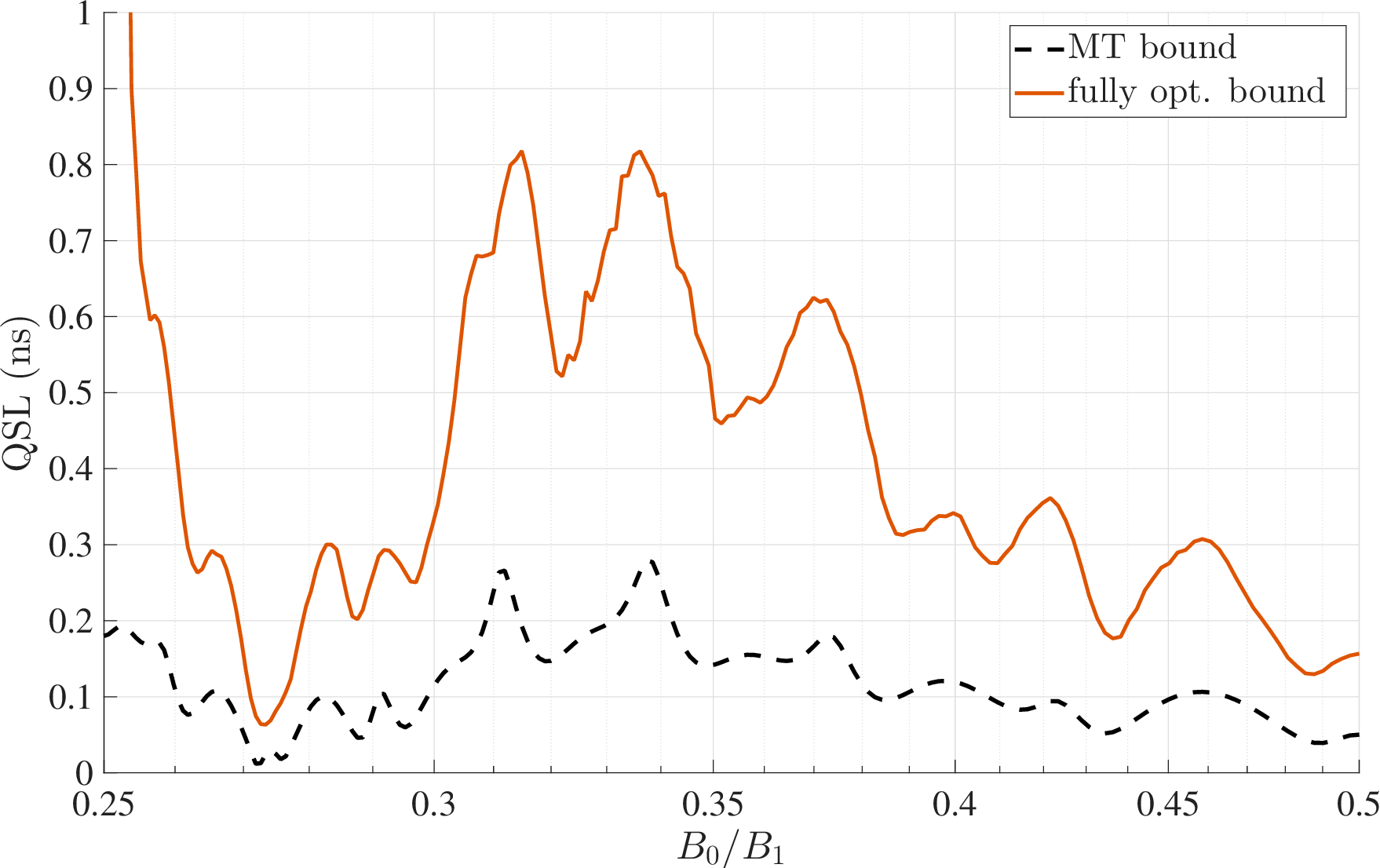}
		}
		\hfill 
		\subfloat[\label{fig:nv_b}]{
			\includegraphics[width=0.48\textwidth]{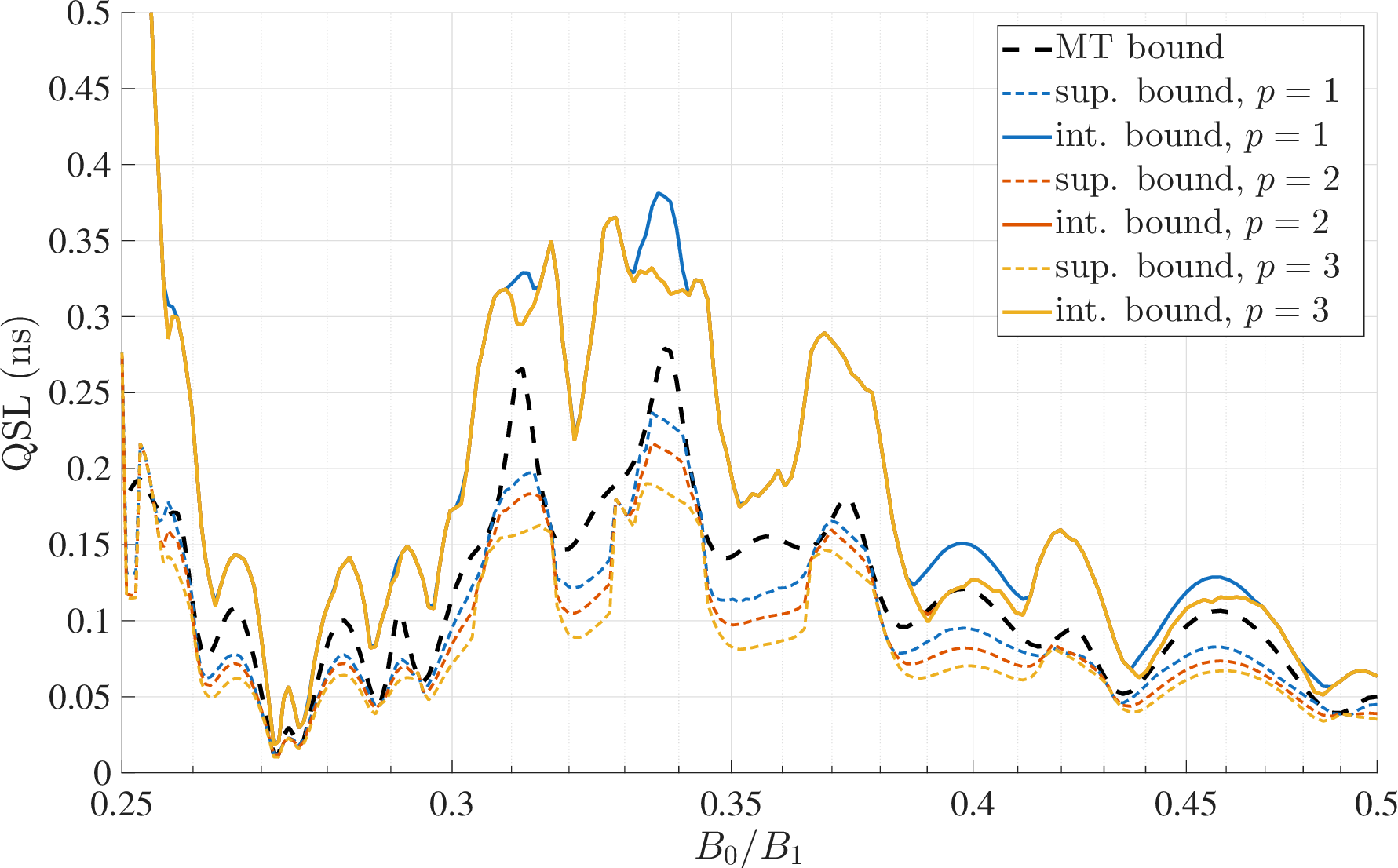}
		}
		
		\subfloat[\label{fig:nv_c}]{
			\includegraphics[width=0.48\textwidth]{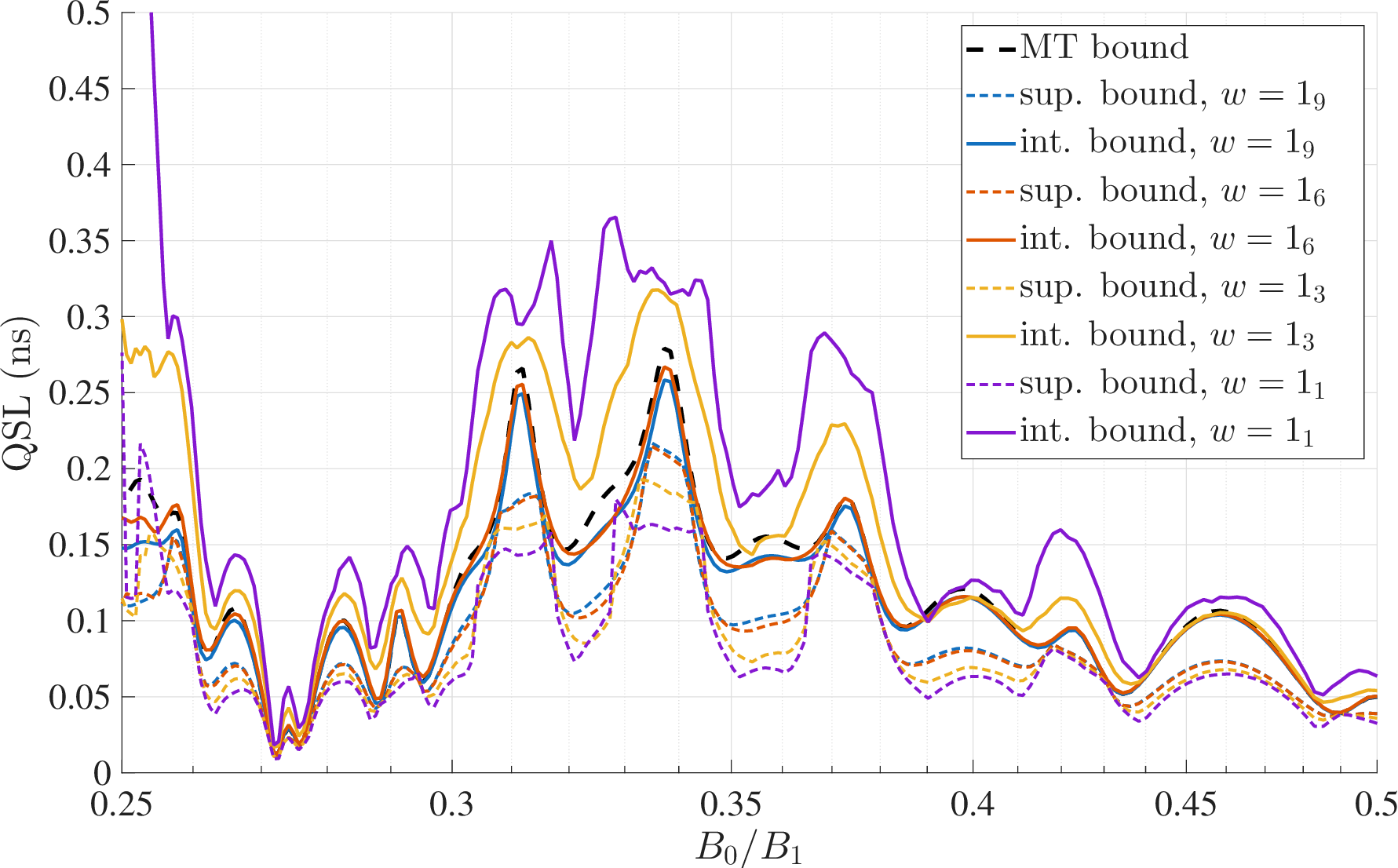}
		}
		\hfill 
		\subfloat[\label{fig:nv_d}]{
			\includegraphics[width=0.48\textwidth]{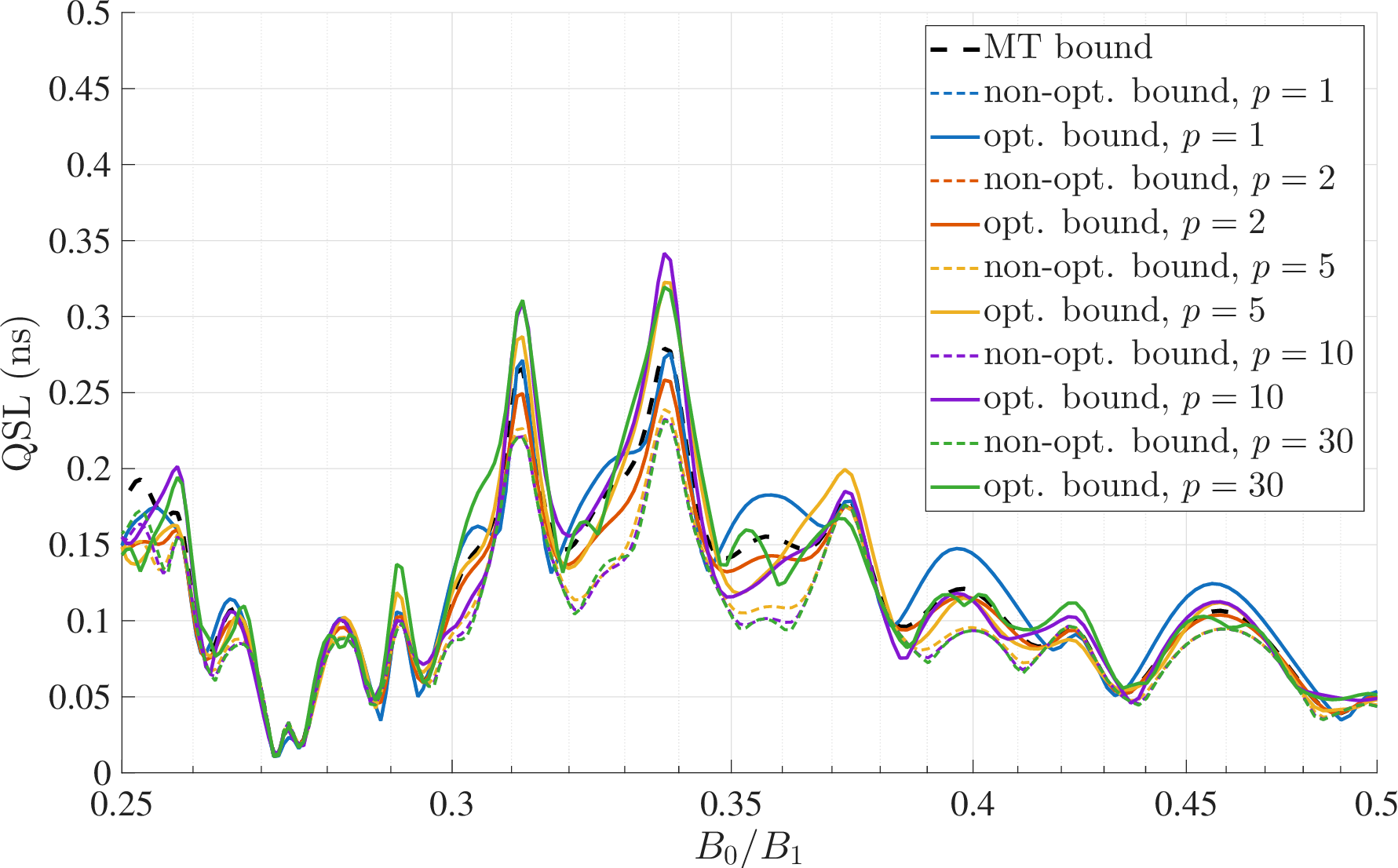}
		}
		
		\caption{
			Comparison of various \QSL bounds for the NV-center spin gate operation (defined in Sec.~III\,C) plotted against the magnetic field ratio $B_0/B_1$. I assume natural units ($\hbar = 1$) throughout.
			Panel~(a) presents the globally optimized integral bound, $\tau_{\text{opt}}^{\text{int}}$, note that this plot utilizes an expanded $y$-axis range for clarity.
			Panels~(b), (c), and (d) isolate the effects of specific parameters: I vary the norm $p$, the weight vector $w$, and the basis $\mathcal{B}$, respectively. In each of these three panels, the non-varying parameters are held at their reference values: $p=2$, uniform weights ($w = \mathbb{1}_9$), and the computational basis.
			\label{fig:nv_center}
		}
	\end{figure*}
	
	\subsection{Quantumness Measures}
	\label{subsec:seminorm_quantumness}
	
	A key feature of my framework, as established in Section \ref{subsec:seminorm_proof}, is its applicability to \textit{any} linear operator, not just density matrices. I now demonstrate this unique versatility by applying my QSL to the dynamics of operator ``quantumness," specifically non-commutativity. I adopt the models from \cite{Shrimali25}, which derived separate, specialized QSLs for dephasing and coherence generation~\cite{Shrimali25}. In contrast, my framework provides a single, unified bound applicable to both processes.
	
	\subsubsection{The Dephasing Process}
	\label{subsubsec:seminorm_dephasing}
	
	First, I consider a dephasing process where the observable of interest is initialized to $\mathcal{A}_0 = \sigma_y$ and evolves under the Hamiltonian $H = \sigma_x$. The time-evolved observable is:
	\begin{equation}
		\mathcal{A}_t = e^{-\gamma t/2} \begin{pmatrix}
			\cos 2t & -i \sin 2t \\
			i \sin 2t & -\cos 2t
		\end{pmatrix}.
	\end{equation}
	A measure of ``quantumness" between the initial and final observables is given by $Q(\mathcal{A}_0, \mathcal{A}_t) \equiv 2 ||[\mathcal{A}_0, \mathcal{A}_t]||^2 = 16 e^{-\gamma t} \sin^2 2t$~\cite{Shrimali25}.
	A specialized QSL, $T_Q$, was derived for this specific process in \cite{Shrimali25}:
	\begin{equation}
		T \ge T_Q \equiv \frac{\sqrt{Q(\mathcal{A}_0, \mathcal{A}_\tau)}}{\sqrt{2 \langle ||[\mathcal{A}_0, \mathcal{L}^\dagger(\mathcal{A}_t)]||_{hs} \rangle_\tau}},
	\end{equation}
	where $\langle \cdot \rangle_\tau$ denotes the time average over $[0, \tau]$~\cite{Shrimali25}.
	
	I apply my QSL to this operator evolution. I find that the optimal supremum bound $\tau_{opt}^{sup}$ is attained for $p=2$ and $w=\mathbb{I}_4$, which, as shown in Section \ref{subsubsec:seminorm_basis_independent}, means the bound is representation-basis-independent.
	
	As shown in Figure \ref{fig:qsl_dephasing}, my bound (solid orange) is tight for short evolution times ($t < 0.35$). While the specialized bound $T_Q$ (dashed blue) is slightly tighter for intermediate times, my QSL performs better for longer evolution times ($t > 1.0$) and, importantly, does not decrease at large $t$~\cite{Shrimali25}.

	\subsubsection{The Coherence Generation Process}
	\label{subsubsec:seminorm_coherence_gen}
	
	Second, I analyze a coherence generation process. The system starts in the state $\rho_0 = |0\rangle\langle 0|$ and evolves under $H = \sigma_x$ with dissipation:
	\begin{equation}
		\frac{d\rho_t}{dt} = -i[H, \rho_t] + \frac{\gamma}{2} (\sigma_z \rho_t \sigma_z - \rho_t).
	\end{equation}
	At time $t$, the state is:
	\begin{equation}
		\rho_t = \frac{1}{2} e^{-\gamma t/2} \begin{pmatrix}
			e^{\gamma t/2} + \cos 2t & -i \sin 2t \\
			i \sin 2t & e^{\gamma t/2} - \cos 2t
		\end{pmatrix}.
	\end{equation}
	The coherence with respect to $\mathcal{A} = \sigma_z$ is $C(\rho_t, \mathcal{A}) = e^{-\gamma t/2} |\sin 2t|$~\cite{Shrimali25}. This process also has its own specialized QSL from \cite{Shrimali25}:
	\begin{equation}
		T \ge T_C \equiv \frac{\sqrt{2} |\sqrt{C(\rho_0, \mathcal{A})} - \sqrt{C(\rho_\tau, \mathcal{A})}|}{\langle \sqrt{\sum_k ||\partial_t \sqrt{\rho_t} |k\rangle\langle k||^2} \rangle_\tau}.
	\end{equation}
	
	As before, I find my optimal supremum bound is the basis-independent case $p=2, w=\mathbb{I}_4$. Figure \ref{fig:qsl_coherence_generation} shows that my bound (solid orange) is again tight for short times ($t < 0.35$) and outperforms the specialized bound $T_C$ (dashed blue) for longer times ($t > 1.1$)~\cite{Shrimali25}.
	
	These two examples highlight the robustness and generality of my framework. It can be applied successfully to operator dynamics, not just state evolution, and provides a single, unified bound that remains a sharp and valid constraint across diverse quantum processes~\cite{Shrimali25,Chau2025}.
	
\begin{figure}[H]
	\centering
	\includegraphics[height=6cm]{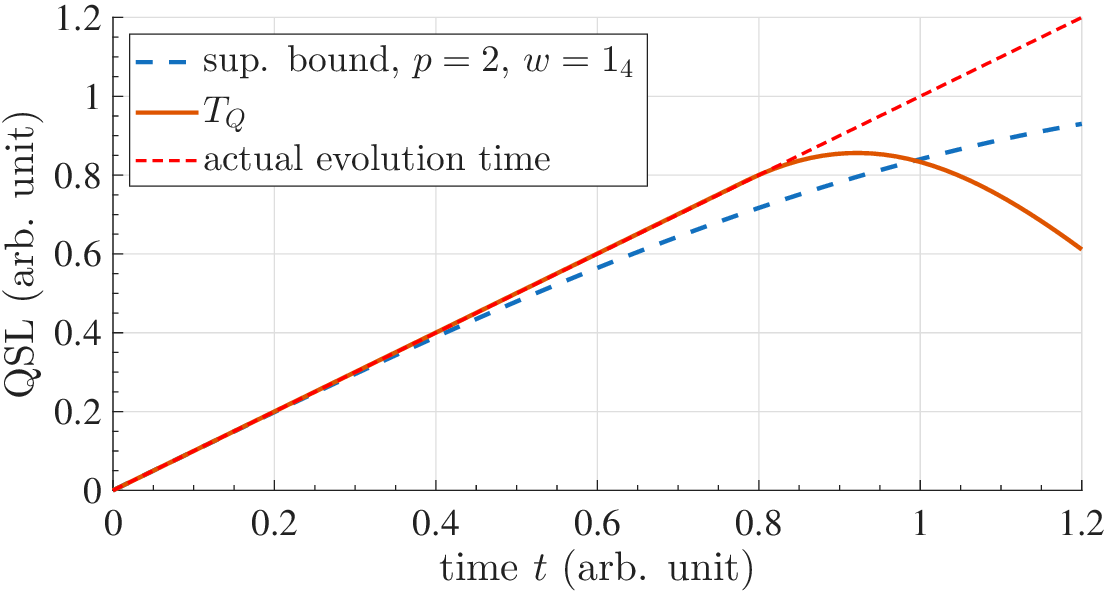}
	\caption{
		Optimized \QSL for a qubit undergoing dephasing. The bound is maximized using $p=2$ and a uniform weight vector ($w=\mathbb{1}_4$). I observe that the bound becomes saturated for $t<0.35$ and yields a tighter limit than Shrimali's for $t>1.0$.
	}
	\label{fig:qsl_dephasing}
\end{figure}

\begin{figure}[H]
	\centering
	\includegraphics[height=6cm]{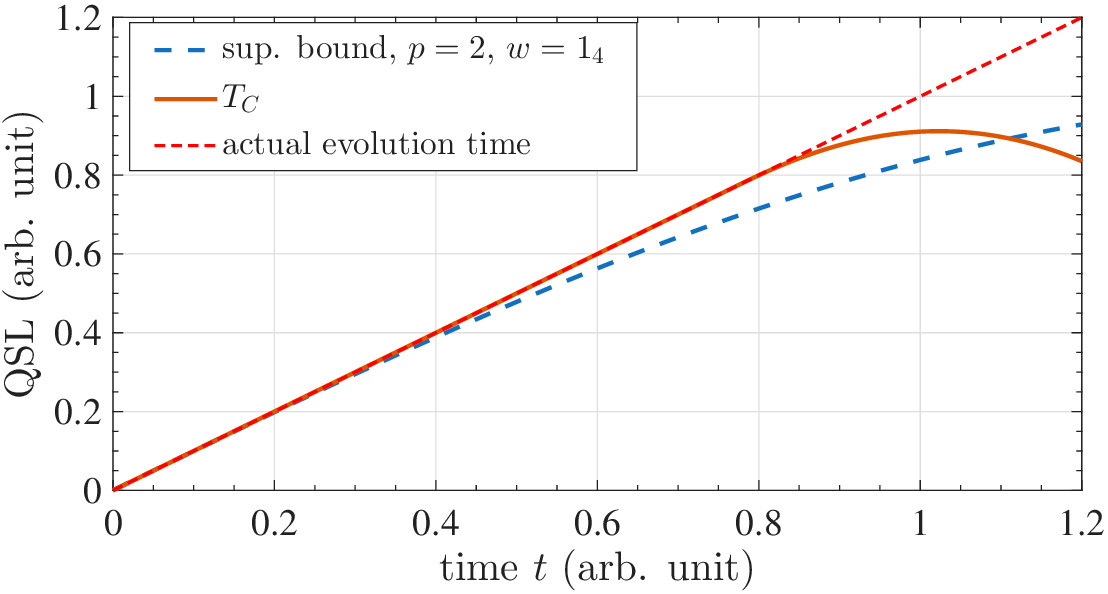}
	\caption{
		The optimized \QSL applied to qubit coherence generation. The optimum configuration corresponds to $p=2$ and uniform weights ($w=\mathbb{1}_4$). The bound saturates in the regime $t<0.35$ and improves upon Shrimali's result for $t>1.1$.
	}
	\label{fig:qsl_coherence_generation}
\end{figure}

	\section{Discussions and Outlook}
	\label{sec:seminorm_conclusion}

	In summary, this chapter has introduced a family of $\ell_{w}^{p,max}$-norm-based QSLs that are both easy to use and powerfully general~\cite{Chau2025}. I have demonstrated their wide applicability, encompassing closed and open quantum systems, time-dependent and independent Hamiltonians, and even extending beyond density matrices to general quantum operators~\cite{Deffner17a,Shrimali25,Herb24}.
	
	The four case studies confirm that my QSLs are sharp enough to provide the tightest known bounds in many scenarios, particularly for relatively short evolution times. In several cases, such as the 4D qudit and spontaneous emission, my optimized integral bound $\tau_{opt}^{int}$ was shown to be exactly tight. Furthermore, the computationally less intensive supremum form $\tau^{sup}$ often proved to be a strong and effective bound.
	
	The key to my framework's success is the novel use of a representation-basis-dependent norm. This feature, which is unique among the QSLs I are aware of, provides the flexibility to tune the bound to the specific dynamics, which is crucial for achieving tightness.
	
	This work also opens several avenues for future research. My numerical examples show that the optimal parameters $(p, w, \mathcal{B})$ do not follow any obvious, simple pattern. A deeper investigation into the structure of this optimization space could yield further insights. Additionally, the observation from my closed-system example,that optimal bases are common and can be found via simple sampling, especially for short times, is a curious phenomenon. While I provided an analytical reason for the simple qubit case, a general explanation for this ``degeneracy" is still lacking and warrants further study.
	
	These findings provide a powerful new tool for benchmarking and understanding the fundamental limits of quantum technologies.

	\chapter{DISCUSSIONS, CONCLUSION, AND OUTLOOK}
	\label{ch:conclusion}
	
	\section{Overview of the Thesis}
	
	This thesis has explored two distinct yet complementary directions in quantum information science. The first part concerns the \emph{fundamental dynamics} of quantum systems, and the second part focuses on the \emph{practical realization} of secure quantum communication protocols.
	
	On the foundational side, I developed a family of quantum speed limits (QSLs) based on representation-dependent weighted $\ell_{p,w}$-seminorms~\cite{Chau2025}. This framework provides a unified and computationally efficient way to characterize the minimal time required for general quantum evolutions, including those governed by Lindbladian dynamics~\cite{Deffner17a}. It reveals that representation dependence—often seen as a limitation—can, in fact, be leveraged to achieve task-specific sharpness and improved interpretability.
	
	On the applied side, I proposed and analyzed fully passive quantum communication protocols, including both measurement-device-independent quantum key distribution (MDI-QKD) and multi-user conference key agreement (CKA)~\cite{Li2024,Li2024CKA}. For the MDI-QKD scheme, a fully passive source architecture was introduced, relying solely on intrinsic optical interference rather than active modulation~\cite{Li2024,Curty2010,Wang2023a,Hu2023,Lu2023}. A complete decoy-state and linear-program analysis was developed for mixed single-photon ensembles, accompanied by a rigorous mapping from mixed to ideal states for secure key-rate evaluation and numerical verification of long-distance performance~\cite{Lo2005,Wang2005,Ma2005,Curty2014}. For the multi-user CKA case, a twin-field–based passive design was constructed, incorporating realistic loss modeling and multi-party decoy analysis~\cite{Li2024CKA,Grasselli2019CKA,Curty2019}. The work further introduced a branch-cutting algorithm that significantly reduced computational effort in parameter optimization, enabling tractable simulation of high-dimensional systems.
	
	Together, these two studies address different layers of the quantum information hierarchy: one examines the physical limits of evolution itself, and the other demonstrates how careful engineering can realize secure quantum communication with minimal device trust.
	
	\section{Discussion and Outlook}
	
	The developments presented in this thesis contribute to a broader understanding of how fundamental principles and practical implementation coexist within quantum information science.
	
	From the theoretical perspective, the $\ell_{p,w}$-based QSL framework introduced here offers a family of easy-to-use, nonunitary-invariant bounds that provide the combined benefits of universality and computational simplicity. As demonstrated across diverse illustrative systems, this framework is powerful enough to provide the tightest known bounds to date in many scenarios, particularly when the evolution time $\tau$ is relatively short. The primary driver of this performance is the strategic use of the representation-basis-dependent norm $||vec_{\mathcal{B}}(\mathcal{L}\rho_t)||_{p,\overline{w}}$ and its time integral, which represents a unique feature compared to the unitarily invariant QSLs currently predominant in the literature. 
	
	Crucially, the framework’s robustness allows for generalization to operator dynamics beyond density matrices, a feature critical for addressing a wide spectrum of quantum phenomena including dephasing and coherence generation. This suggests that the $\ell_{p,w}$ framework is highly suitable for benchmarking performance in quantum control, metrology, and the thermodynamics of open systems.
	
	Future work could be 
		 \textbf{Analytical Optimization}: Investigate the theoretical foundations of $(p,w,B)$ optimization. While numerical results show optimized bases are common for small $\tau$, a sound analytical explanation for this degeneracy is missing. Future research should clarify parameter patterns and develop methods to find the optimal basis $B$ for fixed $p$ and $w$.

	From the applied perspective, this thesis demonstrates that fully passive architectures provide highly secure and practical alternatives to actively modulated quantum communication systems. By simultaneously closing side channels at both the source modulator and detector sides, these protocols offer better implementation security. 
	
	Future applied work should pursue several key directions:
	
	\begin{itemize}
		\item \textbf{Passive MDI-QKD Advancements}: While highly secure, the double-sifting nature of the passive sources currently reduces the key rate by 2--3 orders of magnitude. Future work must tackle this sifting challenge, extend the security proof with a rigorous finite-size analysis, and achieve experimental realization.
		
		\item \textbf{Multi-User Passive CKA Optimization}: For the fully passive CKA protocol, further progress requires leveraging enhanced computational resources to optimize slice numbers and input intensities without relying on branch-cutting approximations. Additionally, implementing higher-order (e.g., three-decoy) state analyses will provide tighter yield bounds and improve performance.
		\item \textbf{Network Scaling and Deployment}: In the broader context, it is critical to analyze the scaling laws between sifting efficiency and an increasing number of network users.
	\end{itemize}
	
	In summary, this thesis contributes to two essential aspects of the field: the fundamental understanding of quantum dynamical limits and the development of secure, implementable communication schemes. While these topics differ in scope, they share a common aim—to advance quantum information science through a balance of theoretical clarity and practical feasibility. It is hoped that the ideas developed here will continue to inspire future studies that connect the abstract and the applied, strengthening the conceptual and technological foundation of emerging quantum networks.

	\appendix
\renewcommand{\thechapter}{\Roman{chapter}}

	\chapter{DECOY-STATE-ANALYSIS -- LINEAR PROGRAMMING CONSTRUCTION}
	\label{sec:mdi_lp}
	
	In passive MDI-QKD, the gain and error–gain satisfy
\begin{align}
	\langle Q\rangle_{S_i^A S_j^B}
	&=\sum_{n,m=0}^{\infty}
	\left\langle P_n^A P_m^B 
	Y_{nm}\right\rangle_{S_i^A S_j^B},
	\label{eq:Q_avg}
	\\
	\langle Q E\rangle_{S_i^A S_j^B}
	&=\sum_{n,m=0}^{\infty}
	\left\langle P_n^A P_m^B e_{nm} Y_{nm}\right\rangle_{S_i^A S_j^B}.
	\label{eq:QE_avg}
\end{align}
	where $S_i^A$ and $S_j^B$ denote Alice’s and Bob’s post-selection regions, $P_n^{A/B}$ are Poisson probabilities, depending on $(\mu_H^{A/B}, \mu_V^{A/B})$, and $Y_{nm}$, $e_{nm}Y_{nm}$ are functions of the observables $(\mu_H^{A}, \mu_V^{A}, \phi_{HV}^{A}, \mu_H^{B}, \mu_V^{B}, \phi_{HV}^{B})$. Writing $(\mu_H^A,\mu_V^A)$ in polar form $(r^A,\theta^A)$ (and similarly for Bob), and setting $\phi_{HV}^A \equiv \phi^A$, I have

\begin{align}
	\left\langle P_n^A P_m^B Y_{nm}\right\rangle_{S_i^A S_j^B} 
	&= \frac{1}{P_{S_i^A S_j^B}^{\mu^A, \mu^B, \phi^A, \phi^B}} \notag \\
	&\quad \times \iiint \!\!\! \iiint_{S_i^A S_j^B} p^A_\mu(r^A, \theta^A) p^B_\mu(r^B, \theta^B) p^A_\phi(\phi^A) p^B_\phi(\phi^B) \notag \\
	&\quad \times P_n^A(r^A, \theta^A) P_m^B(r^B, \theta^B) Y_{nm}(\theta^A, \phi^A, \theta^B, \phi^B) \notag \\
	&\quad \times r^A r^B \, dr^A dr^B d\theta^A d\theta^B d\phi^A d\phi^B,
	\label{eq:I.2}
\end{align}

	and
\begin{align}
	\left\langle P_n^A P_m^B e_{nm} Y_{nm}\right\rangle_{S_i^A S_j^B} 
	&= \frac{1}{P_{S_i^A S_j^B}^{\mu^A, \mu^B, \phi^A, \phi^B}} \notag \\
	&\quad \times \iiint \!\!\! \iiint_{S_i^A S_j^B} p^A_\mu(r^A, \theta^A) p^B_\mu(r^B, \theta^B) p^A_\phi(\phi^A) p^B_\phi(\phi^B) \notag \\
	&\quad \times P_n^A(r^A, \theta^A) P_m^B(r^B, \theta^B) e_{nm} Y_{nm}(\theta^A, \phi^A, \theta^B, \phi^B) \notag \\
	&\quad \times r^A r^B \, dr^A dr^B d\theta^A d\theta^B d\phi^A d\phi^B.
	\label{eq:I.3}
\end{align}
	
	Here $Y_{nm}$ depends on the polarization variables $\theta$ and $\phi$. Since $Y_{nm}$ and $e_{nm}Y_{nm}$ share the same structure, I concentrate on $Y_{nm}$, the error–yield follows analogously. The steps below follow \cite{Wang2023a}.
	
	First integrate out $\phi^A,\phi^B$:
\begin{align}
	\left\langle P_n^A P_m^B Y_{nm}\right\rangle_{S_i^A S_j^B} 
	&= \frac{1}{P_{S_i^A S_j^B}} \iint \!\! \iint_{S_i^A S_j^B} p^A_\mu(r^A, \theta^A) p^B_\mu(r^B, \theta^B) \notag \\
	&\quad \times P_n^A(r^A, \theta^A) P_m^B(r^B, \theta^B) Y_{nm}(\theta^A, \theta^B) \notag \\
	&\quad \times r^A r^B \, dr^A dr^B d\theta^A d\theta^B,
	\label{eq:exp_value_PnPmYnm}
\end{align}
	with
\begin{align}
	Y_{nm}(\theta^A, \theta^B) = &
	\frac{1}{P_{S_i^A S_j^B}^{\phi^A, \phi^B}}
	\iint_{\phi^A, \phi^B} 
	p^A_\phi\!\left(\phi^A\right)
	p^B_\phi\!\left(\phi^B\right)
	Y_{nm}(\theta^A, \phi^A, \theta^B, \phi^B)\,
	d\phi^A d\phi^B.
	\label{eq:Ynm_theta}
\end{align}
	Next, regroup terms involving $r^{A/B}$. From the expression above,

\begin{align}
	\left\langle P_n^A P_m^B Y_{nm}\right\rangle_{S_i^A S_j^B} 
	&= \frac{1}{P_{S_i^A S_j^B}} \iiint \!\!\! \iiint_{S_i^A S_j^B} p^A_\mu(r^A, \theta^A) p^B_\mu(r^B, \theta^B) \notag \\
	&\quad \times P_n^A(r^A, \theta^A) P_m^B(r^B, \theta^B) Y_{nm}(\theta^A, \theta^B) \, r^A r^B \, dr^A dr^B d\theta^A d\theta^B \notag \\
	&= \iint_{\theta^{A/B}} \left( \frac{\int_{r^A(\theta^A)} p^A_\mu(r^A, \theta^A) P_n^A(r^A, \theta^A) \, r^A dr^A }{P_{S_i^A}} \right) \notag \\
	&\quad \times \left( \frac{\iint_{r^B(\theta^B)} p^B_\mu(r^B, \theta^B) P_m^B(r^B, \theta^B) \, r^B dr^B }{P_{S_j^B}} \right) \notag \\
	&\quad \times Y_{nm}(\theta^A, \theta^B) \, d\theta^A d\theta^B \notag \\
	&= \left\langle P_n^A \right\rangle_{S_i^A} \left\langle P_m^B \right\rangle_{S_j^B} \iint_{\theta^{A/B}} \frac{p_{\theta^A, n, S_i^A}(\theta^A)}{\left\langle P_n^A \right\rangle_{S_i^A}} \frac{p_{\theta^B, m, S_j^B}(\theta^B)}{\left\langle P_m^B \right\rangle_{S_j^B}} \notag \\
	&\quad \times Y_{nm}(\theta^A, \theta^B) \, d\theta^A d\theta^B \notag \\
	&= \left\langle P_n^A \right\rangle_{S_i^A} \left\langle P_m^B \right\rangle_{S_j^B}  Y^{\text{mixed}}_{nm, S_i S_j}.
	\label{eq:exp_value_decomposition}
\end{align}

	Here
\begin{equation}
	P_{S_i^A} = \iint_{S_i^A} p^A_\mu(r^A, \theta^A)\, r^A \, dr^A d\theta^A .
	\label{eq:P_SiA}
\end{equation}
	\begin{equation}
		P_{S_i^A} P_{S_J^B} =  P_{S_i^A S_j^B},
	\end{equation}
	and
	\begin{equation}
		p_{\theta^A, n, S_i^A}(\theta^A)
		=
		\frac{
			\int_{r^A(\theta^A)}
			p^A_\mu\!\left(r^A, \theta^A\right)
			P_n^A(r^A, \theta^A)\,
			r^A dr^A }{P_{S_i^A}}.
	\end{equation}
	Following \cite{Wang2023a}, define

\begin{align}
	Y^{\text{mixed}}_{nm, S_i S_j}
	=& \iint_{\theta^{A/B}} 
	\frac{p_{\theta^A, n, S_i^A}(\theta^A)}
	{\left\langle P_n^A \right\rangle_{S_i^A}}
	\frac{p_{\theta^B, m, S_j^B}(\theta^B)}
	{\left\langle P_m^B \right\rangle_{S_j^B}}  
	\quad
	Y_{nm}(\theta^A, \theta^B) \, d\theta^A d\theta^B .
	\label{eq:Ymixed_def}
\end{align}
	
	Thus,
	\begin{equation}
		\begin{aligned}
			\left\langle P_n^A P_m^B Y_{nm}\right\rangle_{S_i^A S_j^B}
			& =
			\left\langle P_n^A \right\rangle_{S_i^A}
			\left\langle P_n^B \right\rangle_{S_j^B}
			Y^{mixed}_{nm, S_i S_j}.
		\end{aligned}
	\end{equation}
	At this point linear programming is still inapplicable, since $Y^{mixed}_{nm,S_i S_j}$ varies with the decoy regions $S_{i/j}$, i.e., it is not a constant across settings.
	
	To remedy this, I construct decoy settings so that $Y^{mixed}_{nm}(\theta)$ no longer depends on the chosen regions—hence it is identical across decoys and standard LP bounds apply \cite{Wang2023a}. The idea is to introduce an additional post-selection that reshapes the distribution to $p^{\prime}_{\mu} = C e^{r (\sin \theta + \cos\theta)}$ for both Alice and Bob \cite{Wang2023a}, which cancels the exponential from the Poisson terms. Consequently \cite{Wang2023a},
	\begin{equation}
		Y^{mixed}_{nm} = 
		\frac{
			\iint_{\theta^{A/B}} 
			(\sin \theta^A + \cos \theta^A )^n (\sin \theta^B + \cos \theta^B )^m  Y_{nm}(\theta^A, \theta^B) \, d\theta^A d\theta^B 
		}{
			\iint_{\theta^{A/B}} 
			(\sin \theta^A + \cos \theta^A )^n (\sin \theta^B + \cos \theta^B )^m   \, d\theta^A d\theta^B 
		}.
	\end{equation}
	In parallel, the decoy regions are restricted to sector-shaped domains (see Fig.~\ref{fig:regions}). 
	
	Under this construction, the LP constraints become
	\begin{equation}
		\langle Q\rangle_{S_i^A S_j^B}
		= 
		\left\langle P_n^A \right\rangle_{S_i^A}
		\left\langle P_n^B \right\rangle_{S_j^B}
		Y^{mixed}_{nm},
	\end{equation}
	and for the error–yield,
	\begin{equation}
		\langle QE\rangle_{S_i^A S_j^B}
		= 
		\left\langle P_n^A \right\rangle_{S_i^A}
		\left\langle P_n^B \right\rangle_{S_j^B}
		e_{nm}Y^{mixed}_{nm}.
	\end{equation}
	
	\chapter{DECOY-STATE-ANALYSIS -- YIELD BOUNDS}
	\label{sec:mdi_yield}
	
	I claim that the lower bound obtained for the ‘mixed’ single-photon yield is in fact a lower bound for the ‘perfectly encoded’ yield, i.e.,
	\begin{equation}
		Y^{\text{mixed, Lower}}_{11} \le Y^{\text{mixed}}_{11} = Y^{\text{perfect}}_{11}.
	\end{equation}
	
	In passive MDI-QKD, both Alice and Bob must obtain either an $H$ or a $V$ state, I parameterize by $\theta_{A/B}$ and $\phi_{A/B}$ on the Bloch sphere \cite{Wang2023a}. The mixed state can be written as
\begin{align}
	\rho^{\text{mixed}} = &
	\left| H^{\prime} H^{\prime} \right\rangle \left\langle H^{\prime} H^{\prime} \right| +
	\left| H^{\prime}V^{\prime} \right\rangle \left\langle H^{\prime}V^{\prime} \right| \notag \\
	&+ \left| V^{\prime} H^{\prime} \right\rangle \left\langle V^{\prime} H^{\prime} \right| +
	\left| V^{\prime}V^{\prime} \right\rangle \left\langle V^{\prime}V^{\prime} \right|.
	\label{eq:mixed_rho}
\end{align}
	
	Without loss of generality set $\phi_A=0$ and let $\phi_B=\phi$, the global phase is irrelevant and I align it to Alice so $\phi_A=0$.
	
	Starting with the first term (the rest follow the same pattern),
\begin{align}
	\left| H^{\prime} H^{\prime} \right\rangle
	&= \left| H^{\prime}\right\rangle_A \otimes \left| H^{\prime} \right\rangle_B \notag \\
	&= \Bigl( \cos\tfrac{\theta_A}{2}\left| H \right\rangle_A
	+ \sin\tfrac{\theta_A}{2}\left| V \right\rangle_A \Bigr)
	\otimes
	\Bigl( \cos\tfrac{\theta_B}{2}\left| H \right\rangle_B
	+ e^{i\phi}\sin\tfrac{\theta_B}{2}\left| V \right\rangle_B \Bigr).
	\label{eq:hh_prime_expansion}
\end{align}
	Hence the corresponding projector is
\begin{align}
	\left| H^{\prime} H^{\prime} \right\rangle \left\langle H^{\prime} H^{\prime} \right| 
	&=
	\begin{pmatrix}
		\cos\!\frac{\theta_A}{2}\cos\!\frac{\theta_B}{2} & 
		e^{-i \phi} \cos\!\frac{\theta_A}{2}\sin\!\frac{\theta_B}{2} & 
		\sin\!\frac{\theta_A}{2}\cos\!\frac{\theta_B}{2} & 
		e^{-i \phi} \sin\!\frac{\theta_A}{2}\sin\!\frac{\theta_B}{2}
	\end{pmatrix} \notag \\
	&\quad \otimes
	\begin{pmatrix}
		\cos\!\frac{\theta_A}{2}\cos\!\frac{\theta_B}{2} \\ 
		e^{i \phi} \cos\!\frac{\theta_A}{2}\sin\!\frac{\theta_B}{2} \\ 
		\sin\!\frac{\theta_A}{2}\cos\!\frac{\theta_B}{2} \\ 
		e^{i \phi} \sin\!\frac{\theta_A}{2}\sin\!\frac{\theta_B}{2}
	\end{pmatrix}.
	\label{eq:HH_outer_product}
\end{align}
	
	Likewise I obtain
	\begin{equation}
		\begin{aligned}
			\left| H^{\prime} V^{\prime} \right\rangle \left\langle H^{\prime} V^{\prime} \right| 
		\end{aligned}
	\end{equation}
	\begin{equation}
		\begin{aligned}
			\left| V^{\prime} H^{\prime} \right\rangle \left\langle V^{\prime} H^{\prime} \right| 
		\end{aligned}
	\end{equation}
	\begin{equation}
		\begin{aligned}
			\left| V^{\prime} V^{\prime} \right\rangle \left\langle V^{\prime} V^{\prime} \right| .
		\end{aligned}
	\end{equation}
	Summing these projectors, one sees that all diagonal entries equal $1$ and all off-diagonal entries are $0$, i.e., the result is the $4\times4$ identity:
\begin{align}
	\rho^{\prime} & = 
	\left| H^{\prime} H^{\prime} \right\rangle \left\langle H^{\prime} H^{\prime} \right| +
	\left| H^{\prime}V^{\prime} \right\rangle \left\langle H^{\prime}V^{\prime} \right| \notag \\
	& \quad+
	\left| V^{\prime} H^{\prime} \right\rangle \left\langle V^{\prime} H^{\prime} \right| +
	\left| V^{\prime}V^{\prime} \right\rangle \left\langle V^{\prime}V^{\prime} \right| \notag \\
	& = \begin{pmatrix}
		1 &  &  &  \\
		& 1 &  &  \\
		&  & 1 &  \\
		&  &  & 1 \\
	\end{pmatrix}= I .
	\label{eq:rho_prime_identity}
\end{align}
	Here I have used $\sin^2x+\cos^2x=1$. Thus, the ‘mixed’ state is maximally mixed ($I$) and therefore equivalent, for my purpose, to the ‘perfect encoding’ mixture
	\begin{equation}
		\left| H H \right\rangle \left\langle H H \right| +
		\left| HV \right\rangle \left\langle HV \right|+
		\left| V H \right\rangle \left\langle V H \right| +
		\left| VV \right\rangle \left\langle VV \right| = I .
	\end{equation}
	Consequently, the mixed and perfect ensembles share the same lower bound. I may therefore take the lower bound on $Y^{\text{mixed}}_{nm}$, denoted $Y^{\text{mixed, Lower}}_{nm}$, as a lower bound on the perfectly encoded case:
	\begin{equation}
		Y^{\text{mixed, Lower}}_{11} \le Y^{\text{mixed}}_{11} = Y^{\text{perfect}}_{11}.
	\end{equation}

	\chapter{DECOY-STATE-ANALYSIS -- ERROR YIELD BOUNDS}
	\label{sec:mdi_error}
	
	I focus on the $HH$ case (both Alice and Bob emit an $H$ state) \cite{Wang2023a}. I aim to evaluate
\begin{align}
	\rho^{HH} & = 
	\left|H_1 H_1^{\perp}\right\rangle \left\langle H_1 H_1^{\perp}\right| +
	\left|H_1 H_2^{\perp}\right\rangle \left\langle H_1 H_2^{\perp}\right| \notag \\
	& \quad +
	\left|H_2 H_1^{\perp}\right\rangle \left\langle H_2 H_1^{\perp}\right| +
	\left|H_2 H_2^{\perp}\right\rangle \left\langle H_2 H_2^{\perp}\right|.
	\label{eq:rho_HH}
\end{align}
	Here $\left| H_1 H_1^{\perp} \right\rangle=\left| H_1 \right\rangle \otimes \left| H_1^{\perp} \right\rangle$, with $\left| H_1 \right\rangle$ a polarized $H$ state at $(\theta_A,\phi_A)$ on the Bloch sphere, and its paired state $\left| H_1^{\perp} \right\rangle$ located at $(\theta_A,\phi_A+\pi)$. Likewise, $\left| H_2 \right\rangle$ and $\left| H_2^{\perp} \right\rangle$ denote Bob’s counterparts at $(\theta_B,\phi_B)$ and $(\theta_B,\phi_B+\pi)$, respectively \cite{Wang2023a}. I shift phases so that Alice has zero phase and Bob carries a relative phase $\phi$.
	
	Each of the four single-qubit states can be expanded as
	\begin{equation}
		\left| H_1 \right\rangle =  
		\cos\!\left(\frac{\theta_A}{2}\right)\left| H \right\rangle_A 
		+ 
		\sin\!\left(\frac{\theta_A}{2}\right)\left| V \right\rangle_A,
	\end{equation}
	with analogous expressions for the remaining three states.
	
	I then evaluate the four projectors 
	$\left| H_1 H_1^{\perp} \right\rangle \!\left\langle H_1 H_1^{\perp} \right|$,
	$\left| H_1 H_2^{\perp} \right\rangle \!\left\langle H_1 H_2^{\perp} \right|$,
	$\left| H_2 H_1^{\perp} \right\rangle \!\left\langle H_2 H_1^{\perp} \right|$, and
	$\left| H_2 H_2^{\perp} \right\rangle \!\left\langle H_2 H_2^{\perp} \right|$,
	each yielding a $4\times4$ matrix. Summing them gives
	\begin{equation}
		\rho = 
		\begin{pmatrix}
			\cos^2\!\frac{\theta_A}{2}\cos^2\!\frac{\theta_B}{2} &  &  &  \\
			& \cos^2\!\frac{\theta_A}{2}\sin^2\!\frac{\theta_B}{2} &  &  \\
			&  & \sin^2\!\frac{\theta_A}{2}\cos^2\!\frac{\theta_B}{2} &  \\
			&  &  & \sin^2\!\frac{\theta_A}{2}\sin^2\!\frac{\theta_B}{2} \\
		\end{pmatrix}.
	\end{equation}
	
	Equivalently, I can decompose $\rho$ into a combination of basis projectors:
\begin{align}
	\rho & = 
	\big(\cos^2\!\tfrac{\theta_A}{2}-\sin^2\!\tfrac{\theta_A}{2}\big)\big(\cos^2\!\tfrac{\theta_B}{2}\sin^2\!\tfrac{\theta_B}{2}\big)
	\begin{pmatrix}
		1 &  &  &  \\
		& 0 &  &  \\
		&  & 0 &  \\
		&  &  & 0 \\
	\end{pmatrix} \notag \\
	&\quad + 
	\big(\cos^2\!\tfrac{\theta_A}{2}-\sin^2\!\tfrac{\theta_A}{2}\big)\sin^2\!\tfrac{\theta_B}{2}
	\begin{pmatrix}
		1 &  &  &  \\
		& 1 &  &  \\
		&  & 0 &  \\
		&  &  & 0 \\
	\end{pmatrix} \notag \\
	&\quad+
	\big(\cos^2\!\tfrac{\theta_B}{2}-\sin^2\!\tfrac{\theta_B}{2}\big)\sin^2\!\tfrac{\theta_A}{2}
	\begin{pmatrix}
		1 &  &  &  \\
		& 0 &  &  \\
		&  & 1 &  \\
		&  &  & 0 \\
	\end{pmatrix} \notag \\
	&\quad +
	\sin^2\!\tfrac{\theta_A}{2}\sin^2\!\tfrac{\theta_B}{2}
	\begin{pmatrix}
		1 &  &  &  \\
		& 1 &  &  \\
		&  & 1 &  \\
		&  &  & 1 \\
	\end{pmatrix}.
	\label{eq.big}
\end{align}
	which corresponds to
\begin{align}
	\rho & = 
	\big(\cos^2\!\tfrac{\theta_A}{2}-\sin^2\!\tfrac{\theta_A}{2}\big)\big(\cos^2\!\tfrac{\theta_B}{2}\sin^2\!\tfrac{\theta_B}{2}\big)
	\left| HH \right\rangle \!\left\langle HH \right| \notag \\
	&\quad+
	\big(\cos^2\!\tfrac{\theta_A}{2}-\sin^2\!\tfrac{\theta_A}{2}\big)\sin^2\!\tfrac{\theta_B}{2} 
	\big( \left| H \right\rangle \!\left\langle H \right| \big)_A \otimes I_B \notag \\
	&\quad+
	\big(\cos^2\!\tfrac{\theta_B}{2}-\sin^2\!\tfrac{\theta_B}{2}\big)\sin^2\!\tfrac{\theta_A}{2}
	I_A \otimes \big( \left| H \right\rangle \!\left\langle H \right| \big)_B \notag \\
	&\quad+
	\sin^2\!\tfrac{\theta_A}{2}\sin^2\!\tfrac{\theta_B}{2}\,
	I_A \otimes I_B .
	\label{eq.big}
\end{align}
	
	Observe that the last three components each yield a QBER of $50\%$, whereas the first component corresponds to the perfectly encoded case, which necessarily produces a lower error rate than the mixed contribution \cite{Wang2023a}.
	
	Specializing to the $H$ polarization (for both users) and integrating over all pairs of polarized states \cite{Wang2023a}—each with a density matrix of the form in Eq.~\ref{eq.big}—the $HH$ ensemble is a combination of the mixture (the latter three terms in Eq.~\ref{eq.big}) and the perfect-encoding term (the first term of Eq.~\ref{eq.big}). Since the mixture terms contribute $50\%$ QBER, I conclude that
	\begin{equation}
		e_{11}Y_{11}^{\text{mixed}} \ge e_{11}Y_{11}^{\text{perfect}},
	\end{equation}
	and the same reasoning applies for $VV$, $HV$, and $VH$. An analogous statement holds in the $X$ basis for $++$, $--$, $+-$, and $-+$.
	
	Therefore, the mixed-state upper bound obtained via linear programming, $e_{nm}Y_{nm}^{\text{mixed, Upper}}$, indeed upper-bounds the perfect-encoding quantity $eY^{\text{perfect}}$:
	\begin{equation}
		e_{11}Y_{11}^{\text{mixed, Upper}} \ge e_{11}Y_{11}^{\text{mixed}}  \ge   e_{11}Y_{11}^{\text{perfect}}.
	\end{equation}

	\chapter{CHANNEL MODEL}
	\label{sec:mdi_channel}
	
	The optical signals emitted from Alice’s fully passive source are characterized by four degrees of freedom (DOFs): $(\mu_A, \theta_{HVA}, \phi_{HVA}, \phi_{\text{global},A})$. Bob’s source possesses an analogous set of parameters. As the signals propagate through the quantum channel, they experience slight polarization rotations. The Bloch-sphere representation of Alice’s initial polarization state is given by \cite{Wang2023a}
	\begin{equation}
		\Vec{s}_A = (\sin \theta_{HVA} \cos \phi_{HVA}, \sin \theta_{HVA} \sin \phi_{HVA}, \cos \theta_{HVA}).
	\end{equation}
	
	The channel-induced polarization rotation can be described using the Rodrigues rotation formula \cite{Friedberg2022}:
	\begin{equation}
		\Vec{s'}_A = \cos \alpha \, \Vec{s} + \sin \alpha \, (\Vec{r} \times \Vec{s}) + (\Vec{s} \cdot \Vec{r}) (1 - \cos \alpha)\, \Vec{r},
	\end{equation}
	where $\alpha$ is the rotation angle and $\Vec{r}$ is the unit vector defining the rotation axis. After this transformation, the rotated state can be re-expressed in Bloch-sphere coordinates as $\{\mu'_A, \theta'_{HVA}, \phi'_{HVA}, \phi'_{GA}\}$, with primes indicating post-rotation quantities. Although the global phase $\phi'_{GA}$ is also affected, it has no observable impact since it will ultimately be averaged out in the integration.
	
	The two relevant post-rotation parameters for Alice, $(\mu'_A, \theta'_{HVA})$, can be mapped back to the H- and V-mode intensities, $(\mu'_{HA}, \mu'_{VA})$. Taking into account both channel loss and detector efficiency, the intensity at the detectors becomes
	\begin{equation*}
		\mu'_{HA} \rightarrow \mu'_{HA} \eta = \mu'_{HA} \eta_L \eta_d = \mu'_{HA} 10^{-\alpha L/10} \eta_d,
	\end{equation*}
	where $\alpha$ is the channel loss coefficient (typically $\approx 0.2~\mathrm{dB/km}$), $\eta_L$ is the transmission factor, and $\eta_d$ is the detector efficiency. The same scaling applies to all other intensities.
	
	After propagation, the rotated signals reach a 50:50 beam splitter (BS), where interference occurs. The interference can be analyzed independently for the H and V polarization components, each following the standard interference relation:
\begin{align}
	&\left|\sqrt{\mu_{1}} e^{i \phi_{1}}\right\rangle_{a}\left|\sqrt{\mu_{2}} e^{i \phi_{2}}\right\rangle_{b} \rightarrow \notag \\
	&\left|\sqrt{\mu_{1}/2} e^{i \phi_{1}} + i \sqrt{\mu_{2}/2} e^{i \phi_{2}}\right\rangle_{c}
	\left|i \sqrt{\mu_{1}/2} e^{i \phi_{1}} + \sqrt{\mu_{2}/2} e^{i \phi_{2}}\right\rangle_{d}.
	\label{eq:beamsplitter_transform}
\end{align}
	
	From this, the mean output intensity at port $c$ is
	\begin{equation}
		\mu_c = \frac{\mu_1}{2} + \frac{\mu_2}{2} - \sqrt{\mu_1 \mu_2}\sin(\phi),
	\end{equation}
	where $\phi$ denotes the phase difference between the two input fields. Similarly, the output at port $d$ is
	\begin{equation}
		\mu_d = \frac{\mu_1}{2} + \frac{\mu_2}{2} + \sqrt{\mu_1 \mu_2}\sin(\phi).
	\end{equation}
	
	For the H polarization, the relevant intensities are $\mu'_{HA}$ and $\mu'_{HB}$, while for the V polarization they are $\mu'_{VA}$ and $\mu'_{VB}$. The corresponding phase differences are defined as
	
		\begin{equation}
		\phi'_R = \phi'_{HA} - \phi'_{HB} \quad (\text{for the H component}),
		\end{equation}
		\begin{equation}
		\phi'_R + \phi'_{HVA} - \phi'_{HVB} \quad (\text{for the V component}).
	\end{equation}
	
	Using the above relations, one can compute the photon intensities arriving at detectors 3H and 4H (from the H-channel interference) and at 3V and 4V (from the V-channel interference).
	
	Because the relative phases $\phi'_R$ and $\phi'_R + \phi'_{HVA} - \phi'_{HVB}$ are random, evaluating the average gain requires phase integration from $0$ to $2\pi$, as indicated in Eq.~\ref{eq:exp7d}. Although the integrations for H and V channels are distinct, the V-phase difference itself depends on $\phi'_R$. Therefore, integrating $\phi'_R$ over $[0, 2\pi]$ suffices to obtain the correct average gain.
	
	
	\renewcommand{\thesection}{\thechapter.\arabic{section}}
	\renewcommand{\theequation}{\thechapter.\arabic{equation}}
	\setcounter{section}{0}
	\setcounter{equation}{0}

	\chapter{DECOY-STATE-ANALYSIS -- TWO-DECOY}
	\label{sec:cka_lp}
	
	To determine the upper bounds of the photon-number yields $Y^{Z,j}_{m_A m_B m_C m_D}$, one can perform a linear-program analysis based on a set of relations for the four-user scenario~\cite{Wang2023,Carrara2023}:
	\begin{equation}
		G^j_{S_{ijkl}} = 
		\sum_{m_A, \ldots, m_D=0}^{\infty}
		Y^j_{m_A, \ldots, m_D}
		\left\langle P(m_A, m_B, m_C, m_D) \right\rangle_{S_{ijkl}}.
	\end{equation}
	Here $G^j_{S_{ijkl}}$ is an experimentally accessible quantity derived from the channel model: it equals the conditional probability of observing a single-click detection event during a PE round, given the decoy configuration $S_{ijkl}$, namely $G^j_{S_{ijkl}} = \Pr(\Omega_j|S_{ijkl})$~\cite{Carrara2023}.

	The mean photon-number distribution is defined as
\begin{align}
	\left\langle P^Z(m_A, \ldots, m_D)\right\rangle_{S_{ijkl}}
	&=\frac{1}{P_{S_{ijkl}}}
	\int_{S_{ijkl}}
	P_{\text{Poisson}}(\mu_A,m_A)
	\cdots
	P_{\text{Poisson}}(\mu_D,m_D) \notag \\
	&\quad \times P_{\text{int}}(\mu_A,\ldots,\mu_D)
	\, d\mu_A \cdots d\mu_D.
	\label{eq:PZ_avg_Sijkl}
\end{align}
	where $P_{\text{Poisson}}$ denotes the Poisson distribution and
	\begin{equation}
		P_{S_{ijkl}}=
		\int_{S_{ijkl}} P_{\text{int}}(\mu_A,\ldots,\mu_D)
		\, d\mu_A \cdots d\mu_D,
	\end{equation}
	with the joint intensity distribution given by
\begin{align}
	P_{\text{int}}(\mu_A,\ldots,\mu_D)
	&= \frac{1}{\pi^4 \sqrt{\mu_A(\mu_{\max}-\mu_A)} \sqrt{\mu_B(\mu_{\max}-\mu_B)}} \notag \\
	&\quad \times \frac{1}{\sqrt{\mu_C(\mu_{\max}-\mu_C)} \sqrt{\mu_D(\mu_{\max}-\mu_D)}}.
	\label{eq:Pint_def}
\end{align}
	The upper and lower bounds derived from the linear program can then be applied in Eq.~(B1) to correct for local transmission losses, followed by their use in the computation of the phase-error rate.

	\chapter{ACTIVE CKA KEY RATE}
	\label{sec:cka_active}

	In the active CKA protocol, assuming collective attacks and one-way classical post-processing, one can lower bound the asymptotic key rate by
	
	\begin{align}
		r \ge
		\sum_{j=0}^{M-1}
		\Pr(\Omega_j|\mathrm{KG})
		\Big[
		H(X_0|E)_{\Omega_j}
		- 
		\max_{i\ge1} H(X_0|X_i)_{\Omega_j}
		\Big],
	\end{align}
	where each successful post-selected event $\Omega_j$ corresponds to an independent key segment, and all segments are combined to yield the final secret key. Here, $M$ is the total number of detectors. The quantity $H(X_0|E)$ denotes the entropy of the first user’s ($A_0$’s) outcome conditioned on the eavesdropper’s information, whereas $H(X_0|X_i)$ represents the entropy of $A_0$’s measurement conditioned on another user $A_i$’s result.

	\chapter{SECURITY}
	\label{sec:cka_security}
	
	The imperfect preparation from the two arms of each fully passive source can be represented as a phase modulation within $[-\Delta,\Delta]$, as depicted in Fig.~\ref{fig:regions}. As explained in the main text, I conservatively assign the complete local channel $\mathcal{E}_{A_i}$—including both phase modulators and beam splitters—to Eve’s control. This results in an increased QBER in the key-generation (KG) basis and requires a refined treatment of yield estimation during the parameter-estimation (PE) phase.
	
	For a four-party fully passive CKA system, the yields $Y^{Z,j}_{m_A m_B m_C m_D}$ (with photon numbers $m_A, m_B, m_C, m_D$ and detection event $\Omega_j$) are obtained via decoy-state analysis. To emphasize basis choice, the superscript $Z$ is retained here but omitted elsewhere for simplicity. When including the impact of local loss channels, the yield correction takes the form
\begin{align}
	Y^j_{n_A n_B n_C n_D}
	&= \int_{-\Delta_\phi}^{\Delta_\phi}\!\!\cdots\!\!\int_{-\Delta_\phi}^{\Delta_\phi}
	\sum_{m_A=0}^{n_A}\!\!\sum_{m_B=0}^{n_B}\!\!\sum_{m_C=0}^{n_C}\!\!\sum_{m_D=0}^{n_D} \notag \\
	&\quad \times P_{\phi_{A1},\phi_{A2}}(m_A|n_A) P_{\phi_{B1},\phi_{B2}}(m_B|n_B) \notag \\
	&\quad \times P_{\phi_{C1},\phi_{C2}}(m_C|n_C) P_{\phi_{D1},\phi_{D2}}(m_D|n_D) \notag \\
	&\quad \times Y^j_{m_A m_B m_C m_D} \, d\phi_{A1}d\phi_{A2}\cdots d\phi_{D2}.
	\label{eq:Yj_nAnBnCnD}
\end{align}
	where $P_{\phi_{A1},\phi_{A2}}(m_A|n_A)$ represents the conditional probability that party $A$’s photon number is reduced from $n_A$ to $m_A$ due to local loss~\cite{Wang2023}:
\begin{align}
	P_{\phi_{A1},\phi_{A2}}(m_A|n_A)
	&= \Biggl| \frac{1}{2!} \sqrt{\frac{n_A!}{m_A!(n_A-m_A)!}}
	\Bigl[1+e^{i(\pi/2+\phi_{A2}-\phi_{A1})}\Bigr]^{m_A} \notag \\
	&\quad \times
	\Bigl[1-e^{i(\pi/2+\phi_{A2}-\phi_{A1})}\Bigr]^{(n_A-m_A)} \Biggr|^2.
	\label{eq:Pphi_A}
\end{align}
	The loss-adjusted yields can then be fed into the standard active-CKA phase-error estimation procedure~\cite{Carrara2023}:
	
	\begin{equation}
		\bar{Q}_Z^j =
		\frac{1}{\Pr(\Omega_j|\mathrm{KG})}
		\sum_{v\in\mathcal{V}}
		\left(
		\sum_{n_0+\cdots+n_{N-1}\le\bar{n}}
		\prod_{i=0}^{N-1} c_{i,n_i}^{(v_i)} 
		\sqrt{\bar{Y}^j_{n_0,\ldots,n_{N-1}}}
		+
		\Delta_{v,\bar{n}}
		\right)^2,
	\end{equation}
	where $\Pr(\Omega_j|\mathrm{KG})$ denotes the conditional probability of a single-click event given a KG round, $N$ is the number of users, and $\bar{n}$ is the photon-number cutoff in the decoy analysis. The remaining terms are defined as
	\begin{equation}
		\mathcal{V} = 
		\{v \in \{0,2^N-1\} : |\vec{v}| \bmod 2 = 0 \},
	\end{equation}
	\begin{equation}
		c_{i,n}^{(l)} =
		\begin{cases}
			e^{-\alpha_i^2/2} \dfrac{\alpha_i^n}{\sqrt{n!}}, & n+l \text{ even},\\[1ex]
			0, & n+l \text{ odd},
		\end{cases}
	\end{equation}
	\begin{equation}
		\Delta_{v,\bar{n}} =
		\sum_{n_0+\cdots+n_{N-1}\ge\bar{n}+2}
		\prod_{i=0}^{N-1} c_{i,n_i}^{(v_i)}.
	\end{equation}

	\chapter{BRANCH-CUTTING METHOD}
	\label{sec:cka_branch}
	
	In the context of fully passive CKA, the computational load can be reduced by excluding slice-choice combinations that make negligible or zero contribution to the total key rate. These excluded combinations satisfy at least one of the following constraints:
	
	\begin{itemize}
		\item \textbf{Intra-party phase mismatch too large:}  
		For each user $i$, let $k_{i1}$ and $k_{i2}$ denote the indices of the two slice selections. The condition $|k_{i1}-k_{i2}|\le x$ must hold for all $i$.
		
		\item \textbf{Inter-party phase misalignment excessive:}  
		For any pair of users $(i,j)$,
		\begin{equation}
			\frac{k_{i1}+k_{i2}}{2} - \frac{k_{j1}+k_{j2}}{2} \le y.
		\end{equation}
	\end{itemize}
	
	Here, $x$ and $y$ are chosen as positive integers strictly less than the total number of phase slices $M$. For the numerical simulations presented in this work, I fix $M=8$ and set $x=y=2$.

	\setcounter{section}{0}
	\setcounter{equation}{0}
	\renewcommand{\thesection}{\thechapter.\arabic{section}}
	\renewcommand{\theequation}{\thechapter.\arabic{equation}}

\printbibliography[heading=bibintoc, title=REFERENCES]

@article{trss-qhc8,
	title = {Quantum speed limits for open system dynamics based on a representation-basis-dependent ${\ensuremath{\ell}}_{w}^{p}$-seminorm},
	author = {Chau, H. F. and Li, Jinjie},
	journal = {Phys. Rev. A},
	volume = {113},
	issue = {4},
	pages = {042218},
	numpages = {13},
	year = {2026},
	month = {Apr},
	publisher = {American Physical Society},
	doi = {10.1103/trss-qhc8},
	url = {https://link.aps.org/doi/10.1103/trss-qhc8}
}

@inproceedings{10.1117/12.507486,
	author = {Vittorio Giovannetti and Seth Lloyd and Lorenzo Maccone},
	title = {{The quantum speed limit}},
	volume = {5111},
	booktitle = {Fluctuations and Noise in Photonics and Quantum Optics},
	editor = {Derek Abbott and Jeffrey H. Shapiro and Yoshihisa Yamamoto},
	organization = {International Society for Optics and Photonics},
	publisher = {SPIE},
	pages = {1 -- 6},
	year = {2003},
	doi = {10.1117/12.507486},
	URL = {https://doi.org/10.1117/12.507486}
}

@inbook{Watrous_2018, place={Cambridge}, title={Contents}, booktitle={The Theory of Quantum Information}, publisher={Cambridge University Press}, author={Watrous, John}, year={2018}, pages={v–vi}}

@article{Aaronson2016,
	author    = {Scott Aaronson},
	title     = {The Complexity of Quantum States and Transformations: From Quantum Money to Black Holes},
	journal   = {arXiv:1607.05256},
	year      = {2016}
}

@article{Brunner2014RMP,
	author    = {Nicolas Brunner and Daniel Cavalcanti and Stefano Pironio and Valerio Scarani and Stephanie Wehner},
	title     = {Bell nonlocality},
	journal   = {Rev. Mod. Phys.},
	volume    = {86},
	pages     = {419--478},
	year      = {2014}
}

@misc{LiChau_AQIS2025,
	author       = {Jinjie Li and H. F. Chau},
	title        = {A New Quantum Speed Limit Based on Unitarily Invariant Norms},
	howpublished = {Contributed talk at the 25th Asian Quantum Information Science (AQIS) Conference},
	year         = {2025},
	note         = {Held in 2025},
}

@article{Carrara2023,
	author       = {Carrara, G. and Murta, G. and Grasselli, F.},
	title        = {Overcoming Fundamental Bounds on Quantum Conference Key Agreement},
	journal      = {Physical Review Applied},
	volume       = {19},
	number       = {6},
	pages        = {064017},
	year         = {2023},
	doi          = {10.1103/PhysRevApplied.19.064017}
}

@article{Shannon1948,
	author  = {Shannon, Claude E.},
	title   = {A Mathematical Theory of Communication},
	journal = {Bell System Technical Journal},
	year    = {1948},
	volume  = {27},
	number  = {3},
	pages   = {379--423}
}

@article{Moore1965,
	author  = {Moore, Gordon E.},
	title   = {Cramming More Components onto Integrated Circuits},
	journal = {Electronics},
	year    = {1965},
	volume  = {38},
	number  = {8},
	pages   = {114--117}
}

@article{Dennard1974,
	author  = {Dennard, R. H. and Gaensslen, F. H. and Yu, H.-N. and Rideout, V. L. and Bassous, E. and LeBlanc, A. R.},
	title   = {Design of Ion-Implanted {MOSFET}'s with Very Small Physical Dimensions},
	journal = {IEEE Journal of Solid-State Circuits},
	year    = {1974},
	volume  = {9},
	number  = {5},
	pages   = {256--268}
}

@article{Pop2010,
	author  = {Pop, Eric},
	title   = {Energy Dissipation and Transport in Nanoscale Devices},
	journal = {Nano Research},
	year    = {2010},
	volume  = {3},
	number  = {3},
	pages   = {147--169}
}

@article{Keyes2005,
	author  = {Keyes, Robert W.},
	title   = {Physical Limits of Silicon Transistors and Circuits},
	journal = {Reports on Progress in Physics},
	year    = {2005},
	volume  = {68},
	number  = {12},
	pages   = {2701--2746}
}

@article{Feynman1982,
	author  = {Feynman, Richard P.},
	title   = {Simulating Physics with Computers},
	journal = {International Journal of Theoretical Physics},
	year    = {1982},
	volume  = {21},
	number  = {6--7},
	pages   = {467--488}
}

@book{Manin1980,
	author    = {Manin, Yuri I.},
	title     = {Computable and Noncomputable},
	publisher = {Sovetskoye Radio},
	address   = {Moscow},
	year      = {1980},
	note      = {In Russian}
}

@article{Benioff1980,
	author  = {Benioff, Paul},
	title   = {The Computer as a Physical System: A Microscopic Quantum Mechanical Hamiltonian Model of Computers as Represented by {Turing} Machines},
	journal = {Journal of Statistical Physics},
	year    = {1980},
	volume  = {22},
	number  = {5},
	pages   = {563--591}
}

@article{Horodecki2009,
	author  = {Horodecki, R. and Horodecki, P. and Horodecki, M. and Horodecki, K.},
	title   = {Quantum Entanglement},
	journal = {Reviews of Modern Physics},
	year    = {2009},
	volume  = {81},
	number  = {2},
	pages   = {865--942}
}

@inproceedings{Shor1994,
	author    = {Shor, Peter W.},
	title     = {Algorithms for Quantum Computation: Discrete Logarithms and Factoring},
	booktitle = {Proceedings of the 35th Annual Symposium on Foundations of Computer Science (FOCS)},
	year      = {1994},
	pages     = {124--134},
	publisher = {IEEE}
}

@inproceedings{Grover1996,
	author    = {Grover, Lov K.},
	title     = {A Fast Quantum Mechanical Algorithm for Database Search},
	booktitle = {Proceedings of the 28th Annual ACM Symposium on Theory of Computing (STOC)},
	year      = {1996},
	pages     = {212--219},
	publisher = {ACM}
}

@article{Bennett1993,
	author  = {Bennett, Charles H. and Brassard, Gilles and Cr\'epeau, Claude and Jozsa, Richard and Peres, Asher and Wootters, William K.},
	title   = {Teleporting an Unknown Quantum State via Dual Classical and Einstein–Podolsky–Rosen Channels},
	journal = {Physical Review Letters},
	year    = {1993},
	volume  = {70},
	number  = {13},
	pages   = {1895--1899}
}

@article{Briegel1998,
	author  = {Briegel, H.-J. and D\"ur, W. and Cirac, J. I. and Zoller, P.},
	title   = {Quantum Repeaters: The Role of Imperfect Local Operations in Quantum Communication},
	journal = {Physical Review Letters},
	year    = {1998},
	volume  = {81},
	number  = {26},
	pages   = {5932--5935}
}

@article{Kimble2008,
	author  = {Kimble, H. J.},
	title   = {The Quantum Internet},
	journal = {Nature},
	year    = {2008},
	volume  = {453},
	pages   = {1023--1030}
}

@article{Giovannetti04,
	author    = {V. Giovannetti and S. Lloyd and L. Maccone},
	title     = {Quantum-Enhanced Measurements: Beating the Standard Quantum Limit},
	journal   = {Science},
	volume    = {306},
	issue     = {5700},
	pages     = {1330--1336},
	year      = {2004}
}

@article{Giovannetti2011,
	author  = {Giovannetti, Vittorio and Lloyd, Seth and Maccone, Lorenzo},
	title   = {Advances in Quantum Metrology},
	journal = {Nature Photonics},
	year    = {2011},
	volume  = {5},
	pages   = {222--229}
}

@article{Acin2018,
	author  = {Ac\'in, Antonio and others},
	title   = {The Quantum Technologies Roadmap: A European Perspective},
	journal = {New Journal of Physics},
	year    = {2018},
	volume  = {20},
	number  = {8},
	pages   = {080201}
}

@book{NielsenChuang,
	author    = {Nielsen, Michael A. and Chuang, Isaac L.},
	title     = {Quantum Computation and Quantum Information},
	publisher = {Cambridge University Press},
	edition   = {10th Anniversary Edition},
	year      = {2010}
}

@article{Rivest1978,
	author  = {Rivest, Ronald L. and Shamir, Adi and Adleman, Leonard},
	title   = {A Method for Obtaining Digital Signatures and Public-Key Cryptosystems},
	journal = {Communications of the ACM},
	year    = {1978},
	volume  = {21},
	number  = {2},
	pages   = {120--126}
}

@article{Ekert1991,
	author = {Artur K. Ekert},
	journal = {Physical Review Letters},
	pages = {661-663},
	title = {Quantum cryptography based on Bell’s theorem},
	volume = {67},
	year = {1991}
}

@article{Wootters1982,
	author  = {Wootters, W. K. and Zurek, W. H.},
	title   = {A Single Quantum Cannot be Cloned},
	journal = {Nature},
	year    = {1982},
	volume  = {299},
	pages   = {802--803}
}

@article{Bennett1984,
	title = {Quantum cryptography: Public key distribution and coin tossing},
	journal = {Theoretical Computer Science},
	volume = {560},
	pages = {7-11},
	year = {2014},
	note = {Theoretical Aspects of Quantum Cryptography – celebrating 30 years of BB84},
	issn = {0304-3975},
	doi = {https://doi.org/10.1016/j.tcs.2014.05.025},
	url = {https://www.sciencedirect.com/science/article/pii/S0304397514004241},
	author = {Charles H. Bennett and Gilles Brassard}
}

@article{Gisin2002,
	author  = {Gisin, Nicolas and Ribordy, Gr\'egoire and Tittel, Wolfgang and Zbinden, Hugo},
	title   = {Quantum Cryptography},
	journal = {Reviews of Modern Physics},
	year    = {2002},
	volume  = {74},
	number  = {1},
	pages   = {145--195}
}

@article{Scarani2009,
	author  = {Scarani, Valerio and Bechmann-Pasquinucci, Helle and Cerf, Nicolas J. and Du{\v{s}}ek, Miloslav and L\"utkenhaus, Norbert and Peev, Momtchil},
	title   = {The Security of Practical Quantum Key Distribution},
	journal = {Reviews of Modern Physics},
	year    = {2009},
	volume  = {81},
	number  = {3},
	pages   = {1301--1350}
}

@article{Lo2014,
	author  = {Lo, Hoi-Kwong and Curty, Marcos and Tamaki, Kiyoshi},
	title   = {Secure Quantum Key Distribution},
	journal = {Nature Photonics},
	year    = {2014},
	volume  = {8},
	pages   = {595--604}
}

@article{Diamanti2016,
	author  = {Diamanti, Eleni and Lo, Hoi-Kwong and Qi, Bing and Yuan, Zhiliang},
	title   = {Practical Challenges in Quantum Key Distribution},
	journal = {npj Quantum Information},
	year    = {2016},
	volume  = {2},
	pages   = {16025}
}

@article{Xu2020review,
	author = {Feihu Xu and Xiongfeng Ma and Qiang Zhang and Hoi-Kwong Lo and Jian-Wei Pan},
	journal = {Reviews of Modern Physics},
	pages = {025002},
	title = {Secure quantum key distribution with realistic devices},
	volume = {92},
	year = {2020}
}

@article{Liao2017,
	author  = {Liao, Sheng-Kai and others},
	title   = {Satellite-to-Ground Quantum Key Distribution},
	journal = {Nature},
	year    = {2017},
	volume  = {549},
	pages   = {43--47}
}

@article{Chen2021,
	author  = {Chen, Yu-Ao and others},
	title   = {An Integrated Space-to-Ground Quantum Communication Network over 4,600 km},
	journal = {Nature},
	year    = {2021},
	volume  = {589},
	pages   = {214--219}
}

@article{Lo2005,
	author = {Hoi-Kwong Lo and Xiongfeng Ma and Kai Chen},
	journal = {Physical Review Letters},
	pages = {230504},
	title = {Decoy State Quantum Key Distribution},
	volume = {94},
	year = {2005}
}

@article{Wang2005,
	author = {Xiang-Bin Wang},
	journal = {Physical Review Letters},
	pages = {230503},
	title = {Beating the Photon-Number-Splitting Attack in Practical Quantum Cryptography},
	volume = {94},
	year = {2005}
}

@article{Ma2005,
	author  = {Ma, Xiongfeng and Qi, Bing and Zhao, Yi and Lo, Hoi-Kwong},
	title   = {Practical Decoy State for Quantum Key Distribution},
	journal = {Physical Review A},
	year    = {2005},
	volume  = {72},
	number  = {1},
	pages   = {012326}
}

@article{Lo2012,
	author = {Hoi-Kwong Lo and Marcos Curty and Bing Qi},
	journal = {Physical Review Letters},
	pages = {130503},
	title = {Measurement-Device-Independent Quantum Key Distribution},
	volume = {108},
	year = {2012}
}

@article{Curty2014,
	author = {Marcos Curty and Feihu Xu and Wei Cui and Charles Ci Wen Lim and Kiyoshi Tamaki and Hoi-Kwong Lo},
	journal = {Nature Communications},
	pages = {3732},
	title = {Finite-key analysis for measurement-device-independent quantum key distribution},
	volume = {5},
	year = {2014}
}

@article{Xu2013,
	author = {Feihu Xu and Marcos Curty and Bing Qi and Hoi Kwong Lo},
	journal = {New Journal of Physics},
	title = {Practical aspects of measurement-device-independent quantum key distribution},
	volume = {15},
	pages = {113007},
	year = {2013}
}

@article{Wang2023a,
	author = {Wenyuan Wang and Rong Wang and Chengqiu Hu and Victor Zapatero and Li Qian and Bing Qi and Marcos Curty and Hoi-Kwong Lo},
	journal = {Physical Review Letters},
	pages = {220801},
	title = {Fully Passive Quantum Key Distribution},
	volume = {130},
	year = {2023}
}

@article{Hu2023,
	author = {Chengqiu Hu and Wenyuan Wang and Kai-Sum Chan and Zhenghan Yuan and Hoi-Kwong Lo},
	journal = {Physical Review Letters},
	pages = {110801},
	title = {Proof-of-Principle Demonstration of Fully Passive Quantum Key Distribution},
	volume = {131},
	year = {2023}
}

@article{Lu2023,
	author = {Feng-Yu Lu and others},
	journal = {Physical Review Letters},
	pages = {110802},
	title = {Experimental Demonstration of Fully Passive Quantum Key Distribution},
	volume = {131},
	year = {2023}
}

@article{Mandelstam45,
	author  = {Mandelstam, L. and Tamm, I.},
	title   = {{The uncertainty relation between energy and time in non-relativistic quantum mechanics}},
	journal = {J. Phys. (USSR)},
	volume  = {9},
	pages   = {249--254},
	year    = {1945}
}

@article{Fleming73,
	author  = {Fleming, G. N.},
	title   = {A Unitarity Bound on the Evolution of Nonstationary States},
	journal = {Il Nuovo Cimento A},
	year    = {1973},
	volume  = {16},
	number  = {2},
	pages   = {232--240}
}

@article{Anandan90,
	author  = {J. Anandan and Y. Aharonov},
	title   = {Geometry of quantum evolution},
	journal = {Phys. Rev. Lett.},
	volume  = {65},
	pages   = {1697},
	year    = {1990}
}

@article{Vaidman92,
	author  = {Vaidman, Lev},
	title   = {Minimum Time for the Evolution to an Orthogonal Quantum State},
	journal = {American Journal of Physics},
	year    = {1992},
	volume  = {60},
	number  = {2},
	pages   = {182--183}
}

@article{Uhlmann92,
	title = {An energy dispersion estimate},
	journal = {Physics Letters A},
	volume = {161},
	number = {4},
	pages = {329-331},
	year = {1992},
	issn = {0375-9601},
	doi = {https://doi.org/10.1016/0375-9601(92)90555-Z},
	author = {Armin Uhlmann}
}

@article{Pires2016,
	title = {Generalized Geometric Quantum Speed Limits},
	author = {Pires, Diego Paiva and Cianciaruso, Marco and C\'eleri, Lucas C. and Adesso, Gerardo and Soares-Pinto, Diogo O.},
	journal = {Phys. Rev. X},
	volume = {6},
	issue = {2},
	pages = {021031},
	numpages = {19},
	year = {2016},
	month = {Jun},
	publisher = {American Physical Society},
	doi = {10.1103/PhysRevX.6.021031},
	url = {https://link.aps.org/doi/10.1103/PhysRevX.6.021031}
}

@article{Deffner17a,
	author    = {S. Deffner and S. Campbell},
	title     = {Quantum speed limits: from Heisenberg’s uncertainty principle to optimal quantum control},
	journal   = {J. Phys. A},
	volume    = {50},
	number    = {45},
	pages     = {453001},
	year      = {2017}
}

@article{Giovannetti03,
	author  = {Giovannetti, Vittorio and Lloyd, Seth and Maccone, Lorenzo},
	title   = {Quantum Limits to Dynamical Evolution},
	journal = {Physical Review A},
	year    = {2003},
	volume  = {67},
	number  = {5},
	pages   = {052109}
}

@article{Pang2017,
	author  = {Pang, Shengshi and Jordan, Andrew N.},
	title   = {Optimal Adaptive Control for Quantum Metrology with Time-Dependent Hamiltonians},
	journal = {Nature Communications},
	year    = {2017},
	volume  = {8},
	pages   = {14695}
}

@article{Binder2015,
	title = {Quantum thermodynamics of general quantum processes},
	author = {Binder, Felix and Vinjanampathy, Sai and Modi, Kavan and Goold, John},
	journal = {Phys. Rev. E},
	volume = {91},
	issue = {3},
	pages = {032119},
	numpages = {6},
	year = {2015},
	month = {Mar},
	publisher = {American Physical Society},
	doi = {10.1103/PhysRevE.91.032119},
	url = {https://link.aps.org/doi/10.1103/PhysRevE.91.032119}
}

@article{Shannon1949,
	author  = {Shannon, Claude E.},
	title   = {Communication Theory of Secrecy Systems},
	journal = {Bell System Technical Journal},
	year    = {1949},
	volume  = {28},
	number  = {4},
	pages   = {656--715}
}

@article{Maurer1993,
	author  = {Maurer, Ueli M.},
	title   = {Secret Key Agreement by Public Discussion from Common Information},
	journal = {IEEE Transactions on Information Theory},
	year    = {1993},
	volume  = {39},
	number  = {3},
	pages   = {733--742}
}

@misc{Daemen1999,
	author    = {Daemen, Joan and Rijmen, Vincent},
	title     = {AES Proposal: Rijndael},
	booktitle = {Proceedings of the First Advanced Encryption Standard (AES) Conference},
	year      = {1999},
	address   = {NIST, Gaithersburg, MD}
}

@ARTICLE{Vernam1926,
	author={Vernam, G. S.},
	journal={Transactions of the American Institute of Electrical Engineers}, 
	title={Cipher Printing Telegraph Systems For Secret Wire and Radio Telegraphic Communications}, 
	year={1926},
	volume={XLV},
	number={},
	pages={295-301},
	doi={10.1109/T-AIEE.1926.5061224}}

@article{Townsend1993,
	author  = {Townsend, P. D.},
	title   = {Quantum Cryptography on Multiuser Optical Fibre Networks},
	journal = {Nature},
	year    = {1997},
	volume  = {385},
	pages   = {47--49}
}

@article{Wegman1981,
	author  = {Wegman, Mark N. and Carter, J. Lawrence},
	title   = {New Hash Functions and Their Use in Authentication and Set Equality},
	journal = {Journal of Computer and System Sciences},
	year    = {1981},
	volume  = {22},
	number  = {3},
	pages   = {265--279}
}

@article{Bennett1992,
	author  = {Bennett, Charles H. and Bessette, Fran\c{c}ois and Brassard, Gilles and Salvail, Louis and Smolin, John},
	title   = {Experimental Quantum Cryptography},
	journal = {Journal of Cryptology},
	year    = {1992},
	volume  = {5},
	number  = {1},
	pages   = {3--28}
}

@article{ShorPreskill2000,
	author  = {Shor, Peter W. and Preskill, John},
	title   = {Simple Proof of Security of the BB84 Quantum Key Distribution Protocol},
	journal = {Physical Review Letters},
	year    = {2000},
	volume  = {85},
	number  = {2},
	pages   = {441--444}
}

@inproceedings{Brassard1993,
	author    = {Brassard, Gilles and Salvail, Louis},
	title     = {Secret-Key Reconciliation by Public Discussion},
	booktitle = {Advances in Cryptology—EUROCRYPT ’93},
	year      = {1993},
	series    = {Lecture Notes in Computer Science},
	volume    = {765},
	pages     = {410--423},
	publisher = {Springer}
}

@phdthesis{Gallager1962,
	author  = {Gallager, Robert G.},
	title   = {Low-Density Parity-Check Codes},
	school  = {MIT},
	year    = {1962}
}

@article{murta,
	author = {Murta, Gláucia and Grasselli, Federico and Kampermann, Hermann and Bruß, Dagmar},
	title = {Quantum Conference Key Agreement: A Review},
	journal = {Advanced Quantum Technologies},
	volume = {3},
	number = {11},
	pages = {2000025},
	doi = {https://doi.org/10.1002/qute.202000025},
	url = {https://advanced.onlinelibrary.wiley.com/doi/abs/10.1002/qute.202000025},
	eprint = {https://advanced.onlinelibrary.wiley.com/doi/pdf/10.1002/qute.202000025},
	year = {2020}
}

@article{Bennett1988,
	author  = {Bennett, Charles H. and Brassard, Gilles and Robert, Jean-Marc},
	title   = {Privacy Amplification by Public Discussion},
	journal = {SIAM Journal on Computing},
	year    = {1988},
	volume  = {17},
	number  = {2},
	pages   = {210--229}
}

@misc{Renner2008,
	title={Security of Quantum Key Distribution}, 
	author={Renato Renner},
	year={2006},
	eprint={quant-ph/0512258},
	archivePrefix={arXiv},
	primaryClass={quant-ph},
	url={https://arxiv.org/abs/quant-ph/0512258}, 
}

@book{vonNeumann1932,
	author    = {von Neumann, John},
	title     = {Mathematical Foundations of Quantum Mechanics},
	year      = {1932},
	publisher = {Springer},
	note      = {English translation: Princeton University Press (1955)}
}

@article{Bloch1946,
	author  = {Bloch, Felix},
	title   = {Nuclear Induction},
	journal = {Physical Review},
	year    = {1946},
	volume  = {70},
	number  = {7-8},
	pages   = {460--474}
}

@article{Fano1957,
	author  = {Fano, Ugo},
	title   = {Description of States in Quantum Mechanics by Density Matrix and Operator Techniques},
	journal = {Reviews of Modern Physics},
	year    = {1957},
	volume  = {29},
	number  = {1},
	pages   = {74--93}
}

@article{Brendel1999,
	author  = {Brendel, J. and Gisin, N. and Tittel, W. and Zbinden, H.},
	title   = {Pulsed Energy-Time Entangled Twin-Photon Source for Quantum Communication},
	journal = {Physical Review Letters},
	year    = {1999},
	volume  = {82},
	number  = {12},
	pages   = {2594--2597}
}

@article{Tittel2000,
	author  = {Tittel, W. and Brendel, J. and Zbinden, H. and Gisin, N.},
	title   = {Quantum Cryptography Using Entangled Photons in Energy-Time Bell States},
	journal = {Physical Review Letters},
	year    = {2000},
	volume  = {84},
	number  = {20},
	pages   = {4737--4740}
}

@article{Stucki2002,
	author  = {Stucki, D. and Gisin, N. and Guinnard, O. and Ribordy, G. and Zbinden, H.},
	title   = {Quantum Key Distribution over 67 km with a Plug \& Play System},
	journal = {New Journal of Physics},
	year    = {2002},
	volume  = {4},
	pages   = {41}
}

@article{Weedbrook2012,
	author  = {Weedbrook, Christian and Pirandola, Stefano and others},
	title   = {Gaussian Quantum Information},
	journal = {Reviews of Modern Physics},
	year    = {2012},
	volume  = {84},
	number  = {2},
	pages   = {621--669}
}

@article{DevetakWinter2005,
	author  = {Devetak, Igor and Winter, Andreas},
	title   = {Distillation of Secret Key and Entanglement from Quantum States},
	journal = {Proceedings of the Royal Society A},
	year    = {2005},
	volume  = {461},
	number  = {2053},
	pages   = {207--235}
}

@article{Tomamichel2012,
	author  = {Tomamichel, Marco and Lim, Charles Ci Wen and Gisin, Nicolas and Renner, Renato},
	title   = {Tight Finite-Key Analysis for Quantum Cryptography},
	journal = {Nature Communications},
	year    = {2012},
	volume  = {3},
	pages   = {634}
}

@article{Lim2014,
	author = {Charles Ci Wen Lim and Marcos Curty and Nino Walenta and Feihu Xu and Hugo Zbinden},
	journal = {Physical Review A},
	title = {Concise security bounds for practical decoy-state quantum key distribution},
	volume = {89},
	year = {2014}
}

@article{Dieks1982,
	author  = {Dieks, Dennis},
	title   = {Communication by EPR Devices},
	journal = {Physics Letters A},
	year    = {1982},
	volume  = {92},
	number  = {6},
	pages   = {271--272}
}

@article{Brassard2000,
	author  = {Brassard, Gilles and L{\"u}tkenhaus, Norbert and Mor, Tal and Sanders, Barry C.},
	title   = {Limitations on Practical Quantum Cryptography},
	journal = {Physical Review Letters},
	year    = {2000},
	volume  = {85},
	number  = {6},
	pages   = {1330--1333}
}

@article{Lutkenhaus2000,
	author  = {L{\"u}tkenhaus, Norbert},
	title   = {Security Against Individual Attacks for Realistic Quantum Key Distribution},
	journal = {Physical Review A},
	year    = {2000},
	volume  = {61},
	number  = {5},
	pages   = {052304}
}

@article{Gisin2006,
	author  = {Gisin, Nicolas and Fasel, S. and Kraus, Barbara and Zbinden, Hugo and Ribordy, Gr{\'e}goire},
	title   = {Trojan-Horse Attacks on Quantum-Key-Distribution Systems},
	journal = {Physical Review A},
	year    = {2006},
	volume  = {73},
	number  = {2},
	pages   = {022320}
}

@article{Jain2014,
	doi = {10.1088/1367-2630/16/12/123030},
	url = {https://doi.org/10.1088/1367-2630/16/12/123030},
	year = {2014},
	month = {dec},
	publisher = {IOP Publishing},
	volume = {16},
	number = {12},
	pages = {123030},
	author = {Jain, Nitin and Anisimova, Elena and Khan, Imran and Makarov, Vadim and Marquardt, Christoph and Leuchs, Gerd},
	title = {Trojan-horse attacks threaten the security of practical quantum cryptography},
	journal = {New Journal of Physics}
}

@article{Hwang2003,
	author  = {Hwang, Won-Young},
	title   = {Quantum Key Distribution with High Loss: Toward Global Secure Communication},
	journal = {Physical Review Letters},
	year    = {2003},
	volume  = {91},
	number  = {5},
	pages   = {057901}
}

@article{Braunstein2012,
	author  = {Braunstein, Samuel L. and Pirandola, Stefano},
	title   = {Side-Channel-Free Quantum Key Distribution},
	journal = {Physical Review Letters},
	year    = {2012},
	volume  = {108},
	number  = {13},
	pages   = {130502}
}

@ARTICLE{Makarov2006,
	author={Lucamarini, Marco and Dynes, James F. and Fröhlich, Bernd and Yuan, Zhiliang and Shields, Andrew J.},
	journal={IEEE Journal of Selected Topics in Quantum Electronics}, 
	title={Security Bounds for Efficient Decoy-State Quantum Key Distribution}, 
	year={2015},
	volume={21},
	number={3},
	pages={197-204},
	doi={10.1109/JSTQE.2015.2394774}}

@article{Lydersen2010,
	author  = {Lydersen, Lars and others},
	title   = {Hacking Commercial Quantum Cryptography Systems by Tailored Bright Illumination},
	journal = {Nature Photonics},
	year    = {2010},
	volume  = {4},
	pages   = {686--689}
}

@article{Gerhardt2011,
	author  = {Gerhardt, Ilja and others},
	title   = {Full-Field Implementation of a Perfect Eavesdropper on a Quantum Cryptography System},
	journal = {Nature Communications},
	year    = {2011},
	volume  = {2},
	pages   = {349}
}

@article{Curty2010,
	author  = {Curty, Marcos and Ma, Xiongfeng and Qi, Bing and Moroder, Tobias},
	title   = {Passive Decoy-State Quantum Key Distribution with Practical Light Sources},
	journal = {Physical Review A},
	year    = {2010},
	volume  = {81},
	number  = {2},
	pages   = {022310}
}

@article{Lucamarini2018,
	author  = {Lucamarini, Marco and Yuan, Zhiliang L. and Dynes, James F. and Shields, Andrew J.},
	title   = {Overcoming the Rate--Distance Limit of Quantum Key Distribution Without Quantum Repeaters},
	journal = {Nature},
	year    = {2018},
	volume  = {557},
	pages   = {400--403}
}

@article{Pirandola2017PLOB,
	title   = {Fundamental Limits of Repeaterless Quantum Communications},
	author  = {Pirandola, Stefano and Laurenza, Riccardo and Ottaviani, Carlo and Banchi, Leonardo},
	journal = {Nature Communications},
	volume  = {8},
	pages   = {15043},
	year    = {2017}
}

@article{Grasselli2019CKA,
	title   = {Conference Key Agreement with Single-Photon Interference},
	author  = {Grasselli, Federico and Kampermann, Hermann and Bru{\ss}, Dagmar},
	journal = {New Journal of Physics},
	volume  = {21},
	number  = {12},
	pages   = {123002},
	year    = {2019}
}

@article{Curty2019,
	title   = {Simple Security Proof of Twin-Field Type Quantum Key Distribution Protocols},
	author  = {Curty, Marcos and Azuma, Koji and Lo, Hoi-Kwong},
	journal = {npj Quantum Information},
	volume  = {5},
	pages   = {64},
	year    = {2019}
}

@misc{Ma2008,
	title={Quantum cryptography: theory and practice}, 
	author={Xiongfeng Ma},
	year={2008},
	eprint={0808.1385},
	archivePrefix={arXiv},
	primaryClass={quant-ph},
	url={https://arxiv.org/abs/0808.1385}, 
}

@article{Li2024,
	author  = {Jinjie Li and Wenyuan Wang and Hoi-Kwong Lo},
	title   = {Fully passive measurement-device-independent quantum key distribution},
	journal = {Physical Review Applied},
	volume  = {21},
	pages   = {064056},
	year    = {2024}
}

@article{Wang2023,
	author  = {Wenyuan Wang and Rong Wang and Hoi-Kwong Lo},
	title   = {Fully-Passive Twin-Field Quantum Key Distribution},
	journal = {arXiv preprint arXiv:2304.12062},
	year    = {2023}
}

@article{Hahn2005,
	author  = {T. Hahn},
	title   = {Cuba—a library for multidimensional numerical integration},
	journal = {Computer Physics Communications},
	volume  = {168},
	number  = {2},
	pages   = {78--95},
	year    = {2005}
}

@misc{Friedberg2022,
	author  = {Richard Friedberg},
	title   = {Rodrigues, Olinde: "Des lois g\'{e}om\'{e}triques qui r\'{e}gissent les d\'{e}placements d'un syst\`{e}me solide...", translation and commentary},
	year    = {2022},
	eprint  = {2211.07787},
	archivePrefix = {arXiv},
	primaryClass  = {math.HO}
}

@article{Yoshino2018,
	author  = {Ken-ichiro Yoshino and others},
	title   = {Quantum key distribution with an efficient countermeasure against correlated intensity fluctuations in optical pulses},
	journal = {npj Quantum Information},
	volume  = {4},
	pages   = {8},
	year    = {2018}
}

@article{Jensen1906,
	author = {Jensen, J. L. W. V.},
	title = {Sur les fonctions convexes et les in\'egalit\'es entre les valeurs moyennes},
	journal = {Acta Mathematica},
	volume = {30},
	pages = {175--193},
	year = {1906}
}

@misc{Li2024CKA,
	title={Fully Passive Quantum Conference Key Agreement}, 
	author={Jinjie Li and Wenyuan Wang and H. F. Chau},
	year={2024},
	eprint={2407.15761},
	archivePrefix={arXiv},
	primaryClass={quant-ph}
}

@article{Lloyd00,
	author  = {Lloyd, Seth},
	title   = {Ultimate physical limits to computation},
	journal = {Nature},
	volume  = {406},
	pages   = {1047--1054},
	year    = {2000}
}

@article{Deffner13,
	author  = {Deffner, Sebastian and Lutz, Eric},
	title   = {Quantum speed limit for non-Markovian dynamics},
	journal = {Physical Review Letters},
	volume  = {111},
	pages   = {010402},
	year    = {2013}
}

@article{Campbell17,
	author  = {Campbell, Steve and Deffner, Sebastian},
	title   = {Trade-off between speed and cost in shortcuts to adiabaticity},
	journal = {Physical Review Letters},
	volume  = {118},
	pages   = {100601},
	year    = {2017}
}

@article{Levitin2009,
	title   = {Fundamental Limit on the Rate of Quantum Dynamics: The Unified Bound},
	author  = {Levitin, Lev B. and Toffoli, Tommaso},
	journal = {Physical Review Letters},
	volume  = {103},
	pages   = {160502},
	year    = {2009}
}

@article{Margolus98,
	author    = {N. Margolus and L. B. Levitin},
	title     = {The maximum speed of dynamical evolution},
	journal   = {Physica D},
	volume    = {120},
	pages     = {188},
	year      = {1998},
	doi       = {10.1016/S0167-2789(98)00054-2}
}

@article{Taddei13,
	author    = {M. M. Taddei and B. M. Escher and L. Davidovich and R. L. de Matos Filho},
	title     = {Quantum Speed Limit for Physical Processes},
	journal   = {Physical Review Letters},
	volume    = {110},
	pages     = {050402},
	year      = {2013},
	doi       = {10.1103/PhysRevLett.110.050402}
}

@article{Braunstein96,
	author  = {Braunstein, S. L. and Caves, C. M. and Milburn, G. J.},
	title   = {{Generalized uncertainty relations: Theory, examples, and Lorentz invariance}},
	journal = {Annals of Physics},
	volume  = {247},
	number  = {1},
	pages   = {135--173},
	year    = {1996},
	doi     = {10.1006/aphy.1996.0040}
}

@article{Caneva09,
	author  = {Caneva, T. and Murphy, M. and Calarco, T. and Fazio, R. and Montangero, S. and Giovannetti, V. and Santoro, G. E.},
	title   = {{Optimal control at the quantum speed limit}},
	journal = {Physical Review Letters},
	volume  = {103},
	number  = {24},
	pages   = {240501},
	year    = {2009},
}

@misc{Chau2025,
	title={Quantum Speed Limits For Open System Dynamics Based On Representation Basis Dependent $l_{w}^{p}$-Seminorm}, 
	author={H. F. Chau and Jinjie Li},
	year={2025},
	eprint={2508.17053},
	archivePrefix={arXiv},
	primaryClass={quant-ph}
}

@article{Boumal14,
	author  = {N. Boumal and B. Mishra and P.-A. Absil and R. Sepulchre},
	title   = {Manopt, a Matlab Toolbox for Optimization on Manifolds},
	journal = {Journal of Machine Learning Research},
	volume  = {15},
	pages   = {1455--1459},
	year    = {2014}
}

@article{CZbound,
	author  = {H. F. Chau and W. Zeng},
	title   = {A unifying quantum speed limit for time-independent Hamiltonian evolution},
	journal = {Journal of Physics A: Mathematical and Theoretical},
	volume  = {57},
	pages   = {235304},
	year    = {2024}
}

@article{Shrimali25,
	author  = {Divyansh Shrimali and Swapnil Bhowmick and Arun Kumar Pati},
	title   = {Quantum speed limit on the production of quantumness of observables},
	journal = {Physical Review A},
	volume  = {111},
	pages   = {022445},
	year    = {2025}
}

@article{Herb24,
	author  = {K. Herb and C. L. Degen},
	title   = {Quantum speed limit in quantum sensing},
	journal = {Physical Review Letters},
	volume  = {133},
	pages   = {210802},
	year    = {2024}
}

@article{Zapatero2023,
	author  = {Zapatero, V{\'i}ctor and Wang, Wenyuan and Curty, Marcos},
	title   = {A fully passive transmitter for decoy-state quantum key distribution},
	journal = {Quantum Science and Technology},
	year    = {2023},
	volume  = {8},
	number  = {2},
	pages   = {025014},
	doi     = {10.1088/2058-9565/acbc46}
}

\end{document}